\documentclass[11pt]{article}

\usepackage[all,defaultlines=3]{nowidow}

\usepackage{subcaption}

\usepackage[T1]{fontenc}
\usepackage[utf8]{inputenc}
\usepackage{lmodern}
\AtBeginDocument{\DeclareFontShape{T1}{lmr}{m}{scit}{<->ssub*lmr/m/scsl}{}}
\usepackage[letterpaper,margin=1in]{geometry}
\usepackage{microtype}

\usepackage{tikz}
\usetikzlibrary{arrows.meta,decorations.markings,positioning}
\usepackage{placeins}
\definecolor{potfive0}{HTML}{B82536}
\definecolor{potfive1}{HTML}{E98182}
\definecolor{potfive2}{HTML}{ECE7EE}
\definecolor{potfive3}{HTML}{88B9D9}
\definecolor{potfive4}{HTML}{24528A}
\definecolor{potfour0}{HTML}{B82536}
\definecolor{potfour1}{HTML}{E6A3A0}
\definecolor{potfour2}{HTML}{9AC6DD}
\definecolor{potfour3}{HTML}{24528A}
\definecolor{potink}{HTML}{273449}
\definecolor{potcycle}{HTML}{23847D}
\definecolor{potbad}{HTML}{BE7200}
\tikzset{potvertex/.style={circle,draw=white,line width=.8pt,minimum size=5.2mm,inner sep=0pt,font=\sffamily\bfseries\scriptsize},
 potedge/.style={draw=black!27,line width=.5pt},
 potpath/.style={draw=potink,line width=1.7pt,-{Stealth[length=1.8mm,width=1.5mm]},shorten >=2.7mm,shorten <=2.7mm},
 potfail/.style={draw=potbad,line width=2pt,densely dashed},
 pothalo/.style={draw=potcycle,line width=4pt,opacity=.38},
 potlabel/.style={font=\small,fill=white,inner sep=1pt,text=potink}}
\newcommand{\PotentialLegendFive}{%
\begin{tikzpicture}[baseline=-.6ex]
\foreach \j/\fg in {0/white,1/black,2/black,3/black,4/white}{
\node[potvertex,fill=potfive\j,text=\fg] at (.65*\j,0) {\j};}
\node[anchor=west,font=\small] at (3,0) {$\pmod 5$,\quad $4+1\equiv0$};
\end{tikzpicture}}

\usepackage{amsmath,amssymb,amsthm,mathtools}
\usepackage{thmtools,thm-restate}
\usepackage{enumitem}
\usepackage{xcolor}

\usepackage[
  colorlinks=true,
  linkcolor=blue!55!black,
  citecolor=green!40!black,
  urlcolor=blue!60!black,
  pdfauthor={Dibyayan Chakraborty, Sebastian Ordyniak, Sebastian Siebertz},
  pdftitle={Model Checking for Metric
Extensions of Monadic Second-Order Logic},
  pdfsubject={Parameterized algorithms and metric extensions of monadic second-order logic},
  pdfkeywords={monadic second-order logic, shortest paths, distance thresholds, bag-connected treewidth, geodesic cycles}
]{hyperref}
\usepackage[nameinlink,noabbrev]{cleveref}

\allowdisplaybreaks
\setlist[itemize]{leftmargin=2em,itemsep=0.25em,topsep=0.4em}
\setlist[enumerate]{leftmargin=2.2em,itemsep=0.25em,topsep=0.4em}

\newcommand{\CMSOGC}{\CMSO[\EqGC]}

\newcommand{\CMSO}{\mathrm{CMSO}_1}
\newcommand{\CMSOtwo}{\mathrm{CMSO}_2}
\newcommand{\CMSOCTW}{\CMSOtwo[\SP,\EqGC]}
\newcommand{\CMSOOCTW}{\CMSO[\SP,\EqGC]}
\newcommand{\CMSOSP}{\CMSO[\SP]}
\newcommand{\CMSOtwoSP}{\CMSOtwo[\SP]}

\declaretheorem[name=Theorem,numberwithin=section]{theorem}
\declaretheorem[name=Lemma,sibling=theorem]{lemma}
\declaretheorem[name=Proposition,sibling=theorem]{proposition}
\declaretheorem[name=Corollary,sibling=theorem]{corollary}

\declaretheorem[name=Definition,sibling=theorem,style=definition]{definition}

\renewcommand{\phi}{\varphi}
\newcommand{\FO}{\mathrm{FO}}
\newcommand{\MSO}{\mathrm{MSO}_1}

\newcommand{\MSOtwo}{\mathrm{MSO}_2}

\newcommand{\EqGC}{\mathsf{EqGC}}

\newcommand{\EqDist}{\mathsf{EqDist}}
\newcommand{\Dist}{\mathsf{Dist}}
\newcommand{\Sep}{\mathsf{Sep}}
\newcommand{\Cover}{\mathsf{Cover}}

\newcommand{\MSOGC}{\mathrm{MSO}_2[\mathsf{EqGC}]}

\newcommand{\MSOCTW}{\mathrm{MSO}_2[\mathsf{SP},\mathsf{EqGC}]}
\newcommand{\SP}{\mathsf{SP}}
\newcommand{\Close}{\mathsf{Close}}
\newcommand{\NearLab}{\mathsf{NearLab}}
\newcommand{\Cert}{\mathsf{CertSP}}
\newcommand{\Path}{\mathsf{Path}}

\newcommand{\Part}{\mathsf{Part}}
\newcommand{\Adm}{\mathsf{Adm}}
\newcommand{\Orient}{\mathsf{Orient}}

\newcommand{\Inc}{\mathsf{Inc}}

\newcommand{\tw}{\operatorname{tw}}
\newcommand{\cw}{\operatorname{cw}}

\newcommand{\ctw}{\operatorname{ctw}}
\newcommand{\gcw}{\operatorname{gcw}}
\newcommand{\cbl}{\operatorname{cbl}}
\newcommand{\gcl}{\operatorname{gcl}}

\newcommand{\ZZ}{\mathbb{Z}}

\newcommand{\calZ}{\mathcal{Z}}
\newcommand{\calB}{\mathcal{B}}

\newcommand{\symdiff}{\mathbin{\triangle}}
\newcommand{\set}[1]{\left\{#1\right\}}
\newcommand{\abs}[1]{\left|#1\right|}
\newcommand{\Oof}{\mathcal{O}}
\usepackage{mathrsfs}

\hypersetup{pageanchor=false}

\title{\bfseries Model Checking for Metric Extensions of \\Monadic Second-Order Logic}
\author{%
  \normalsize Dibyayan Chakraborty$^{1}$, Sebastian Ordyniak$^{1}$, and Sebastian Siebertz$^{2}$\\[0.5ex]
  \normalsize $^{1}$University of Leeds, United Kingdom\\
  \normalsize $^{2}$University of Bremen, Germany\\[0.5ex]
  \normalsize \href{mailto:D.chakraborty@leeds.ac.uk}{D.chakraborty@leeds.ac.uk},
  \href{mailto:S.Ordyniak@leeds.ac.uk}{S.Ordyniak@leeds.ac.uk}\\
  \normalsize \href{mailto:siebertz@uni-bremen.de}{siebertz@uni-bremen.de}}
\date{}
\begin{document}

\maketitle
\begingroup
\renewcommand{\small}{\normalsize}
\begin{abstract}
Monadic second-order logic yields powerful algorithmic meta-theorems, but it cannot express that a vertex set induces a shortest path.
We study metric extensions of (counting) monadic second-order logic and identify structural parameters supporting efficient model checking, optimization, and counting.

We show that a shortest-path predicate $\SP$ and a predicate $\EqGC$ comparing distances on a geodesic cycle are expressible in plain $\mathrm{MSO}_1$ on graph classes of bounded cycle-basis length.
Consequently, $\CMSOCTW$ admits fixed-parameter tractable model checking parameterized by formula length plus bag-connected treewidth.
For $\CMSOOCTW$, we obtain the analogous result for formula length plus \emph{geodesic cliquewidth}, the sum of cliquewidth and cycle-basis length.
The translations preserve free variables, yielding additive optimization and counting.
Applications include FPT algorithms for \textsc{Minimum Geodetic Set}, \textsc{Strong Metric Dimension}, and geodesic convexity and visibility problems parameterized by geodesic cliquewidth, and for \textsc{Minimum Isometric Path Partition} parameterized by bag-connected treewidth.

To handle numerical distance thresholds outside the parameter, we introduce restricted packing-covering templates combining a $\CMSOCTW$ formula with positive separation and covering requirements.
For every fixed template~$\Phi$, feasibility, additive vertex-weight optimization, and counting take time $f_\Phi(\ctw(G))N^{O_\Phi(1)}$, where $N$ is the total binary input length.
This gives FPT algorithms for \textsc{Distance-$r$ Dominating Set}, \textsc{$d$-Scattered Set}, and \textsc{Minimum Eccentricity Shortest Path} parameterized by bag-connected treewidth, and for \textsc{Minimum Consistent Subset} when the number of colors is also a parameter.

We characterize bounded bag-connected treewidth by monadic $\MSOGC$-dependence and bounded geodesic cliquewidth by monadic $\CMSOGC$-dependence.
For unguarded distance comparisons, the corresponding characterizations instead require bounded component diameter together with bounded treewidth or cliquewidth.
Our encodings yield fixed-sentence NP-hardness under suitable efficiency hypotheses.
For either unguarded predicate $R\in\{\EqDist,\Close\}$, $\FO[R]$ model checking is $\mathrm{AW}[*]$-complete already on uncolored caterpillars, and a fixed existential monadic sentence has NP-complete model checking on this class.
\end{abstract}

\endgroup

\clearpage

\medskip
\noindent\textbf{Keywords:} monadic second-order logic, shortest paths, distance profiles, bag-connected treewidth, geodesic cycles, monadic dependence, model checking, parameterized algorithms.

\tableofcontents

\clearpage
\pagenumbering{arabic} %
\hypersetup{pageanchor=true}

\section{Introduction}

\textcolor{black}{
\paragraph{Courcelle's theorem: algorithms, logic, and structure.}
Courcelle's theorem is a cornerstone of parameterized algorithmics.
It states that every fixed graph property expressible in monadic second-order logic with vertex- and edge-set quantification can be decided in linear time on graphs of bounded treewidth~\cite{Courcelle1990}.
This gives a common explanation for the tractability of a wide range of otherwise difficult problems, including \textsc{3-Colorability} and \textsc{Hamiltonian Cycle}, and makes logical definability a standard tool for proving fixed-parameter tractability.
Expressing a decision problem by a fixed formula in this language already establishes fixed-parameter tractability in treewidth, without requiring a separate dynamic program and correctness proof.
The theorem therefore both unifies existing algorithmic results and provides a systematic way to obtain new ones, including variants formed by combining constraints in the same language.}

\textcolor{black}{
More precisely, monadic second-order logic (MSO) comes in two standard forms for graphs: its vertex-set form $\mathrm{MSO}_1$ quantifies over vertices and vertex sets, while its incidence form $\mathrm{MSO}_2$ additionally quantifies over edges and edge sets.
The corresponding counting extensions $\mathrm{CMSO}_1$ and $\mathrm{CMSO}_2$ also allow predicates of the form $|X|\equiv a\pmod m$, where $a$ and $m$ are constants specified in the formula.
Courcelle's theorem applies to $\mathrm{CMSO}_2$ and gives a model-checking algorithm with running time $f(|\varphi|+\tw(G))\,|V(G)|$, for a computable function~$f$.
Thus the dependence on both the specification of the problem and the structural complexity of the graph is confined to the parameter-dependent factor.
Additive optimization and evaluation extensions further yield algorithms for finding optimum solutions and counting solutions to logically specified problems~\cite{ArnborgLagergrenSeese1991,CourcelleMosbah1993}.
The corresponding algorithmic meta-theorem for $\mathrm{CMSO}_1$ and cliquewidth is due to Courcelle, Makowsky, and Rotics~\cite{CourcelleMakowskyRotics2000}.}

\textcolor{black}{
These algorithmic results have a striking structural counterpart: treewidth and cliquewidth also admit exact characterizations through the expressive power of the corresponding logics.
Informally, a graph class is \emph{monadically dependent} for a logic if no fixed transduction in that logic can produce all finite graphs from the class, even after adding a fixed finite collection of unary predicates.
Classes of bounded treewidth are precisely the monadically $\mathrm{(C)MSO}_2$-dependent classes, and classes of bounded cliquewidth are precisely the monadically $\mathrm{CMSO}_1$-dependent classes.
The Grid Minor Theorem of Robertson and Seymour, together with an interpretation argument underlying Seese's theorem, implies that every graph class of unbounded treewidth admits a fixed $\mathrm{MSO}_2$-interpretation producing all finite graphs from monadic lifts of its incidence structures~\cite{RobertsonSeymour1986,Seese1991}.
For cliquewidth, the corresponding characterization follows from a theorem of Courcelle and Oum~\cite{CourcelleOum2007}.
The same structural restrictions that support the algorithmic meta-theorems therefore also mark the boundary beyond which the logics can encode arbitrary graphs.
We define these logical notions formally below and refer to Pilipczuk's survey~\cite{Pilipczuk2025} for further background.}

\textcolor{black}{
Our results extend both aspects of this connection to metric graph problems.
We develop meta-theorems for extensions of MSO that can express shortest paths and other metric constraints.
For the logics of our first meta-theorem, we also obtain matching monadic-dependence characterizations of bounded bag-connected treewidth and bounded geodesic cliquewidth.
\Cref{fig:logics-capturing-classes} places these characterizations at the side of the well-established correspondences between logic and graph structure.}

\begin{figure}[htbp]
\centering
\begin{subfigure}[t]{.485\textwidth}
\centering
\resizebox{\linewidth}{!}{%
\begin{tikzpicture}[
  x=1cm,y=1cm,
  every node/.style={font=\fontsize{9}{10.6}\selectfont},
  caplogic/.style={draw=black!22,fill=black!2,rounded corners=1.4pt,
    text width=2.75cm,minimum height=.90cm,align=center,inner sep=3pt},
  capclass/.style={draw=blue!35!black,fill=blue!3,rounded corners=1.4pt,
    text width=3.55cm,minimum height=.90cm,align=center,inner sep=3pt},
  capnew/.style={draw=red!55!black,fill=red!5,line width=.8pt},
  caplink/.style={Latex-Latex,line width=.55pt,draw=black!65},
  capincl/.style={-Latex,line width=.45pt,draw=black!45}
]
  \path[use as bounding box] (-1.52,-5.40) rectangle (6.12,1.03);
  \node[font=\small\bfseries] at (0,.76) {Logic};
  \node[font=\small\bfseries] at (4.20,.76) {Graph classes};
  \node[caplogic,capnew] (sl1) at (0,-.24)
    {$\mathrm{MSO}_2[\mathsf{EqGC}]$\\[-.2ex]
     {\scriptsize\color{red!55!black}this work}};
  \node[capclass,capnew] (sc1) at (4.20,-.24)
    {Bounded bag-connected\\treewidth};
  \node[caplogic] (sl2) at (0,-1.37) {$\mathrm{MSO}_2$};
  \node[capclass] (sc2) at (4.20,-1.37) {Bounded treewidth};
  \node[caplogic] (sl3) at (0,-2.50)
    {Subconnectivity logic\\[-.2ex]$\mathrm{FO}+\mathsf{subconn}$};
  \node[capclass] (sc3) at (4.20,-2.50) {Exclude a fixed minor};
  \node[caplogic] (sl4) at (0,-3.63)
    {Separator logic\\[-.2ex]$\mathrm{FO}+\mathsf{conn}$};
  \node[capclass] (sc4) at (4.20,-3.63)
    {Exclude a fixed\\topological minor};
  \node[caplogic] (sl5) at (0,-4.76) {$\mathrm{FO}$};
  \node[capclass] (sc5) at (4.20,-4.76) {Nowhere dense};
  \foreach \i in {1,...,5} {\draw[caplink] (sl\i.east) -- (sc\i.west);}
  \foreach \i/\j in {1/2,2/3,3/4,4/5} {
    \draw[capincl] (sc\i.south) -- (sc\j.north);
  }
\end{tikzpicture}%
}
\caption{The sparse world, shown for subgraph-closed classes.}
\label{fig:capturing-sparse}
\end{subfigure}\hfill
\begin{subfigure}[t]{.485\textwidth}
\centering
\resizebox{\linewidth}{!}{%
\begin{tikzpicture}[
  x=1cm,y=1cm,
  every node/.style={font=\fontsize{9}{10.6}\selectfont},
  caplogic/.style={draw=black!22,fill=black!2,rounded corners=1.4pt,
    text width=2.75cm,minimum height=.90cm,align=center,inner sep=3pt},
  capclass/.style={draw=blue!35!black,fill=blue!3,rounded corners=1.4pt,
    text width=3.55cm,minimum height=.90cm,align=center,inner sep=3pt},
  capnew/.style={draw=red!55!black,fill=red!5,line width=.8pt},
  caplink/.style={Latex-Latex,line width=.55pt,draw=black!65},
  capincl/.style={-Latex,line width=.45pt,draw=black!45}
]
  \path[use as bounding box] (-1.52,-5.40) rectangle (6.12,1.03);
  \node[font=\small\bfseries] at (0,.76) {Logic};
  \node[font=\small\bfseries] at (4.20,.76) {Graph classes};
  \node[caplogic,capnew] (dl1) at (0,-.30)
    {$\mathrm{CMSO}_1[\mathsf{EqGC}]$\\[-.2ex]
     {\scriptsize\color{red!55!black}this work}};
  \node[capclass,capnew] (dc1) at (4.20,-.30)
    {Bounded geodesic\\cliquewidth};
  \node[caplogic] (dl2) at (0,-1.77) {$\mathrm{CMSO}_1$};
  \node[capclass] (dc2) at (4.20,-1.77) {Bounded cliquewidth};
  \node[caplogic] (dl3) at (0,-3.24) {Low-rank MSO};
  \node[capclass,dashed,draw=black!55,fill=black!2] (dc3) at (4.20,-3.24)
    {Low-rank MSO dependence\\[-.1ex]
     {\fontsize{7}{8}\selectfont\itshape\mbox{Structural description open}}};
  \node[caplogic] (dl4) at (0,-4.71) {$\mathrm{FO}$};
  \node[capclass] (dc4) at (4.20,-4.71) {Monadic dependence};
  \foreach \i in {1,...,4} {\draw[caplink] (dl\i.east) -- (dc\i.west);}
  \foreach \i/\j in {1/2,2/3,3/4} {
    \draw[capincl] (dc\i.south) -- (dc\j.north);
  }
\end{tikzpicture}%
}
\caption{The dense world is much less explored. One of the major open questions is whether counting can be avoided in the capturing results.}
\label{fig:capturing-dense}
\end{subfigure}
\caption{Logics and the graph classes they capture through monadic dependence. Horizontal arrows indicate the capturing correspondences, and downward arrows indicate inclusion. }
\label{fig:logics-capturing-classes}
\end{figure}
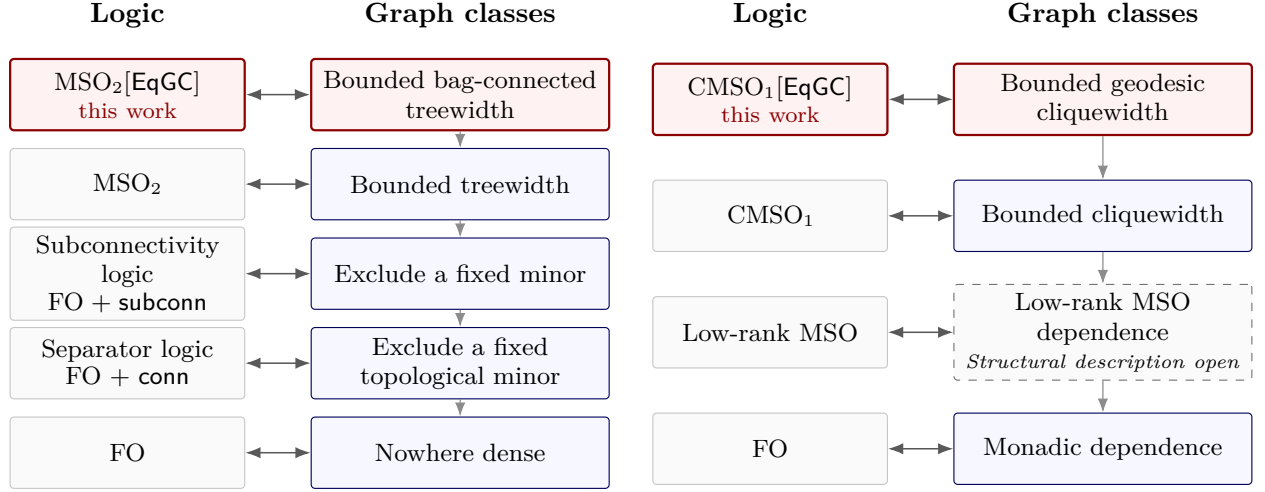

\paragraph{Metric extensions.} A shortcoming of $\mathrm{(C)MSO}$ is that, in general, it cannot express that a set of vertices induces a shortest path with a fixed formula.
Indeed, a standard Ehrenfeucht-Fra\"iss\'e argument already shows this on cycles, where no fixed $\mathrm{(C)MSO}$ formula can determine whether a designated arc is no longer than the complementary arc.
This motivates us to study extensions of monadic second-order logic and first-order logic that address this shortcoming. 

We consider the following atoms.
\begin{itemize}
	\item $\SP(P,s,t)$ expresses that the vertex set $P$ induces a shortest path from $s$ to $t$.
	\item $\EqGC(C,x,y,u,v)$ expresses that $C$ induces a geodesic cycle containing $x,y,u,v$ and that $d_{G}(x,y)=d_{G}(u,v)$. A geodesic cycle is a cycle in which the distance between any two vertices, measured along the cycle, equals their distance in the whole graph.
	\item $\Dist(u,v,d)$ expresses that $d_G(u,v)\geq d$, where $d$ is a numerical threshold.
    \item The two restricted forms $\Sep_d(X,Y)$ and $\Cover_r(U,Y)$ of $\Dist(u,v,d)$ assert, respectively, that every pair of distinct vertices $x\in X$ and $y\in Y$ has distance at least $d$, and that every vertex of $U$ lies at distance at most $r$ from some vertex of $Y$.
    \item 
    For the nearest-label atoms, we assume that the graph carries unary input predicates $L_1,\ldots,L_q$ that form a partition of $V(G)$.
We write $\NearLab_i(v;S)$ if at least one vertex of $S$
closest to $v$ has label $i$, that is,
\[
  \NearLab_i(v;S)
  \quad\Longleftrightarrow\quad
  \exists u\in S\cap L_i\colon d_G(v,u)=d_G(v,S).
\]
	\item $\EqDist(x,y,u,v)$ expresses that $d_G(x,y)=d_G(u,v)<\infty$.
	\item $\Close(x,u,v)$ expresses that $d_G(x,u)<d_G(x,v)$.
\end{itemize}

For a logic $\mathcal L$ and a collection $\mathcal F$ of these predicates, we write $\mathcal L[\mathcal F]$, for example $\mathcal L[\SP,\EqGC]$, for the corresponding extension of $\mathcal{L}$.
\medskip

Model checking provides a bridge from logical descriptions to algorithms.
Given a graph~$G$ and a sentence~$\varphi$, the question is whether $\phi$ is true in $G$, in symbols $G\models\varphi$.
For a formula $\varphi(X_1,\ldots,X_t)$ with free monadic variables, and a vertex-weighted graph $G$, the corresponding optimization problem asks for a satisfying assignment that minimizes or maximizes a given additive weight function.
Evaluation and optimization forms of Courcelle's theorem were proved by Arnborg, Lagergren, and Seese and by Courcelle and Mosbah~\cite{ArnborgLagergrenSeese1991,CourcelleMosbah1993}.
For $\mathrm{CMSO}_1$, they extend to bounded cliquewidth~\cite{CourcelleMakowskyRotics2000,CourcelleMakowskyRotics2001}.
These methods also yield counting of satisfying assignments to the free monadic variables.
Our aim in this paper is to extend these meta-theorems to our more general logics and derive new parameterized algorithms for problems expressible in the respective logics. 
Below, we give many examples of problems expressible in these extensions of $\mathrm{(C)MSO}$ that are $\mathrm W[1]$-hard when parameterized by treewidth alone, so we look for additional parameters that allow us to obtain fixed-parameter algorithms.

\paragraph{Width parameters and our first new meta-theorem.}

The \emph{bag-connected treewidth} of a graph~$G$, denoted by $\ctw(G)$ and introduced by Diestel and M\"uller~\cite{DiestelMuller2018}, is the minimum width of a tree decomposition of $G$ in which every bag induces a connected subgraph\footnote{We note that other notions of connectivity for tree decompositions have also
	Other notions of connectivity for tree decompositions have also
	been studied, for example in the work of Fraigniaud and Nisse~\cite{FraigniaudNisse2006}.
	We use the term \emph{bag-connected} to make explicit that we require each individual bag to induce a connected subgraph.}. 
Clearly, we have $\tw(G)\leq\ctw(G)$, since we impose an additional requirement on the bags, but already cycles show that the two parameters can differ substantially: while $\tw(C_n)=2$, we have $\ctw(C_n)=\lceil n/2\rceil$.
Thus, even on graphs of very small treewidth, requiring the bags to be connected may force us to use much larger bags.

In fact, large treewidth and long geodesic cycles are the only
obstructions to bounded bag-connected treewidth:
Diestel and M\"uller~\cite{DiestelMuller2018} proved that a graph class has bounded bag-connected treewidth if and only if it has bounded treewidth and there is a bound on the lengths of its geodesic cycles.
Using the refinement of Hamann and Wei{\ss}auer~\cite{HamannWeissauer2016}, we can equivalently replace the latter condition by bounded \emph{cycle-basis length}.
The cycle-basis length $\cbl(G)$ of a graph~$G$ is the minimum, over all bases of the binary cycle space of $G$, of the length of the longest cycle in the basis, with $\cbl(G)=0$ when $G$ is a forest.
We obtain our most precise algorithmic statements for the combined parameter $\tw(G)+\cbl(G)$, but, for brevity, we will mostly refer to bag-connected treewidth as the parameter in what follows.

For dense graphs we follow the same idea and combine cliquewidth with cycle-basis length, defining the
\emph{geodesic cliquewidth} of $G$ by \[\gcw(G):=\cw(G)+\cbl(G).\]

It is easy to observe (as we do below) that a graph class has bounded geodesic cliquewidth if and only if it has bounded cliquewidth and no arbitrarily long geodesic cycles, which explains our choice of terminology.

The algorithmic usefulness of bag-connected decomposition was recently demonstrated by Chakraborty et al.~\cite{ChakrabortyOrdyniakPersonalCommunication} who gave a double exponential FPT-algorithm for \textsc{Geodetic Set} parameterized by bag-connected treewidth. Given a graph $G$ and an integer $k$, the objective of \textsc{Geodetic Set} is to decide if $G$ has a set $S$ of vertices of cardinality at most $k$ and any vertex of $G$ lies in some shortest path between some pair of vertices in $S$. 
Even though Chakraborty et al.~\cite{ChakrabortyOrdyniakPersonalCommunication} proved that computing
bag-connected treewidth exactly is W[1]-hard when parameterized
by the optimal width, there is an FPT approximation algorithm is
available whose approximation ratio depends only on that width.
The hardness result contrasts with the fixed-parameter tractability of computing ordinary treewidth, but the approximation result still allows us to construct a bag-connected decomposition whose width is bounded by a function of the optimum, which is sufficient for many parameterized algorithms.

\textcolor{black}{
A strong feature of our work is that the meta-theorems arrive at an early stage in the algorithmic study of bag-connected treewidth, when only a few dedicated algorithms for this parameter are available.
The general logical framework thus precedes a broad repertoire of problem-specific algorithms and provides new fixed-parameter algorithms across several areas of metric graph theory.
For several problems, this extends the strongest previously known tractability results in terms of width and diameter.
This makes the meta-theorems a starting point for developing the algorithmic potential of bag-connected treewidth, with weighted optimization and exact counting available alongside the decision results.}

\medskip

Our first main result is an algorithmic meta-theorem for the two parameters bag-connected treewidth and geodesic cliquewidth.
For additive optimization, we assign signed binary weights to the vertices for each free vertex-set variable, and we minimize or maximize the sum of the weights over the corresponding selected sets.
When we also have free edge-set variables, we allow weights on the edges in the same way.
We therefore state the running times in terms of the total binary input length~$N$, including the weights.

\begin{restatable}{theorem}{strongtheorem}
	\label{thm:strong-intro}
	\label{thm:guarded-model-checking}\label{thm:strong}
	\label{cor:guarded-ctw-model-checking}\label{thm:main}
	\label{thm:optimization}\label{cor:mso2sp-optimization}\label{cor:guarded-optimization}
	\label{cor:guarded-counting}\label{thm:cw-metric}
	For each of the following logics, we can perform model checking, decide whether a satisfying assignment to the free monadic variables exists, compute an assignment minimizing or maximizing an additive weight objective, and count the satisfying assignments or the optimum assignments.
	For a computable function~$f$, we obtain the following running times.
	\begin{enumerate}[label=(\roman*)]
		\item For $\varphi\in\CMSOCTW$, we can solve these problems in time $f(\abs{\varphi}+\ctw(G))N^{O(1)}$.
		Thus these problems are fixed-parameter tractable parameterized by $\abs{\varphi}+\ctw(G)$, or equivalently by $\abs{\varphi}+\tw(G)+\cbl(G)$.
		\item For $\varphi\in\CMSOOCTW$, we can solve these problems in time $f(\abs{\varphi}+\gcw(G))N^{O(1)}$.
		Thus these problems are fixed-parameter tractable parameterized by $\abs{\varphi}+\gcw(G)=\abs{\varphi}+\cw(G)+\cbl(G)$.
	\end{enumerate}
\end{restatable}

When we have no binary weights, we can state the same running-time bounds with $n:=\abs{V(G)}$ instead of~$N$.
Note that this still covers cardinality optimization, where we use unit weights for the selected sets whose total size we wish to minimize or maximize.
We recall the background on the optimization and counting variants of Courcelle's Theorem in \Cref{sec:mso}.

Let us highlight three algorithmic applications of the above theorem in the context of the so-called ``distance-based covering'' problems~\cite{ChakrabortyFoucaudHakanen2025} which is an umbrella term for many algorithmic problems involving the shortest-path metric. A set $S\subseteq V(G)$ is \emph{geodetic} if every vertex lies on a shortest path between two vertices of~$S$. Given a graph $G$ and an integer $k$,
recall that the objective of \textsc{Geodetic Set} is to find a geodetic set of $G$ with  cardinality at most~$k$. 
We describe the geodetic sets by the following fixed $\MSO[\SP]$ formula, and we obtain the minimum size by minimizing $|S|$ over its satisfying assignments.
We then compare this minimum with the input bound $k$.
\begin{align*}
	\varphi_{\mathrm{geo}}(S)
	:=\forall v\,\exists s\exists t\exists P\,
	\bigl(s\in S\wedge t\in S\wedge\SP(P,s,t)\wedge v\in P\bigr).
\end{align*}

\textsc{Geodetic Set} is $\mathrm{W}[1]$-hard for the combined parameter consisting of solution size, pathwidth, and feedback vertex set number~\cite{KellerhalsKoana2022}. Recently, Tale~\cite{Tale2025} proved that \textsc{Geodetic Set} is NP-hard even at constant pathwidth and feedback vertex set number. See~\cite{KellerhalsKoana2022,Tale2025,ChakrabortyDasFoucaudGahlawatLajou2026,FoucaudEtAl2024,BergougnouxDefrainMcInerney2025,ChakrabortyGahlawatRoy2023,ChakrabortyEtAl2020GeodeticHardness} for more results on \textsc{Geodetic Set}.

\Cref{thm:strong-intro} extends the earlier FPT result for bag-connected treewidth to geodesic cliquewidth.
We also obtain the corresponding results for weighted optimization and exact counting of solutions or optimum solutions, although we do not repeat these statements now for every application.

\begin{corollary}
	\textsc{Geodetic Set} is FPT w.r.t.\  $\gcw$.
\end{corollary}

\textsc{Geodetic Set} is a representative of the broader class of ``vertex selection'' problems involving shortest paths. In Section~\ref{sec:applications}, we present more algorithmic problems that have been studied in metric and algorithmic graph theory and for which our meta-theorem provides more general or even completely new results. 

A second application of \Cref{thm:strong} is \textsc{Isometric Path Partition}, where we are given a graph $G$ and an integer $k$ and ask whether we can partition $V(G)$ into exactly $k$ shortest paths.
This problem has been studied in structural and algorithmic graph theory~\cite{ChakrabortyDefrainFoucaudMariTale2026}.
\textsc{Isometric Path Partition} is $\mathrm{W}[1]$-hard parameterized by treewidth, but FPT parameterized by treewidth plus diameter~\cite{ChakrabortyDefrainFoucaudMariTale2026}. See~\cite{ChakrabortyDefrainFoucaudMariTale2026,FernauEtAl2025PathPartitions,PenevEtAl2026IsometricPartitions,ChakrabortyEtAl2026IsometricPaths} for more results on this topic.  

\textsc{Isometric Path Partition} can be expressed with a $\MSOtwo[\SP]$ formula as follows. 
For an edge set $F$, let $\operatorname{Comp}_F(C)$ be the ordinary $\CMSOtwo$ formula saying that $C$ is a nonempty connected component of the spanning subgraph $(V(G),F)$.
Isolated vertices are included as one-vertex components.
Define
\begin{align*}
	\operatorname{IPP}(F)
	:=\forall C\,
	\bigl(\operatorname{Comp}_F(C)\rightarrow\exists s\exists t\,\SP(C,s,t)\bigr).
\end{align*}

If $\operatorname{IPP}(F)$ holds, then every connected component
of $(V(G),F)$ has a vertex set $C$ that induces a shortest path
in $G$, where we also allow singleton paths.
Since $C$ is connected in $(V(G),F)$ and every shortest path is
induced, we must include every edge of $G[C]$ in $F$.
Consequently, $(V(G),F)$ is a disjoint union of shortest paths
and has exactly $|V(G)|-|F|$ connected components.
Conversely, from any isometric path partition we obtain a
satisfying edge set $F$ by taking the union of its path edges.

We can therefore use the additive optimization part of
Theorem~\ref{thm:strong} to maximize $|F|$, which gives us a
partition into the minimum possible number of shortest paths.
A partition into exactly~$k$~paths exists if and only if this
minimum is at most $k$ and $k\leq |V(G)|$, since we can split
shortest paths into shorter shortest paths until we obtain
the desired number of parts.
This yields the following corollary.

\begin{corollary}
\textsc{Isometric Path Partition} is FPT w.r.t.\  $\ctw$.
\end{corollary}

As above, \textsc{Isometric Path Partition} is a representative of the broader class of ``shortest path selection'' problems and in Section~\ref{sec:applications}, we highlight more algorithmic problems that have been studied in algorithmic graph theory and apply Theorem~\ref{thm:strong-intro} to obtain FPT-algorithms w.r.t.   bag-connected treewidth. 

A third application of Theorem~\ref{thm:strong} concerns hitting geodesic cycles. 
Given a graph $G$ and an integer $k$, the objective of \textsc{Geodesic Cycle Transversal} is to decide if $G$ has a subset $S\subseteq V(G)$ with $|S|\leq k$ such that all geodesic cycles of $G$ contain a vertex of $S$. 
Since \textsc{Geodesic Cycle Transversal} admits a fixed formula in $\MSO[\EqGC]$, we have the following corollary.

\begin{corollary}
	\textsc{Geodesic Cycle Transversal} is FPT w.r.t.\  $\gcw$.
\end{corollary}

\paragraph{Proof techniques.}
Our approach to the shortest-path predicate builds on the work of Boulinier, Petit, and Villain~\cite{BoulinierPetitVillain2004,BoulinierPetitVillain2008}.
A path $P=v_0v_1\cdots v_r$ is shortest if and only if there is a map $\lambda\colon V(G)\to\ZZ$ whose values differ by at most one across each edge and increase by one along~$P$.
Of course, an unbounded integer-valued map cannot be quantified directly in ordinary MSO.
Here is the trick: 
For an integer $M\geq3$, we call a labeling $c\colon V(G)\to\mathbb Z_M$ \emph{admissible} if adjacent vertices receive equal or cyclically consecutive labels.
If $M>\cbl(G)$, then every such labeling can be lifted to an integer-valued potential with the same signed edge differences.
We prove this lifting statement by adapting a result of Boulinier, Petit, and Villain~\cite{BoulinierPetitVillain2004,BoulinierPetitVillain2008}.
A fixed cyclic labeling can be represented by $M$ unary sets, and hence is $\MSO$ definable. 

\input{grid-potentials}

We show this idea in \Cref{fig:potential}, where we color the vertices of a grid by their distances from~$s$ modulo~$5$.
In~(a), the labels increase by one modulo~$5$ along the shortest path, while the detour in~(b) has two steps in which they decrease.
If we change the labels along the detour so that every step increases by one, while keeping the other labels unchanged, we obtain the coloring in~(c), where three edges no longer satisfy our admissibility condition.
In~(d), we show one of these edges on a geodesic $4$-cycle, which makes the problem with this coloring easy to see.

\begin{restatable}{theorem}{shortestcertificate}
\label{thm:cert-mso}
	For every integer $M\geq3$, there exists an $\MSO$ formula $\Cert_M(P,x,y)$ of length polynomial in~$M$ that, on every graph~$G$ with $\cbl(G)<M$, is equivalent to $\SP(P,x,y)$.
\end{restatable}

Similarly, for $\EqGC$, we compare two arcs of a geodesic cycle by subtracting their signed edge contributions in a common cyclic direction.
We show that these differences sum to zero around the cycle precisely when the arc lengths agree.
Again, the existence of such an integer potential is equivalent to a finite cyclic labeling, which can be quantified by $\MSO$. We obtain the following analogous theorem. 

\begin{restatable}{theorem}{guardedcertificate}
\label{thm:eqgc-cert-mso}
	For every integer $M\geq3$, there exists an $\MSO$ formula $\mathsf{CertEqGC}_M(C,x,y,u,v)$ of length polynomial in~$M$ that, on every graph~$G$ with $\cbl(G)<M$, is equivalent to $\EqGC(C,x,y,u,v)$.
\end{restatable}

Now we can use the structural bounds $\cbl(G)\leq2\ctw(G)$ and $\cbl(G)\leq3\gcw(G)$ and replace every $\SP$ and $\EqGC$ atom by its certificate and obtain an ordinary $\CMSOtwo$ formula or an ordinary $\CMSO$ formula, respectively, whose length depends only on the original formula length and the chosen parameter.
We can then apply the ordinary CMSO meta-theorems and obtain \Cref{thm:main}.

\paragraph{Packing-covering templates and the second meta-theorem.}

Having dealt with shortest paths and guarded equidistance, we now turn to
distance bounds that are given as part of the input. 
Recall that we write
$\Dist(u,v,d)$ for the condition $d_G(u,v)\geq d$, where $d$ is a
nonnegative integer given in binary. For every fixed value of $d$, we can
express this condition in first-order logic by excluding all shorter
walks, although the length of the resulting formula depends on $d$.
Our aim here is to allow the numerical value of $d$ to vary while keeping
it outside the parameter, so expanding the atom into an ordinary formula
does not give the kind of algorithm that we are looking for.

It turns out that we cannot allow these input-valued comparisons without restricting their
use: as we show in \Cref{thm:dist-hardness-intro}, evaluating a fixed
existential-monadic formula with one input-valued distance threshold is
already NP-complete on paths with a fixed finite number of 
colors. We return to this hardness result below and for now take it as a motivation to study a 
fragment in which the distance comparisons express only separation and
covering requirements.

Two familiar problems illustrate what we would like to capture. In
\textsc{Distance-$d$-Dominating Set}, we are given a graph $G$ and
nonnegative integers $d,k$, and we ask whether there is a set $S$ of at
most $k$ vertices such that every vertex of $G$ is at distance at most
$d$ from some vertex of $S$. In \textsc{$d$-Scattered Set}, we instead
ask for a set $S$ of at least $k$ vertices whose distinct members have
pairwise distance at least $d$. 
Thus our input numbers 
serve as a covering radius in the first problem and as a separation
threshold in the second.

Recall that we use $\Sep_d(X,Y)$ to express separation between the vertex sets $X$ and $Y$, while $\Cover_r(U,Y)$ requires every vertex of the set $U$, which is given with the input, to lie within distance $r$ of the vertices of $Y$.
Fix vertex-set variables $X_1,\ldots,X_p$ and a side
formula $\psi(\bar X)\in\CMSOCTW$ over a fixed finite unary input
signature. We additionally allow input predicates $L_1,\ldots, L_q\subseteq V(G)$, which here do not have to partition the vertex set, and may overlap in an arbitrary way. 

A \emph{packing-covering template} has the form 

\begin{align*}
  \Phi(\bar X;\bar\delta,\bar\rho)
  :=\psi(\bar X)
  \wedge\bigwedge_{a\in A}\Sep_{\delta_a}(X_{i_a},X_{j_a})
  \wedge\bigwedge_{b\in B}\Cover_{\rho_b}(U_b,X_{\ell_b}).
\end{align*}
$A$ and $B$ are finite index sets, and $i_a,j_a,\ell_b\in\{1,\ldots,p\}$.
Each target set $U_b$ is either $V(G)$ or one of the unary input sets $L_1,\ldots,L_q$.

\begin{restatable}{theorem}{packingcoveringtheorem}
\label{thm:packing-covering}\label{thm:dist-intro}
For every fixed packing-covering template $\Phi$, we can decide whether a satisfying assignment exists, compute a minimum-weight or maximum-weight satisfying assignment whenever one exists, and count all satisfying assignments or all optimum assignments exactly in time $f_\Phi(\ctw(G))N^{O_\Phi(1)}$, for a computable function $f_\Phi$.
\end{restatable}

Again, when we do not have vertex weights as part of the input, we obtain the running time $f_\Phi(\ctw(G))n^{O_\Phi(1)}$.
We assume here (and below) that all numerical distance thresholds are bounded by $n$. 

Since the template is fixed, so is the number of selected sets, and
therefore also the number of cardinalities that we need to track.
The proof hence gives a uniform XP algorithm when we include the full template
description, together with the number and form of any additional
cardinality constraints, in the parameter alongside $\ctw(G)$, while
their numerical bounds remain outside the parameter. 
We show in \Cref{sec:hardness-packing-covering} that this is optimal (assuming FPT $\neq$ W[1]) by showing that deciding whether a packing-covering template has a satisfying assignment
is $\mathrm{W}[1]$-hard parameterized by $|\Phi|+\ctw(G)$, even on
unweighted trees over a fixed finite unary signature, with an ordinary
$\mathrm{MSO}_1$ side formula.

For \textsc{Distance-$d$-Dominating Set}, we use the single condition
$\Cover_d(V(G),S)$ and minimize~$|S|$, whereas for
\textsc{$d$-Scattered Set} we use $\Sep_d(S,S)$ and maximize~$|S|$.
Both templates have fixed size, and we obtain their decision versions by
comparing the optimum with the input bound $k$. This gives the following corollary.

\begin{corollary}
\textsc{Distance-$d$-Dominating Set} and \textsc{$d$-Scattered Set}
are FPT w.r.t.\  
$\ctw(G)$.
\end{corollary}

Both problems are $\mathrm{W}[1]$-hard parameterized by treewidth
alone~\cite{KatsikarelisLampisPaschos2019,KatsikarelisLampisPaschos2022},
while their known treewidth algorithms have an exponential base depending
on the distance bound, with matching lower bounds under
SETH~\cite{BorradaileLe2017,KatsikarelisLampisPaschos2022}.
We give further applications of the theorem in \Cref{sec:applications}.

\paragraph{One covering set and minimum eccentricity.}

When we have only one covering atom and no separation conditions, we
can also let the side formula vary with the input while keeping the
polynomial exponent independent of its length. We call
\[
  \Phi(\bar X;\rho)
  =\psi(\bar X)\wedge\Cover_\rho(U,S)
\]
a \emph{one-cover formula} if $\psi\in\CMSOCTW$, the target set $U$ is either $V(G)$ or one of the unary input sets $L_1,\ldots,L_q$, and $S$ is one distinguished free vertex-set variable among $\bar X$.

\begin{restatable}{theorem}{onecovertheorem}
\label{thm:one-cover}
For a one-cover formula, we can solve feasibility, additive vertex-weight
optimization, and exact counting of satisfying or optimum assignments
in time $f(|\Phi|+\ctw(G))N^{O(1)}$, where $f$ is computable and the
polynomial exponent is independent of the formula.
\end{restatable}

Again, when we do not have vertex weights (we can still do cardinality optimization) as part of the input, we obtain the running time $f(|\Phi|+\ctw(G))n^{O(1)}$.

An application is \textsc{Minimum Eccentricity Shortest Path}, where we
seek a shortest path whose maximum distance from a vertex of the graph
is as small as possible~\cite{KuceraSuchy2023}. For a given radius $r$,
we express the existence of such a path with eccentricity at most $r$ by
\begin{align*}
  \bigl(\exists s\exists t\,\SP(P,s,t)\bigr)
  \wedge\Cover_r(V(G),P).
\end{align*}
The side formula has fixed length, and \Cref{thm:one-cover} lets us
evaluate this condition with $r$ as part of the input. 
Since the condition is monotone in $r$, we can use binary search to find the minimum eccentricity, obtaining an exact FPT
algorithm parameterized by $\ctw(G)$ without including the optimum
eccentricity in the parameter.

\paragraph{Nearest labels and consistent subsets.}

We next consider the $\NearLab$ atoms and allow a formula to refer to the labels of the closest selected
vertices. Recall that in this case we require that the unary input predicates $L_1,\ldots,L_q$ partition $V(G)$. 
A \emph{single-interface nearest-label formula}
$\psi(S,\bar X)$ is a formula of $\CMSOCTW$ in which we additionally
allow atoms $\NearLab_i(v;S)$, subject to the requirement that every such
atom refers to the same distinguished free set $S$, which we require to
be nonempty. We allow arbitrary Boolean combinations and quantification
over other variables, but no quantifier binds this distinguished $S$.
The additional free sets $\bar X$ are optional, and when there are none
we simply write $\psi(S)$. In optimization and counting, we consider
assignments to all displayed free sets.

\begin{restatable}{theorem}{nearestlabeltheorem}
\label{thm:nearest-label-optimization}\label{thm:nearest-label-intro}
For a single-interface nearest-label formula $\psi(S,\bar X)$, we can
solve feasibility, additive vertex-weight optimization, and exact
counting of satisfying or optimum assignments in time
$f(|\psi|+q+\ctw(G))N^{O(1)}$, where $f$ is computable and the
polynomial exponent is independent of the formula.
\end{restatable}

And again, when we do not have vertex weights as part of the input, we obtain the running time $f(|\psi|+q+\ctw(G))n^{O(1)}$.

As an application consider the \textsc{Minimum Consistent Subset} problem. Here, we seek a smallest nonempty
set $S$ such that every vertex has a closest selected vertex of its own
color~\cite{BanikEtAl2024MCS,GanianVasilakisWietheger2026}.
We allow ties between colors, provided that at least one of the closest
selected vertices has the required color, so the feasible sets are
precisely those satisfying
\begin{align*}
  S\neq\varnothing\wedge
  \forall v\,\bigwedge_{i=1}^{q}
  \bigl(L_i(v)\longrightarrow\NearLab_i(v;S)\bigr).
\end{align*}
We obtain the unweighted problem by minimizing $|S|$, and the weighted
version by minimizing an additive selection cost. For
\textsc{Minimum Strict Consistent Subset}, we require every closest
selected vertex to have the correct color, which we express by replacing
the implication with an equivalence. Since the $L_i$ form a partition
and $S$ is nonempty, this says exactly that the set of nearest labels at
each vertex consists of its own color alone. Both formulas have length
bounded by a function of $q$, so we obtain FPT algorithms parameterized
by $q+\ctw(G)$. For \textsc{Minimum Consistent Subset}, we cannot
obtain an FPT algorithm parameterized by $\ctw(G)$ alone unless
$\mathrm{P}=\mathrm{NP}$, since the problem is NP-complete on trees
when the number of colors is unrestricted~\cite{BanikEtAl2024MCS}.

\paragraph{Proof techniques. }

We prove \Cref{thm:dist-intro,thm:one-cover,thm:nearest-label-intro} by dynamic programming over a nice tree decomposition
whose bags induce connected subgraphs, with width bounded by a
function of~$\ctw(G)$. 
The key observation is that distances change by at most one across an edge, which substantially limits the information
that we need to store on a connected bag. For a selected set~$A$, a
nonempty bag $B$, and a nonnegative integer cutoff $D$, we call
$(\min\{D,d_G(v,A)\}:v\in B\bigr)$
the \emph{truncated distance profile} of $A$ on the bag $B$. 
With a general (non-bag-connected) tree decomposition we would only have the
bound $(D+1)^{|B|}$ on the number of profiles. Instead, we choose a root
and a spanning tree of $G[B]$, and observe that the root value has only
$D+1$ possibilities, while the change along each of the remaining
$|B|-1$ tree edges belongs to $\{-1,0,1\}$. We therefore obtain at most
\[
  (D+1)3^{|B|-1}
\]
profiles. 
After capping $D$ at $n$, each distance map thus contributes only one
factor depending polynomially on $n$, while its remaining contribution
depends on the bag size. 
We then maintain further problem specific information over the dynamic program.

\paragraph{Logical limitations of the meta-theorems.}

We finally investigate the limitations of our methods.
We first take a viewpoint commonly used in contemporary model theory: a graph class is complicated if it can encode all finite graphs by a fixed logical transduction.
Intuitively, an $\mathcal L$-transduction first colors a graph by adding a fixed collection of unary predicates and then applies an interpretation whose formulas belong to $\mathcal L$.
Such an interpretation specifies its output graph by a domain formula $\nu(x)$, keeping all vertices that satisfy $\nu$ in the output, and an adjacency formula~$\eta(x,y)$, defining all the edges $uv$ that satisfy $\eta(u,v)$.
Under additional assumptions ensuring that these encodings can be constructed efficiently, this also yields algorithmic hardness results.

As a simple example, map graphs can be encoded in colored planar bipartite graphs.
If the two sides are marked by $A$ and $B$, use $\nu(x):=A(x)$ and $\eta(x,y):=x\neq y\wedge\exists z\,(B(z)\wedge E(x,z)\wedge E(y,z))$.
The output has vertex set~$A$, with two vertices adjacent when they have a common neighbor in~$B$.
The graphs obtainable in this way are exactly the map graphs~\cite{ChenGrigniPapadimitriou2002}.

For a logic $\mathcal L$, we call a graph class \emph{monadically $\mathcal L$-dependent} if no fixed $\mathcal L$-transduction produces all finite graphs from the class. 
The opposite property is \emph{monadic $\mathcal L$-independence}.
This terminology extends the model-theoretic notion going back to Baldwin and Shelah~\cite{BaldwinShelah1985}.
For the variations of $\MSO$ we use adjacency structures and vertex colors, while for the variations of $\MSOtwo$ we work with incidence structures and hence may color vertices and edges.
Note that also in the incidence encoding of a graph we evaluate the metric atoms in the original input graph.

This viewpoint has been very fruitful in finite model theory and structural graph theory.
The classical examples are the characterizations of bounded treewidth by monadic $\MSOtwo$-dependence and bounded cliquewidth by monadic $\CMSO$-dependence recalled above.
On monotone graph classes, monadic first-order dependence coincides with nowhere denseness~\cite{AdlerAdler2014} and (assuming $\mathrm{FPT}\neq\mathrm{AW}[*]$) model checking is fixed-parameter tractable exactly on these classes~\cite{grohe2017deciding}.
Also other important classes of graphs have such logical descriptions. For a logic called separator logic, the corresponding boundary on monotone classes is exclusion of a fixed topological minor, which also underlies the efficient model-checking results for separator logic, disjoint-paths logic and torso-treewidth logic~\cite{PilipczukSchirrmacherSiebertzTorunczykVigny2022,sau2026model,SchirrmacherSiebertzStamoulisThilikosVigny2024,SchirrmacherSiebertzVigny2023}.
The recently studied sub-connectivity logic has minor exclusion as its monadic dependence boundary on weakly sparse classes~\cite{MaehlmannEtAl2026SubConnectivity} and is tractable exactly on these classes.

Our guarded equidistance predicate can express and detect precisely the metric obstruction given by long geodesic cycles.
Together with the classical width characterizations, this gives an exact logical characterization of the structural assumptions in \Cref{thm:strong-intro}.

\begin{restatable}{theorem}{guardeddependencetheorem}
\label{thm:dependence-intro}\label{thm:guarded-dependence}
For every graph class $\mathscr C$, the following statements hold.
\begin{enumerate}[label=(\roman*)]
  \item The class $\mathscr C$ has bounded bag-connected treewidth if and only if it is monadically $\MSOGC$-dependent, and equivalently if and only if it is monadically $\MSOCTW$-dependent.
  \item The class $\mathscr C$ has bounded geodesic cliquewidth if and only if it is monadically $\CMSOGC$-dependent, and equivalently if and only if it is monadically $\CMSOOCTW$-dependent.
\end{enumerate}
The equivalences in~\textnormal{(i)} also hold for the counting extensions of both logics.
\end{restatable}

\textcolor{black}{
This characterization gives bag-connected treewidth an intrinsic logical meaning.
It places this parameter in the same broader programme of characterizing important graph classes via logic. 
Hence, bag-connected treewidth plays for $\MSOCTW$ the role that treewidth plays for $\MSOtwo$: the structural condition supporting the algorithmic meta-theorem is exactly the condition preventing the logic from encoding all finite graphs by a fixed transduction.
In particular, the need to control long geodesic cycles is reflected in the expressive power of the guarded metric predicate itself.
Our equivalence holds for arbitrary graph classes, without closure assumptions.
Together with the parallel characterization of geodesic cliquewidth, it connects the new algorithmic results to an exact logical description of their structural setting.}
\medskip

Shortest-path quantification alone does not give these characterizations.
Complete graphs remain monadically $\CMSOSP$-dependent, and even cycles remain monadically $\CMSOtwoSP$-dependent.
In contrast, the unguarded predicates $\EqDist$ and $\Close$ already encode arbitrary graphs on colored paths.
Their dependence boundaries therefore require stronger control of distances.

\begin{restatable}{theorem}{unguardeddependencetheorem}
\label{thm:unguarded-dependence-intro}\label{thm:md-arbitrary-classes}
Let $\mathscr C$ be a graph class and let $R\in\{\EqDist,\Close\}$.
\begin{enumerate}[label=(\roman*)]
  \item The class $\mathscr C$ is monadically $\mathrm{MSO}_2[R]$-dependent if and only if it has bounded treewidth and bounded component diameter.
  If $\mathscr C$ is hereditary, this is equivalent to bounded treedepth.
  \item The class $\mathscr C$ is monadically $\mathrm{CMSO}_1[R]$-dependent if and only if it has bounded cliquewidth and bounded component diameter.
\end{enumerate}
The characterization in~\textnormal{(i)} also holds for $\mathrm{CMSO}_2[R]$, and all statements remain true when both predicates are present.
\end{restatable}

Bounded component diameter means that one constant bounds the diameter of every connected component of every graph in the class.
For any fixed bound, both predicates are definable by fixed first-order formulas.
This makes their CMSO extensions tractable with the parameters formula length, treewidth or cliquewidth, and component diameter.

These results add new important entries in the dependence map for logics and important graph classes. 
We refer to the excellent survey~\cite{Pilipczuk2025} for more background. 

\paragraph{Algorithmic hardness consequences.}

Logical independence alone does in general not imply computational hardness.
An algorithmic reduction needs polynomial-size host graphs and encodings that can be constructed in polynomial time.
We make this requirement explicit and reuse the logical independence interpretations whenever this is possible. We refer to \Cref{def:efficient-ctw-encoding} for the precise definition of efficient encoding. 

\begin{theorem}\label{thm:hardness-intro}
Let $\mathscr C$ admit efficient encodings of its bag-connected-treewidth obstructions.
Then model checking $\MSOGC$ on monadic lifts of incidence structures of graphs from $\mathscr C$, using a fixed finite family of unary predicates, is NP-hard for one fixed sentence.
It is therefore para-NP-hard parameterized by formula length.
If the unary predicates admit a polynomial color-elimination encoding inside $\mathscr C$ with a fixed metric-correct backward translation, the same conclusion holds on the uncolored class~$\mathscr C$.
\end{theorem}

Polynomial-time constructions of large grid-minor models handle the treewidth obstruction, while constructions of long geodesic cycles handle the metric obstruction.
In particular, model checking for a fixed existential $\MSO[\EqGC]$ sentence is NP-complete on cycles with three unary colors.

For unguarded comparisons, our logical hardness proof yields the following. 
For either $R\in\{\EqDist,\Close\}$, model checking $\FO[R]$ is $\mathrm{AW}[*]$-complete parameterized by formula length already on uncolored caterpillars.
The existential first-order fragments are $\mathrm{W}[1]$-complete there, and for either predicate there is a fixed existential monadic sentence whose model-checking problem is NP-complete.

The input-valued numerical thresholds require one additional construction, encoding a Turing-machine computation in a grid laid out along a colored path.
We obtain the following result. 

\begin{restatable}{theorem}{thmdisthardnessintro}
\label{thm:dist-hardness-intro}
There is a fixed finite unary signature $\sigma$ for which the following statements hold on $\sigma$-colored paths.
\begin{enumerate}[label=(\roman*)]
  \item There is a fixed $\mathrm{EMSO}_1[\Dist]$ formula with one numerical parameter symbol~$\delta$ whose evaluation is NP-complete when the value of~$\delta$ is given in binary as part of the input.
  \item There is a fixed $\mathrm{EMSO}_1[\Close]$ sentence whose model-checking problem is NP-complete.
\end{enumerate}
\end{restatable}

\paragraph*{Organization.}

After the preliminaries, we prove the shortest-path certificate theorem, i.e., \Cref{thm:cert-mso}, in \Cref{sec:guarded-model-checking}, and the guarded-equidistance certificate theorem, \Cref{thm:eqgc-cert-mso}, in \Cref{sec:eqgc-in-mso}.
We use these certificates in \Cref{sec:first-meta-theorem} to prove the first algorithmic meta-theorem, \Cref{thm:strong-intro}.
In \Cref{sec:dist-meta}, we give a joint dynamic program proving the packing-covering, one-cover, and nearest-label results, \Cref{thm:dist-intro,thm:one-cover,thm:nearest-label-intro}.
In \Cref{sec:encoding-hardness}, we develop the encoding constructions and prove the dependence characterizations, \Cref{thm:dependence-intro,thm:unguarded-dependence-intro}, together with the model-checking hardness results, \Cref{thm:hardness-intro}, and other theorems with similar flavour.
Finally, \Cref{sec:applications} gives further algorithmic applications.

\section{Preliminaries}

\subsection{Graphs, paths, and distances}

For a positive integer $q$, we write $[q]:=\set{1,\ldots,q}$.
All graphs we consider are finite, simple, and undirected.
We allow vertex weights encoded in binary, but these are only relevant for weighted optimization and not for the definition of distances.
We write $V(G)$ for the vertices and $E(G)$ for the edges of a graph~$G$, and put $n:=\abs{V(G)}$ and $m:=\abs{E(G)}$.

A \emph{walk} in $G$ is a sequence $W=v_0v_1\cdots v_\ell$, where $\ell\geq0$ and $v_{i-1}v_i\in E(G)$ for every $i\in\{1,\ldots,\ell\}$.
We allow vertices and edges to occur repeatedly in a walk.
We call a walk a \emph{path} if all its vertices are distinct, and we call a path \emph{induced} if no two nonconsecutive vertices are adjacent in~$G$.
We call a walk $W=v_0v_1\cdots v_\ell$ \emph{closed} if $v_\ell=v_0$ and it is a \emph{cycle} if $\ell\geq3$, $v_\ell=v_0$, and the vertices $v_0,\ldots,v_{\ell-1}$ are pairwise distinct. 
The \emph{length} of a walk is its number of edge traversals, so the walk~$W$ above has length~$\ell$, while for a path or cycle the length is simply its number of edges.
If $W$ runs from $x$ to~$y$, then $W^{-1}$ denotes the walk obtained by traversing $W$ in reverse order from $y$ to~$x$.
If $W_1$ ends where $W_2$ begins, then $W_1W_2$ denotes their concatenation, where we write the common endpoint only once.

For a vertex set $P\subseteq V(G)$, we write $G[P]$ for the subgraph induced by~$P$.
We also call the set~$P$ a path if the induced subgraph $G[P]$ is a path.
More precisely, we say that $P$ is an \emph{$s$-$t$ path set} if either
\begin{enumerate}[label=(\roman*)]
  \item $s=t$ and $P=\set{s}$, or
  \item $s\neq t$, the graph $G[P]$ is a path, and its two endpoints are exactly $s$ and~$t$.
\end{enumerate}
Thus, whenever we describe a path by its vertex set, we require that set to induce the path.

For vertices $u,v$ in the same connected component, we write $d_G(u,v)$ for the length of a shortest $u$-$v$ path, and if they lie in different components, we put $d_G(u,v)=\infty$.
For a vertex $v$ and sets $A,B\subseteq V(G)$, we use
\begin{align*}
  d_G(v,A)&:=\min_{a\in A}d_G(v,a),
  \qquad \text{ and } \qquad
  d_G(A,B):=\min_{a\in A,\,b\in B}d_G(a,b),
\end{align*}
where we take the minimum of an empty set to be~$\infty$.
We call an $s$-$t$ path set $P$ \emph{shortest} if its length equals~$d_G(s,t)$.

We call a cycle $C$ \emph{geodesic}, or \emph{isometric}, if $d_C(x,y)=d_G(x,y)$ for all $x,y\in V(C)$, where~$d_C$ denotes distance measured along the cycle.
Equivalently, for every pair of distinct vertices \mbox{$x,y\in V(C)$}, at least one of the two $x$-$y$ arcs of $C$ is a shortest $x$-$y$ path in~$G$.

\subsection{The binary cycle space}

We identify every edge set $F\subseteq E(G)$ with its characteristic vector $\chi_F\in\mathbb F_2^{E(G)}$, where $\chi_F(e)=1$ if $e\in F$ and $\chi_F(e)=0$ otherwise.
Under this identification, vector addition is symmetric difference, so $\chi_F+\chi_{F'}=\chi_{F\symdiff F'}$.
The \emph{binary cycle space} $\calZ_2(G)$ is the vector subspace generated by the characteristic vectors of the simple cycles of~$G$.
Equivalently, it consists of the edge sets in which every vertex has even degree.
A \emph{cycle basis} is a collection of simple cycles whose characteristic vectors form a basis of~$\calZ_2(G)$.
We refer to Diestel for these standard facts about the cycle space \cite[\S1.9]{Diestel2017}.

We define the cycle-basis length $\cbl(G)$ of a graph $G$ as follows. 
If $G$ is a forest, set $\cbl(G)=0$.
Otherwise define $\cbl(G) := \min_{\calB}\max_{C\in\calB}\abs{E(C)}$,
where the minimum ranges over all binary cycle bases~$\calB$ of~$G$.
Equivalently, $\cbl(G)$ is the least nonnegative integer $L$ such that the cycles of length at most~$L$ generate~$\calZ_2(G)$.
Boulinier, Petit, and Villain use the related \emph{cyclomatic characteristic}, which they define through signed cycle bases and set equal to~$2$ for forests~\cite{BoulinierPetitVillain2004,BoulinierPetitVillain2008}.
Here we use binary cycle bases and set $\cbl(G)=0$ for forests.
We can still use their intrinsic-delay argument, since every binary cycle basis gives a basis of the signed cycle space~\cite[Lemma~2.4]{KavithaEtAl2009CycleBases}.
The problem of minimizing the longest cycle in a basis is called the \emph{shortest maximal cycle basis} problem by Chickering, Geiger, and Heckerman \cite{ChickeringGeigerHeckerman1995}.

A \emph{minimum cycle basis} is a cycle basis minimizing the sum of the lengths of its cycles.
Horton proved that a minimum cycle basis can be computed in polynomial time \cite{Horton1987}.
The following fact identifies its longest cycle with our parameter.

\begin{lemma}[{\cite[Theorem~4]{ChickeringGeigerHeckerman1995}}]\label{lem:mcb-bottleneck}
Let $\calB_{\min}$ be a minimum cycle basis of~$G$.
Set $L=0$ when $G$ is a forest, and otherwise set $L:=\max_{C\in\calB_{\min}}\abs{E(C)}$.
Then $L=\cbl(G)$.
\end{lemma}

Thus $\cbl(G)$ can be computed in polynomial time.
In particular, from a minimum cycle basis we can compute $M:=\max\set{3,L+1}$, which satisfies $M>\cbl(G)$.

\providecommand{\gcw}{\operatorname{gcw}}

\subsection{Treewidth and bag-connected treewidth}
\label{sec:bag-connected-treewidth}

A \emph{tree decomposition} of $G$ is a pair $(T,(B_t)_{t\in V(T)})$, where $T$ is a tree and the bags $B_t\subseteq V(G)$ satisfy the following conditions.
\begin{enumerate}[label=(\roman*)]
  \item Every vertex of $G$ belongs to at least one bag.
  \item Every edge of $G$ has both endpoints in some bag.
  \item For every vertex $v\in V(G)$, the nodes $t$ with $v\in B_t$ induce a connected subtree of $T$.
\end{enumerate}
We define the width of the decomposition as $\max_t\abs{B_t}-1$, and the \emph{treewidth} $\tw(G)$ is the minimum width of a tree decomposition of $G$.
We call a tree decomposition \emph{bag-connected} if every nonempty bag induces a connected subgraph of $G$, and we write $\ctw(G)$ for the minimum width of such a decomposition.
This parameter is called connected treewidth by Diestel and M\"uller~\cite{DiestelMuller2018}, while we use the term \emph{bag-connected treewidth} to make the connectivity requirement explicit.
For disconnected graphs, we obtain both parameters by taking the maximum over the connected components.
Since a bag-connected tree decomposition is also a tree decomposition, we have
$\tw(G)\leq\ctw(G)$.

We write $\gcl(G)$ for the maximum length of a geodesic cycle in $G$, with $\gcl(G)=0$ when $G$ is a forest.
Every cycle is a binary sum of geodesic cycles of no greater length~\cite[Theorem~3.1]{GeorgakopoulosSprussel2009}, so these cycles generate the cycle space and give the first inequality below.
For the second inequality, we use the bound of Diestel and M\"uller~\cite[Lemma~7.1]{DiestelMuller2018}, which shows that a geodesic cycle of length $r$ forces bag-connected treewidth at least $\lceil r/2\rceil$:
\begin{equation}
  \cbl(G)\leq\gcl(G)\leq2\ctw(G).
  \label{eq:cycle-length-bounds}
\end{equation}
In the other direction, Hamann and Wei{\ss}auer proved that, whenever $G$ contains a cycle, we have
\begin{equation}
  \ctw(G)\leq\tw(G)\bigl(\cbl(G)-2\bigr),
  \label{eq:hamann-weissauer}
\end{equation}
while forests have bag-connected treewidth at most one~\cite[Theorem~2]{HamannWeissauer2016}.
Together, these inequalities show that a graph class has bounded bag-connected treewidth if and only if it has bounded treewidth and bounded cycle-basis length, or equivalently bounded treewidth and bounded geodesic-cycle length.
They also give $M\leq\max\set{3,2\ctw(G)+1}$ for the number $M$ chosen above.

\subsection{Nice bag-connected tree decompositions}
\label{sec:nice-bag-connected-decompositions}

For our dynamic programs, we need a rooted decomposition in which we can introduce and forget vertices one at a time, while every nonempty bag remains connected.
We use the following version of a nice tree decomposition, allowing separate nodes at which we process the edges of the graph.

\begin{definition}[Nice bag-connected tree decomposition]
\label{def:nice-bag-connected-decomposition}
We call a bag-connected tree decomposition $(T,(B_t)_{t\in V(T)})$ \emph{nice} if $T$ is rooted, the root and all leaves have empty bags, and every non-leaf node $t$ has one of the following forms.
\begin{enumerate}[label=(\roman*)]
  \item At an \emph{introduce-vertex node}, we have one child $t'$ and $B_t=B_{t'}\cup\{v\}$ for some $v\notin B_{t'}$.
  \item At a \emph{forget-vertex node}, we have one child $t'$ and $B_t=B_{t'}\setminus\{v\}$ for some $v\in B_{t'}$.
  \item At an \emph{introduce-edge node}, we have one child $t'$ with $B_t=B_{t'}$, and we assign to the node an edge $uv\in E(G)$ with both endpoints in $B_t$.
  \item At a \emph{join node}, we have two children $t_1,t_2$ and $B_t=B_{t_1}=B_{t_2}$.
\end{enumerate}
We require every edge of $G$ to be assigned to exactly one introduce-edge node.
When we require a nonempty bag to be connected, we always refer to the subgraph of $G$ induced by that bag, including edges that have not yet been introduced.
\end{definition}

When we process the decomposition from the leaves towards the root, a vertex may be introduced in several branches, but it is forgotten exactly once, and all its incident edges have been introduced before that event.
We can therefore check the local conditions involving a vertex and charge its selection weights when we forget it, since at that point we have already processed all the information supplied by its incident edges.
A nice bag-connected decomposition of approximately optimal width can be efficiently computed, as shown in~\cite{ChakrabortyOrdyniakPersonalCommunication}. 

\begin{theorem}[\cite{ChakrabortyOrdyniakPersonalCommunication}]
\label{thm:compute-nice-bag-connected-decomposition}
Given a graph $G$ with $n$ vertices and $m$ edges, we can compute a nice bag-connected tree decomposition of width $\Oof(c^2)$ in time $2^{\Oof(c^2)}n^{\Oof(1)}$, where $c:=\ctw(G)$.
\end{theorem}

Thus we can construct the decompositions needed by our dynamic programs within the stated parameterization, without requiring a decomposition as part of the input.

\subsection{Cliquewidth and geodesic cliquewidth}
\label{sec:geodesic-cliquewidth}
\label{sec:connected-cliquewidth}

We introduce \emph{geodesic cliquewidth} as one of our main structural parameters, combining cliquewidth with cycle-basis length.
We first recall the definition of cliquewidth, and then explain how the resulting parameter relates to the length of geodesic cycles and to the connected decompositions used in our dynamic programs.

With a $k$-expression, we construct a vertex-labeled graph using labels from $[k]$ and four operations: we create a labeled vertex, take the disjoint union of two labeled graphs, add all edges between the vertices of two distinct labels, or change every occurrence of one label to another.
The \emph{cliquewidth} $\cw(G)$ is the least $k$ for which we can construct $G$ in this way, where we disregard the labels at the end of the construction.
The labels used in an expression serve only to describe the construction, and we distinguish them from the unary input labels that occur in our formulas.

\begin{definition}[Geodesic cliquewidth]
\label{def:geodesic-cliquewidth}
\label{def:connected-cliquewidth}
We define the \emph{geodesic cliquewidth} of a graph $G$ by
\[
  \gcw(G):=\cw(G)+\cbl(G).
\]
\end{definition}

For example, 
for a forest, the cycle space is trivial and hence $\gcw(G)=\cw(G)$, while a complete graph on at least three vertices has geodesic cliquewidth five, since its cliquewidth is two and its cycle space is generated by triangles.
In particular, bounded geodesic cliquewidth still allows dense graphs of arbitrarily large order.
The number $M:=\max\set{3,\cbl(G)+1}$ chosen above satisfies $M\leq\max\set{3,\gcw(G)+1}$, so we can bound it directly in terms of our main parameter.

\subsection{Geodesic cliquewidth and geodesic cycles}
\label{sec:geodesic-cliquewidth-comparison}

The definition of geodesic cliquewidth uses cycle-basis length, since this is what we need for our potential arguments.
Nevertheless, once we also bound cliquewidth, we can equivalently require a bound on the length of every geodesic cycle, as we show next.

\begin{lemma}
\label{lem:geodesic-cliquewidth-comparison}
\label{thm:geodesic-cliquewidth-characterization}
\label{lem:connected-cliquewidth-comparison}
\label{thm:connected-cliquewidth-characterization}
For every graph $G$, we have
\[
  \gcw(G)\leq\cw(G)+\gcl(G)
  =\Oof\bigl(\gcw(G)^2\bigr).
\]
Consequently, a graph class has bounded geodesic cliquewidth if and only if it has bounded cliquewidth and bounded geodesic-cycle length.
\end{lemma}

\begin{proof}
The first inequality follows from $\cbl(G)\leq\gcl(G)$, so we only need to bound $\gcl(G)$ in terms of cliquewidth and cycle-basis length.
For forests the assertion is immediate, and we may therefore assume that $G$ contains a cycle, writing $k:=\cw(G)\geq2$ and $L:=\cbl(G)\geq3$.
By Distel's theorem~\cite[Theorem~3]{Distel2026}, we can choose a partition $\mathcal P$ of $V(G)$ such that each part is contained in the closed neighborhood of some vertex and the quotient $H:=G/\mathcal P$ has treewidth at most $k-1$.
Two distinct parts are adjacent in $H$ precisely when an edge of $G$ joins them, and the vertex whose closed neighborhood contains a part need not itself belong to that part.

We first show that $\cbl(H)\leq L$ by considering the projection of edges of $G$ to edges of $H$, where we discard edges within a part and count repeated quotient edges modulo two.
This projection maps the binary cycle space of $G$ into that of $H$, and in fact every cycle of $H$ occurs as the projection of a closed walk in $G$.
To see this, we choose an edge of $G$ realizing each edge of the quotient cycle, and within each part we join the endpoints of successive realizing edges through a vertex whose closed neighborhood contains the part.
Each of these connecting walks has at most two edges and projects either to no edge or to the same quotient edge twice, so it contributes zero over $\mathbb F_2$.
The resulting closed walk therefore projects to the original quotient cycle, which proves that the map on cycle spaces is surjective.
We can now project a cycle basis of $G$ whose cycles have length at most~$L$ and decompose each nonempty image into cycles, obtaining a generating set for the cycle space of~$H$ in which every cycle has length at most $L$.

By \eqref{eq:hamann-weissauer}, we can choose a bag-connected tree decomposition of $H$ of width at most $w:=(k-1)(L-2)$.
This remains true if $H$ is a forest, since in that case its bag-connected treewidth is at most one and we have $w\geq1$.
We replace each bag of this decomposition by the union of the parts that it contains, which gives a tree decomposition of $G$.
If two vertices belong to the same expanded bag, their parts are joined within the corresponding quotient bag by a path of at most $w$ edges, because that bag is connected and has at most $w+1$ vertices.
We lift this path to a walk in~$G$ by using one realizing edge for each quotient edge and at most two edges within each visited part, so the two prescribed vertices have distance at most $3w+2$ in $G$.

Now let $C$ be a geodesic cycle of length $r$, and intersect every expanded bag with $V(C)$ to obtain a tree decomposition of $C$.
Since $C$ is isometric, any two vertices in one of these bags have distance at most $3w+2$ along the cycle as well.
On the other hand, every tree decomposition of an $r$-cycle has a bag containing two vertices whose distance along the cycle is at least $r/3$~\cite[Lemma~4]{DourisboureGavoille2007}.
We therefore obtain
$r\leq3(3w+2)=9(k-1)(L-2)+6$,
and hence $\cw(G)+\gcl(G)=\Oof(k(L+1))=\Oof(\gcw(G)^2)$, as required.
\end{proof}

\subsection{Counting monadic second-order logic}
\label{sec:mso}

We first recall the underlying logics.
In first-order logic $\FO$, we quantify over vertices and use equality, adjacency $E(x,y)$, and the unary input predicates.
In monadic second-order logic~$\MSO$ we may additionally quantify over vertex sets and use membership atoms $x\in X$.

To work with edge variables and edge-set variables as well, we use the incidence structure~$\mathcal I(G)$, whose universe is $V(G)\mathbin{\dot\cup}E(G)$.
We distinguish the two sorts by unary predicates and use the relation $\Inc(e,v)$ to say that the edge $e$ is incident with the vertex $v$.
The logic $\MSOtwo$ is monadic second-order logic on this structure, or equivalently the logic in which we may quantify over vertices, edges, vertex sets, and edge sets.
We can express adjacency by requiring two vertices to be distinct and to have an edge incident with both of them.
The Gaifman graph of $\mathcal I(G)$ is the one-subdivision of $G$, and therefore has treewidth at most $\max\set{\tw(G),2}$.

For counting monadic second-order logic, we also allow atoms $\operatorname{card}_{a,\ell}(X)$ expressing that $\abs{X}\equiv a\pmod\ell$, where $\ell\geq2$ and $0\leq a<\ell$.
Here $a$ and $\ell$ are constants specified in the formula, and their encodings contribute to its length.
We write $\CMSO$ and $\CMSOtwo$ for the counting extensions of $\MSO$ and $\MSOtwo$, respectively, and in the latter logic we may apply these atoms to both vertex sets and edge sets.
We use these modulo tests to express conditions on individual sets, while the algorithmic counting problem considered below asks for the exact number of satisfying assignments.

The unary input predicates $L_1,\ldots,L_q$ describe given subsets of $V(G)$, and we distinguish them from the set variables whose values we choose when evaluating a formula.
These input sets may overlap unless we explicitly require a partition, as we do for the nearest-label atoms.
When we work with incidence structures, we may also have unary input predicates on edge elements.
In every case, the interpretations of the input predicates are fixed by the instance and are not included among the choices that we optimize or count.

We use the model-checking, optimization, and counting forms of Courcelle's theorem for these logics~\cite[Chapter~6]{CourcelleEngelfriet2012}.
Let $\varphi(X_1,\ldots,X_r)$ be a formula whose free variables are exactly the displayed set variables.
For each $i$, we put $D_i(G):=V(G)$ if $X_i$ is a vertex-set variable and $D_i(G):=E(G)$ if it is an edge-set variable, so for $\CMSO$ every domain $D_i(G)$ is $V(G)$.
We write $\operatorname{Sol}_{\varphi}(G)$ for the following set of satisfying assignments:
\[
  \operatorname{Sol}_{\varphi}(G)
  :=
  \bigl\{
    (S_1,\ldots,S_r)\in\textstyle\prod_{i=1}^{r}2^{D_i(G)}
    :G\models\varphi(S_1,\ldots,S_r)
  \bigr\}.
\]

\begin{definition}[Additive optimization]
\label{def:mso-optimization}
For additive optimization, we are given a graph $G$, a formula $\varphi(X_1,\ldots,X_r)$, and a weight function $w_i\colon D_i(G)\to\mathbb Z$ for each free set variable $X_i$, with all weights encoded in signed binary notation.
We are looking for a satisfying assignment that minimizes or maximizes the additive objective
\[
  w(S_1,\ldots,S_r)
  :=\sum_{i=1}^{r}\sum_{a\in S_i}w_i(a).
\]
We return the optimum value together with an assignment attaining it, or report that no satisfying assignment exists.
These weights concern only the objective, and we continue to measure graph distances by the number of edges.
\end{definition}

\begin{definition}[Counting]
\label{def:mso-counting}
For counting, we are given a graph $G$ and a formula $\varphi(X_1,\ldots,X_r)$, and we compute the number $\#\varphi(G):=\abs{\operatorname{Sol}_{\varphi}(G)}$, with the output written in binary.
If weights and a minimization or maximization objective are also given, we may instead count the satisfying assignments attaining the optimum.
In either case, we return zero when no satisfying assignment exists.
\end{definition}

In the following statements, we write $N$ for the total binary input length, including the formula, the interpretations of the input predicates, and any weights, and we include the cost of arithmetic and of writing the output in the running time.

\begin{theorem}[Courcelle's theorem for treewidth {\cite{Courcelle1992,ArnborgLagergrenSeese1991,CourcelleMosbah1993,Bodlaender1996,CourcelleEngelfriet2012}}]
\label{thm:courcelle-treewidth}
There are a computable function $f$ and a constant $c$ such that we can perform each of the following tasks in time $f(\abs{\varphi}+\tw(G))N^c$.
\begin{enumerate}
  \item Given a graph $G$ and a $\CMSOtwo$ sentence $\varphi$, we decide whether $G\models\varphi$.
  \item Given a graph $G$, a $\CMSOtwo$ formula $\varphi$ whose free variables are vertex-set or edge-set variables, and signed binary weights, we compute a minimum-weight or maximum-weight satisfying assignment and its weight, or report that no satisfying assignment exists.
  \item Given a graph $G$ and a $\CMSOtwo$ formula $\varphi$ whose free variables are vertex-set or edge-set variables, we count all satisfying assignments exactly, and when weights and an optimization objective are also given, we may instead count all optimum assignments exactly.
\end{enumerate}
\end{theorem}

\begin{theorem}[Courcelle's theorem for cliquewidth {\cite{CourcelleMakowskyRotics2000,OumSeymour2006,CourcelleEngelfriet2012}}]
\label{thm:courcelle-cliquewidth}
There are a computable function~$f$ and a constant $c$ such that we can perform each of the following tasks in time \mbox{$f(\abs{\varphi}+\cw(G))N^c$}.
\begin{enumerate}
  \item Given a graph $G$ and a $\CMSO$ sentence $\varphi$, we decide whether $G\models\varphi$.
  \item Given a graph $G$, a $\CMSO$ formula $\varphi$ whose free variables are vertex-set variables, and signed binary weights, we compute a minimum-weight or maximum-weight satisfying assignment and its weight, or report that no satisfying assignment exists.
  \item Given a graph $G$ and a $\CMSO$ formula $\varphi$ whose free variables are vertex-set variables, we count all satisfying assignments exactly, and when weights and an optimization objective are also given, we may instead count all optimum assignments exactly.
\end{enumerate}
\end{theorem}

For both width parameters, we can construct the decomposition or expression needed by the algorithm within the stated running time.
For treewidth, we use the standard fixed-parameter algorithms for finding a decomposition of bounded width~\cite{Bodlaender1996,KorhonenLokshtanov2024}.
For cliquewidth, we use the algorithm of Oum and Seymour to find a rank-decomposition of width bounded by a function of~$\cw(G)$, and then convert it into a cliquewidth expression whose width is bounded by a function of the same parameter~\cite{OumSeymour2006}.

\subsection{Extensions of monadic second-order logic}
\label{sec:metric-logics}

We now give the formal conventions for all our additional atoms.
We always measure distances in the original graph $G$, also when we interpret the surrounding formula on the incidence structure~$\mathcal I(G)$.

\begin{definition}[Shortest-path atom]
\label{def:sp-atom}
For a vertex-set variable $P$ and vertex variables $s,t$, we interpret $\SP(P,s,t)$ as saying that $P$ is a shortest $s$-$t$ path set.
Thus $G[P]$ must be a path with endpoints $s,t$ and length $d_G(s,t)$, where for $s=t$ we require $P=\set{s}$.
\end{definition}

\begin{definition}[Guarded equidistance]
\label{def:eqgc}
For a vertex-set variable $C$ and vertex variables $x,y,u,v$, we interpret $\EqGC(C,x,y,u,v)$ as saying that $G[C]$ is a geodesic cycle, that all four vertices belong to $C$, and that $d_G(x,y)=d_G(u,v)$.
\end{definition}
Since the cycle is isometric, we obtain the same equality of distances on $G[C]$.
We include the requirement that $C$ induces a geodesic cycle in the meaning of the atom itself, so the atom is false whenever this requirement fails, in particular on forests.

\begin{definition}[Distance thresholds, separation, and covering]
\label{def:dist-atom}
For the input-valued distance bounds, we use a finite set $\Delta$ of numerical parameter symbols together with a valuation $\eta\colon\Delta\to\mathbb Z_{\geq0}$ whose values are given in binary as part of the input.
For vertex variables $u,v$ and a numerical symbol $\delta\in\Delta$, we interpret the distance atom by
\[
  \Dist(u,v,\delta)
  \quad\Longleftrightarrow\quad
  d_G(u,v)\geq\eta(\delta).
\]
By a \emph{template}, we mean a formula in which we use these numerical parameter symbols as placeholders for the distance bounds.
For each input instance, we assign values to the symbols through the valuation $\eta$, while the formula itself remains unchanged.
Thus, when we measure the template length, we count its logical syntax, including the occurrences of the numerical symbols, but we do not include the binary encodings of their values, which are supplied separately as part of the input.
For example, the template $\Dist(u,v,\delta)$ is the same whether we assign the value $3$ or the value $1000$ to $\delta$.
When we display a concrete nonnegative threshold $d$, we write $\Dist(u,v,d)$.

For vertex-set arguments $X,Y$, we interpret $\Sep_\delta(X,Y)$ as requiring $d_G(x,y)\geq\eta(\delta)$ for every $x\in X$ and $y\in Y$ with $x\neq y$.
We omit equal vertices even when $X$ and $Y$ are different variables, so the atom does not require the sets to be disjoint, and $\Sep_\delta(X,X)$ expresses pairwise separation of the distinct vertices of $X$.
For a concrete threshold $d$, we use the notation $\Sep_d(X,Y)$.

For vertex-set arguments $U,Y$ and a numerical symbol $\rho\in\Delta$, we interpret $\Cover_\rho(U,Y)$ as requiring that every $u\in U$ has a vertex $y\in Y$ with $d_G(u,y)\leq\eta(\rho)$.
For a concrete radius~$r$, we write $\Cover_r(U,Y)$ and have
$\Cover_r(U,Y) \Leftrightarrow 
  \forall u\in U\,\exists y\in Y\,\neg\Dist(u,y,r+1)$.
In the packing-covering and one-cover fragments, we take the target set $U$ to be $V(G)$ or one of the unary input label sets $L_1,\ldots, L_q$. 

\end{definition}

\begin{definition}[Nearest-label atoms]
\label{def:nearest-label-atoms}
For the nearest-label atoms, we assume that the unary input label sets $L_1,\ldots,L_q$ form a partition of $V(G)$, as in the introduction.
For a vertex $v$, a vertex set $S$, and an index $i\in[q]$, we use
\[
  \NearLab_i(v;S)
  \quad\Longleftrightarrow\quad
  \exists u\in S\cap L_i\colon d_G(v,u)=d_G(v,S).
\]

That is, we ask whether at least one closest selected vertex has label $i$, and when several vertices with different labels attain the same minimum distance, we record all their labels.
We interpret every such atom as false when $S=\varnothing$.
\end{definition}

\begin{definition}[Unguarded distance comparisons]
\label{def:unguarded-distance-atoms}
For vertex variables $x,y,u,v$, we interpret $\EqDist(x,y,u,v)$ as requiring $d_G(x,y)=d_G(u,v)<\infty$.
For $\Close(x,u,v)$, we instead require $d_G(x,u)<d_G(x,v)$, using the usual order on $\mathbb Z_{\geq0}\cup\set{\infty}$.
Thus a finite distance is smaller than~$\infty$, while $\infty<\infty$ is false.
\end{definition}

Note how the distance atoms behave on disconnected graphs.
Vertices in different components have distance at least $d$ for every finite threshold $d$, whereas a covering condition can only be satisfied by vertices in the same component as the vertex that we wish to cover.
In particular, $\Cover_r(U,\varnothing)$ holds exactly when $U=\varnothing$.

\begin{definition}[Metric extensions]
\label{def:guarded-logics}
Let 
\[
  \mathcal A:=\set{\SP,\EqGC,\Dist,\Sep,\Cover,\NearLab,\EqDist,\Close}.
\]
Here $\Dist$, $\Sep$, and $\Cover$ stand for the families indexed by numerical bounds, while $\NearLab$ stands for the atoms $\NearLab_i$ indexed by input labels.
If $\mathcal L$ is one of the logics $\FO$, $\MSO$, $\CMSO$, $\MSOtwo$, and $\CMSOtwo$, and $\mathcal F\subseteq\mathcal A$, we write $\mathcal L[\mathcal F]$ for the extension in which we additionally allow the atoms from $\mathcal F$ with arguments of the appropriate sorts.
Adding an atom does not add new kinds of quantification, so in a first-order extension a set argument must be a named unary predicate or a fixed set parameter, rather than a quantified set variable.
We write, for example, $\mathcal L[\SP,\EqGC]$ instead of $\mathcal L[\{\SP,\EqGC\}]$.
\end{definition}

In particular, the guarded logics $\CMSOCTW$ and $\CMSOOCTW$ extend $\CMSOtwo$ and $\CMSO$, respectively, by $\SP$ and $\EqGC$ alone.
We treat the numerical distance conditions and nearest-label atoms in the separate fragments introduced above.
For the single-interface nearest-label fragment, every nearest-label atom refers to the same distinguished free set $S$, which we require to be nonempty and do not bind by a quantifier, while other variables may be quantified and additional free sets may occur.
The packing-covering, one-cover, and nearest-label fragments are described in \Cref{sec:dist-meta}.

\section{Expressing shortest paths in MSO}
\label{sec:guarded-model-checking} 

We begin our study with the shortest-path predicate and recall the main theorem that we prove in this section.

\shortestcertificate*

We prove the theorem in three steps. 
We first characterize shortest paths by integer potentials.
We then replace the unbounded potentials by finite cyclic labelings.
And finally, we express the resulting certificates in ordinary~$\MSO$. 

\subsection{Shortest paths as integer potentials}

For an integer $k\geq0$, a map $f\colon V(G)\to\ZZ$ is \emph{$k$-Lipschitz} if  $|f(x)-f(y)|\leq k$ on every edge~$xy$.
We use $1$-Lipschitz maps in this section and $2$-Lipschitz maps in the guarded-equidistance certificate in the next section.

A $1$-Lipschitz map, say $\lambda$, obviously gives a lower bound on the length of every path between two vertices, that is, every $s$-$t$ path has length at least $|\lambda(t)-\lambda(s)|$.
If the values of $\lambda$ increase by one along a path $P$ from $s$ to $t$, then we have $|E(P)|=\lambda(t)-\lambda(s)$, so every $s$-$t$ path has length at least $|E(P)|$, and hence $P$ is shortest. 
This is captured in the following well-known lemma; see e.g.\ Dantzig's treatment of the shortest-route problem~\cite{Dantzig1957} and \cite[Chapters~4 and~5]{AhujaMagnantiOrlin1993}.

\begin{lemma}\label{lem:integer-potential}
Let $P=v_0v_1\cdots v_r$ be a path in~$G$, with $s=v_0$ and $t=v_r$.
The following statements are equivalent.
\begin{enumerate}[label=(\alph*)]
  \item $P$ is a shortest $s$-$t$ path.
  \item There exists a map $\lambda\colon V(G)\to\ZZ$ such that
  \begin{align}
    \abs{\lambda(u)-\lambda(v)}&\leq1
      &&\text{for every }uv\in E(G), \label{eq:integer-potential}\\
    \lambda(v_{i+1})&=\lambda(v_i)+1
      &&\text{for every }0\leq i<r. \label{eq:saturated-path}
  \end{align}
\end{enumerate}
\end{lemma}

We cannot use \Cref{lem:integer-potential} directly in MSO, as $\lambda$ can take arbitrarily large integer values. 
On graphs with a bounded-length cycle basis, however, it is enough to track the labels modulo a suitable integer.

\subsection{Finite cyclic certificates}

We fix an integer $M\geq3$ and write $\mathbb Z_M=\ZZ/M\ZZ$ for the residues modulo~$M$, represented by $0,\ldots,M-1$.
We now perform all arithmetic on labels in $\mathbb Z_M$ modulo~$M$.

\begin{definition}\label{def:admissible-labeling}
A map $c\colon V(G)\to\mathbb Z_M$ is an \emph{admissible $\mathbb Z_M$-labeling} if, for every edge $uv\in E(G)$, the difference $c(v)-c(u)$ is congruent to one of $-1$, $0$, and~$1$ modulo~$M$.
Thus adjacent vertices receive equal or consecutive labels in the cyclic order (where $M-1$ and $0$ are also consecutive).
When the modulus is clear, we simply call $c$ an \emph{admissible labeling}.
\end{definition}

Since $M\geq3$, the integers $-1,0,1$ represent distinct residues, so an admissible labeling determines a unique integer increment in $\set{-1,0,1}$ whenever we traverse an edge from one endpoint to the other.
For a traversal from $u$ to $v$, we take the increment congruent to $c(v)-c(u)$, and when we reverse the traversal, we negate this increment.
We would like to recover an integer potential by adding the increments along a path from a fixed starting vertex, but for this to work, the result must be independent of the path that we choose.
We first prove a general lemma showing that we can recover a potential whenever the increments sum to zero on the cycles of a binary cycle basis.

To keep track of signs, we choose, independently for each undirected edge $e=uv$, one of the directions $u\to v$ and $v\to u$, which we call its \emph{reference orientation}.
This choice does not change the graph, which remains undirected, and we may still traverse every edge in either direction.
Given a function $\delta\colon E(G)\to\ZZ$, we interpret $\delta(e)$ as the increment in the reference direction, and $-\delta(e)$ as the increment in the opposite direction.
For a walk $W=v_0v_1\cdots v_\ell$, we define its \emph{signed sum} by
\[
  \delta(W):=\sum_{i=1}^{\ell}\varepsilon_i\delta(v_{i-1}v_i),
\]
where $\varepsilon_i=1$ if the reference orientation of $v_{i-1}v_i$ is $v_{i-1}\to v_i$, and $\varepsilon_i=-1$ otherwise.
We count each traversal separately, even when an edge occurs several times, and for a walk of length zero we take the sum to be zero.
When we sum around a cycle, we may choose either direction of traversal, since reversing that direction only changes the sign and therefore does not affect whether the sum is zero.

The following lemma is a consequence of the relation between binary and signed cycle bases~\cite[Lemma~2.4]{KavithaEtAl2009CycleBases} and the characterization of potential differences by vanishing signed cycle sums~\cite[Section~20.1]{BondyMurty2008}.
We give an elementary proof using a spanning tree and a parity argument.

\begin{lemma}
\label{lem:integer-cycle-sums}
\label{cor:appendix-closed-walk-orthogonality}
Fix a reference orientation of every edge of~$G$ and a function $\delta\colon E(G)\to\ZZ$.
If the signed sum of $\delta$ around every cycle of a binary cycle basis~$\mathcal B$ is zero, then there exists a map $\lambda\colon V(G)\to\ZZ$ such that $\lambda(v)-\lambda(u)=\delta(e)$ for every edge $e$ whose reference orientation is~$u\to v$.
Consequently, the signed sum of $\delta$ around every closed walk is zero.
\end{lemma}

\begin{proof}
We can prove the statement separately for each connected component, since every basis cycle lies in one component, so we may assume that $G$ is connected.
We choose a spanning tree~$T$ and a root $r$, and we define $\lambda(v)$ by adding the signed increments along the unique $r$-$v$ path in $T$.
In particular, we have $\lambda(r)=0$, and the differences of $\lambda$ on the tree edges are exactly the given increments.

We now check whether the same holds for the other edges.
For an edge $e$ whose reference orientation is $u\to v$, we put
$b(e):=\delta(e)-\bigl(\lambda(v)-\lambda(u)\bigr)$.
Thus $b(e)$ measures the difference between the increment that we want and the one given by $\lambda$.
We already know that $b(e)=0$ on every tree edge, and we want to show that it is zero on every edge.

When we add the potential differences along a closed walk $W=v_0v_1\cdots v_\ell$, all intermediate terms cancel, and since $v_\ell=v_0$, we have
\[
  \sum_{i=0}^{\ell-1}\bigl(\lambda(v_{i+1})-\lambda(v_i)\bigr)
  =\lambda(v_\ell)-\lambda(v_0)=0.
\]

This cancellation is a standard observation called \emph{telescoping}.
The signed sum of $\delta$ is zero on every basis cycle, and the same is true for the potential differences, so their difference $b$ also has signed sum zero on every basis cycle.

Suppose that some value of $b$ is nonzero.
We let $g$ be the positive greatest common divisor of the nonzero values $\abs{b(e)}$, and we put $a(e):=b(e)/g$.
All values of $a$ are integers, and at least one is odd, because otherwise $2g$ would divide every value of $b$, contrary to our choice of $g$.
We let $F$ be the nonempty set of edges on which $a$ is odd.
Since $b$ is zero on the tree edges, $F$ contains no edge of $T$.

The signed sum of $a$ is still zero on every basis cycle.
We now consider these sums modulo~$2$ (hence adding and subtracting have the same effect).
Each edge in $F$ contributes $1$, while every other edge contributes $0$, so every basis cycle contains an even number of edges from $F$.
Every cycle of $G$ is a symmetric difference of basis cycles, and taking symmetric differences preserves this property, so every cycle contains an even number of edges from $F$.

We now choose an edge $e\in F$ and take the cycle formed by $e$ and the unique path in $T$ between its endpoints.
This is called the \emph{fundamental cycle} of $e$ with respect to $T$.
All edges of this cycle except $e$ belong to $T$, so the cycle contains exactly one edge from $F$, which is a contradiction.

We conclude that $b(e)=0$ for every edge, and hence $\lambda$ has all the required differences.
The signed sum of these differences around every closed walk is zero by the cancellation shown above.
\end{proof}

We now apply the lemma to the increments of an admissible labeling.
Around any cycle, these increments sum to a multiple of $M$, because the label differences telescope modulo~$M$.
If the cycle has length smaller than $M$, the absolute value of the sum is also smaller than $M$, and hence the sum must be zero.
We use this observation to prove the following lifting statement, following the argument of Boulinier, Petit, and Villain for their intrinsic-delay criterion~\cite[\S4.3, Definition~4.4 and the proof of Theorem~4.5]{BoulinierPetitVillain2004}; see also~\cite[\S3.2]{BoulinierPetitVillain2008}.
\pagebreak

\begin{lemma}\label{lem:cyclic-lifting}
Let $M\geq3$ satisfy $M>\cbl(G)$, and let $c\colon V(G)\to\mathbb Z_M$ be an admissible labeling.
There exists a map $\lambda\colon V(G)\to\ZZ$ such that
\begin{enumerate}[label=(\roman*)]
  \item $\lambda(v)\equiv c(v)\pmod M$ for every $v\in V(G)$, and
  \item for every pair of adjacent vertices $u,v$, the difference $\lambda(v)-\lambda(u)$ is the unique integer in $\set{-1,0,1}$ congruent to $c(v)-c(u)$ modulo~$M$.
\end{enumerate}
We call such a map $\lambda$ an \emph{integer lift} of $c$.
\end{lemma}

\begin{proof}
We choose a reference orientation for each edge, and for an edge $e$ with reference direction~\mbox{$u\to v$}, we take $\delta(e)$ to be the unique integer in $\set{-1,0,1}$ that is congruent to $c(v)-c(u)$ modulo~$M$.
We choose a binary cycle basis $\mathcal B$ in which every cycle has length at most~$\cbl(G)$, and if~$G$ is a forest, we take the empty basis.

For each cycle $C\in\mathcal B$, the label differences cancel when we add them around the cycle, so the signed sum of $\delta$ is zero modulo~$M$.
This means that the sum is divisible by $M$.
On the other hand, each increment has absolute value at most one, so the absolute value of the sum is at most $\abs{E(C)}\leq\cbl(G)<M$.
The only multiple of $M$ with absolute value smaller than $M$ is zero, so the signed sum is zero on every basis cycle.

We can therefore apply \Cref{lem:integer-cycle-sums} and obtain an integer potential $\lambda$ whose difference on every edge is the given increment.
The same holds when we traverse an edge in the opposite direction, since both the increment and the potential difference change sign.

We still need to make sure that $\lambda(v)$ has the same residue as $c(v)$ at every vertex.
Within each connected component, we can add the same integer to all values of $\lambda$ without changing any edge difference.
We choose a root $r$ in each component and use this freedom to make $\lambda(r)\equiv c(r)\pmod M$.
For any vertex $v$ in that component, we then add the edge differences along an $r$-$v$ path and obtain
\[
  \lambda(v)-\lambda(r)\equiv c(v)-c(r)\pmod M.
\]
Since the labels at $r$ agree modulo~$M$, the labels at $v$ also agree, and hence $\lambda(v)\equiv c(v)\pmod M$, as required.
\end{proof}

In \Cref{fig:cyclic-lifting}, we show how we pass from labels modulo~$11$ to integer labels on an $8$-cycle.
The step from $10$ to $0$ in the cyclic labeling becomes the step from $-1$ to $0$ in the integer labeling, so in both cases we increase the value by one.

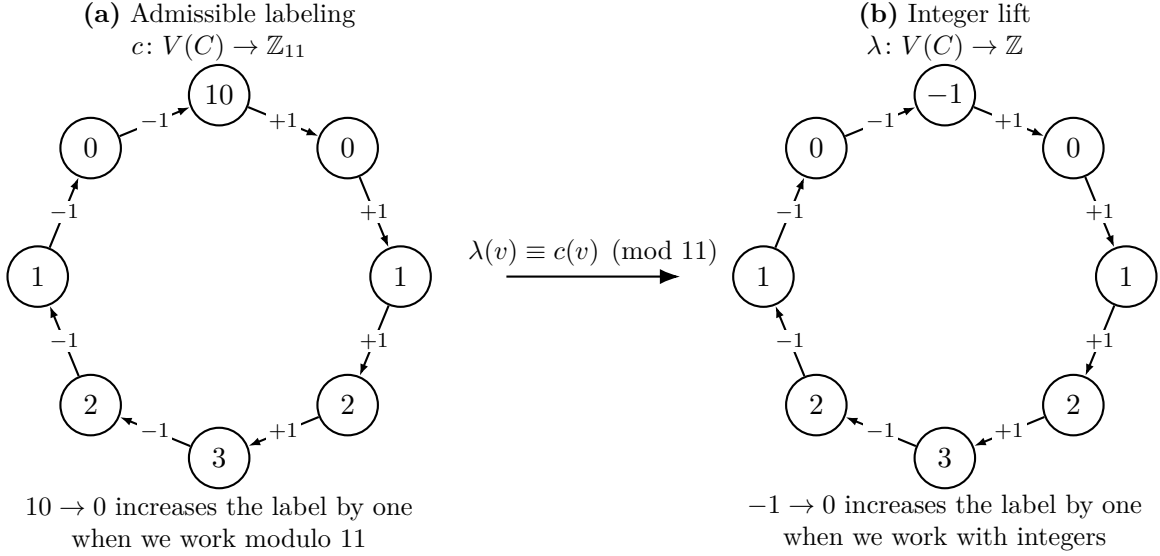
\begin{figure}[t]
  \centering
  \begin{tikzpicture}[
      vertex/.style={
        circle,
        draw,
        thick,
        minimum size=8mm,
        inner sep=0pt,
        fill=white
      },
      edge label/.style={
        font=\scriptsize,
        fill=white,
        inner sep=1pt
      },
      cycle edge/.style={
        thick,
        -{Latex[length=1.6mm,width=1.1mm]}
      },
      panel title/.style={
        font=\small,
        align=center
      },
      lift/.style={
        -{Latex[length=2.7mm]},
        thick
      },
      >=Latex
    ]

    \node[vertex] (l1) at (0,2.4)       {$10$};
    \node[vertex] (l2) at (1.7,1.7)     {$0$};
    \node[vertex] (l3) at (2.4,0)       {$1$};
    \node[vertex] (l4) at (1.7,-1.7)    {$2$};
    \node[vertex] (l5) at (0,-2.4)      {$3$};
    \node[vertex] (l6) at (-1.7,-1.7)   {$2$};
    \node[vertex] (l7) at (-2.4,0)      {$1$};
    \node[vertex] (l8) at (-1.7,1.7)    {$0$};

    \draw[cycle edge]
      (l1) -- node[edge label] {$+1$} (l2);
    \draw[cycle edge]
      (l2) -- node[edge label] {$+1$} (l3);
    \draw[cycle edge]
      (l3) -- node[edge label] {$+1$} (l4);
    \draw[cycle edge]
      (l4) -- node[edge label] {$+1$} (l5);
    \draw[cycle edge]
      (l5) -- node[edge label] {$-1$} (l6);
    \draw[cycle edge]
      (l6) -- node[edge label] {$-1$} (l7);
    \draw[cycle edge]
      (l7) -- node[edge label] {$-1$} (l8);
    \draw[cycle edge]
      (l8) -- node[edge label] {$-1$} (l1);

    \node[panel title] at (0,3.25)
      {\textbf{(a)} Admissible labeling\\
       $c\colon V(C)\to\mathbb Z_{11}$};

    \node[font=\small,align=center] at (0,-3.25)
      {$10\to0$ increases the label by one\\
       when we work modulo~$11$};

    \draw[lift] (3.8,0) -- node[above,font=\small,align=center]
      {$\lambda(v)\equiv c(v)\pmod{11}$} (6.1,0);

    \begin{scope}[xshift=9.6cm]
      \node[vertex] (r1) at (0,2.4)       {$-1$};
      \node[vertex] (r2) at (1.7,1.7)     {$0$};
      \node[vertex] (r3) at (2.4,0)       {$1$};
      \node[vertex] (r4) at (1.7,-1.7)    {$2$};
      \node[vertex] (r5) at (0,-2.4)      {$3$};
      \node[vertex] (r6) at (-1.7,-1.7)   {$2$};
      \node[vertex] (r7) at (-2.4,0)      {$1$};
      \node[vertex] (r8) at (-1.7,1.7)    {$0$};

      \draw[cycle edge]
        (r1) -- node[edge label] {$+1$} (r2);
      \draw[cycle edge]
        (r2) -- node[edge label] {$+1$} (r3);
      \draw[cycle edge]
        (r3) -- node[edge label] {$+1$} (r4);
      \draw[cycle edge]
        (r4) -- node[edge label] {$+1$} (r5);
      \draw[cycle edge]
        (r5) -- node[edge label] {$-1$} (r6);
      \draw[cycle edge]
        (r6) -- node[edge label] {$-1$} (r7);
      \draw[cycle edge]
        (r7) -- node[edge label] {$-1$} (r8);
      \draw[cycle edge]
        (r8) -- node[edge label] {$-1$} (r1);

      \node[panel title] at (0,3.25)
        {\textbf{(b)} Integer lift\\
         $\lambda\colon V(C)\to\mathbb Z$};

      \node[font=\small,align=center] at (0,-3.25)
        {$-1\to0$ increases the label by one\\
         when we work with integers};
    \end{scope}

  \end{tikzpicture}

  \caption{An admissible labeling of an $8$-cycle modulo~$11$ and its integer lift.
  The arrows show the direction in which we read the increments, while the cycle itself is undirected.
  Every edge has the same increment in both drawings, and these increments add to $1+1+1+1-1-1-1-1=0$ around the cycle.}
  \label{fig:cyclic-lifting}
\end{figure}

Combining the integer lift with \Cref{lem:integer-potential}, we obtain a certificate for shortest paths that uses only finitely many labels.

\begin{theorem}\label{thm:finite-certificate}
Let $M\geq3$ satisfy $M>\cbl(G)$, and let $P=v_0v_1\cdots v_r$ be a path with $s=v_0$ and $t=v_r$.
Then $P$ is a shortest $s$-$t$ path if and only if there exists an admissible labeling $c\colon V(G)\to\mathbb Z_M$ such that $c(v_{i+1})=c(v_i)+1$ in $\mathbb Z_M$ for every $0\leq i<r$.
\end{theorem}

\begin{proof}
If $P$ is shortest, we take the integer potential from \Cref{lem:integer-potential} and reduce its values modulo~$M$.
We obtain an admissible labeling, and its values increase by one along every edge of~$P$ from $s$ to $t$.

Conversely, if we have such an admissible labeling, we apply \Cref{lem:cyclic-lifting} to obtain an integer lift.
The difference on every edge from $v_i$ to $v_{i+1}$ is then exactly $1$, since this is the unique element of $\set{-1,0,1}$ with the required residue.
The lift is also $1$-Lipschitz on the whole graph, so \Cref{lem:integer-potential} shows that $P$ is shortest.
\end{proof}

\subsection{Expressing the finite certificate in plain MSO}
\label{sec:SP-formula}

We now write the certificate as an MSO formula.
We fix $M\geq3$ and write $\Path(P,s,t)$ for an MSO formula saying that $P$ is an $s$-$t$ path set.
If $s=t$, we require $P=\set{s}$.
Otherwise, we require that $s,t\in P$, that $G[P]$ is connected, and that $s$ and $t$ have degree one while every other vertex has degree two (in the graph $G[P]$).
We can express these conditions by MSO formulas of constant length.

We describe the labeling by vertex-set variables $\bar C=(C_0,\ldots,C_{M-1})$, where $C_i$ contains the vertices with label $i$.
We read all indices modulo~$M$, so $C_M=C_0$ and $C_{-1}=C_{M-1}$.
We write $\Part_M(\bar C)$ for the formula saying that every vertex belongs to exactly one of these sets, although some sets may be empty.
We write $\Adm_M(\bar C)$ for the formula requiring, for each $i$, that every neighbor of a vertex in $C_i$ belongs to at least one of the sets $C_{i-1}$, $C_i$, and $C_{i+1}$.
Together, $\Part_M(\bar C)$ and $\Adm_M(\bar C)$ describe an admissible labeling of the whole graph, including the vertices outside $P$.

We still need the labels to increase by one as we follow the path from $s$ to $t$.
For this, we use the formula
\begin{align*}
  \Orient_M(P,s,t,\bar C):={}&
  \forall x\,\Bigl((x\in P\wedge x\neq t)\rightarrow
    \exists y\,\bigl(y\in P\wedge E(x,y)\\
  &\hspace{22mm}{}\wedge
    \bigvee_{i=0}^{M-1}(x\in C_i\wedge y\in C_{i+1})\bigr)\Bigr).
\end{align*}

This formula says that every vertex of $P$ except $t$ has a neighbor in $P$ with the next cyclic label.
The variable $s$ does not appear in the formula, but we write it with the arguments because we use the formula together with $\Path(P,s,t)$.
We then define
\begin{align}
  \Cert_M(P,s,t):={}&\Path(P,s,t)\wedge
    \exists C_0\cdots\exists C_{M-1}\Bigl(
      \Part_M(\bar C)\wedge\Adm_M(\bar C)\notag\\
    &\hspace{40mm}{}\wedge\Orient_M(P,s,t,\bar C)
    \Bigr).
  \label{eq:cert-formula}
\end{align}

\begin{proof}[Proof of \Cref{thm:cert-mso}]
We assume that $M>\cbl(G)$ and first suppose that $\Cert_M(P,s,t)$ holds with $s\neq t$.
By the path condition, we can write $G[P]=v_0v_1\cdots v_r$, where $v_0=s$ and $v_r=t$, and the sets $C_0,\ldots,C_{M-1}$ give an admissible labeling $c$.

We start at $s$, whose only neighbor in $P$ is $v_1$, so $\Orient_M$ requires $c(v_1)=c(v_0)+1$.
We now suppose that the labels increase by one along the path up to $v_i$, where $1\leq i<r$.
The vertex $v_{i-1}$ has the previous label, which is different from the next label because $M\geq3$.
Since the only neighbors of $v_i$ in $P$ are $v_{i-1}$ and $v_{i+1}$, the neighbor with the next label must be $v_{i+1}$.
Continuing along the path, we see that the labels increase by one at every step, so \Cref{thm:finite-certificate} shows that the path is shortest.

Conversely, we suppose that $P$ is a shortest $s$-$t$ path set with $s\neq t$.
By \Cref{thm:finite-certificate}, we can choose an admissible labeling whose values increase by one along the path from $s$ to $t$.
We use its label classes as the sets $C_0,\ldots,C_{M-1}$, which then satisfy $\Part_M$ and $\Adm_M$.
Every vertex of the path except $t$ has the next vertex on the path as a neighbor with the next label, so $\Orient_M$ also holds.

If $s=t$, both the shortest-path condition and the path formula require $P=\set{s}$.
For this one-vertex path, there is nothing to check in $\Orient_M$, and we can give every vertex of the graph the same label.
Thus $\Cert_M(P,s,t)$ is equivalent to $\SP(P,s,t)$ in all cases.

Finally, we check the length of the formula when we write out all abbreviations.
For $\Part_M$, we use one condition with $M$ alternatives to say that every vertex has a label, and $O(M^2)$ conditions to say that no vertex has two different labels.
For $\Adm_M$, we use $M$ conditions of constant size, while $\Orient_M$ contains a disjunction with $M$ alternatives.
Since the path formula has constant length, the resulting ordinary MSO formula has length polynomial in~$M$.
\end{proof}

We can now use the certificate formula to replace shortest-path atoms in larger formulas.
For $\varphi\in\CMSOSP$, we define $\tau_M(\varphi)\in\CMSO$ by replacing every atom $\SP(P,s,t)$ with a copy of $\Cert_M(P,s,t)$.
In each copy, we give the quantified variables new names that do not occur elsewhere in the formula, so their quantifiers do not bind any other variables by mistake.
We leave all ordinary atoms unchanged, including modulo-counting atoms and unary input predicates, and we keep the Boolean connectives and quantifiers of the original formula.
It is standard to prove equivalence by induction on the structure of $\phi$, and we obtain the following lemma. 

\begin{lemma}\label{lem:translation}
Let $G$ be a graph and let $M\geq3$ satisfy $M>\cbl(G)$.
For every $\varphi\in\CMSOSP$ and every assignment $\alpha$ of its free variables, we have
\begin{align*}
  G,\alpha\models\varphi
  \quad\Longleftrightarrow\quad
  G,\alpha\models\tau_M(\varphi).
\end{align*}
Moreover, there is a fixed polynomial $q$ such that $\abs{\tau_M(\varphi)}\leq q(\abs{\varphi}+M)$.
\end{lemma}

For the incidence encoding, we use the same certificate, but we require its newly introduced element variables to range over vertices and its set variables to range over subsets of $V(G)$.
We obtain the same translation from $\CMSOtwoSP$ to $\CMSOtwo$, with the same equivalence and polynomial bound on the formula length.

\section{Expressing guarded equidistance in MSO}
\label{sec:eqgc-in-mso}

We now turn to the guarded-equidistance predicate and prove the following theorem, which we recall from the introduction.

\guardedcertificate*

As in the shortest-path proof, we first describe the property by an integer potential, then replace its values by finitely many labels, and finally express the resulting conditions in MSO.
This time, we compare two arcs of a cycle by following them in the same direction around the cycle and subtracting their edge contributions.
We show that the sum of these differences around the cycle is zero exactly when the two arcs have the same length.

\subsection{Equal arc lengths as integer potentials}

We consider a vertex set $C$ such that $G[C]$ is a cycle.
An $a$-$b$ \emph{arc} of $C$ is a path subgraph of~$G[C]$ with endpoints $a$ and $b$, and when $a=b$, we allow the \emph{stationary arc} consisting only of $a$, whose length is zero.
For now, we view an arc as a path subgraph, so we allow either of the two arcs between distinct endpoints, even when its vertex set does not induce a path (this can happen when the endpoints are adjacent on the cycle and we choose the longer arc).
We let $P$ be an $s$-$t$ arc and~$Q$ a $u$-$v$ arc, where the arcs may overlap and their endpoints need not be distinct.

Here, we do not fix a reference orientation as in the previous section.
By an \emph{oriented edge}~$(r,z)$, we simply mean an ordered pair of adjacent vertices, which specifies that we traverse the edge from $r$ to $z$.
Thus, for every undirected edge $rz$, we consider both $(r,z)$ and $(z,r)$, and the signed increments defined below change sign when we reverse this order.
We define $\varepsilon_P(r,z)$ to be $1$ if the traversal of $P$ from $s$ to $t$ uses the edge from $r$ to $z$, to be $-1$ if it uses that edge in the opposite direction, and to be $0$ if it does not use the edge.
We define $\varepsilon_Q$ in the same way, using the traversal of $Q$ from $u$ to $v$, and all these values are zero for a stationary arc.
To compare the traversal of $P$ with either direction of traversal of $Q$, we put
\begin{align*}
  \beta^+_{P,Q}(r,z):=\varepsilon_P(r,z)-\varepsilon_Q(r,z), \quad \text{ and } \quad
  \beta^-_{P,Q}(r,z):=\varepsilon_P(r,z)+\varepsilon_Q(r,z).
\end{align*}
Here the superscript $+$ means that we keep the traversal of $Q$, while the superscript $-$ means that we reverse it before subtracting its contribution.
Reversing the traversal replaces $\varepsilon_Q$ by~$-\varepsilon_Q$, which explains why we use addition in the second formula.
Both functions take values in $\{-2,-1,0,1,2\}$, and reversing the oriented edge changes the sign of the value.

To construct the integer vertex labels, we start with label zero at one vertex and then add the required change at each edge as we go around the cycle.
These labels are consistent at the starting vertex precisely when the sum of the differences is zero, since this is the value that we obtain when we return to that vertex.
We apply this observation to the differences $\beta^\eta_{P,Q}$ in the following lemma.

\begin{lemma}\label{lem:eqgc-potential-outline}
\label{lem:eqgc-integer-potential}
\label{cor:eqgc-directed-potential}
The arcs $P$ and $Q$ have equal lengths if and only if there exist $\eta\in\{+,-\}$ and a map $q\colon C\to\ZZ$ such that
\begin{equation}
  q(z)-q(r)=\beta^\eta_{P,Q}(r,z)
  \qquad\text{for every oriented edge $(r,z)$ of $G[C]$}.
  \label{eq:eqgc-gradient-outline}
\end{equation}
If the lengths agree, we can choose $\eta$ and $q$ so that $q$ is $1$-Lipschitz on $G[C]$.
\end{lemma}

\begin{proof}
We first suppose that the two arcs have the same positive length, and we choose $\eta$ so that the traversal of $P$ and the selected traversal of $Q$ follow the same direction around the cycle.
We then traverse the whole cycle in this direction.
The contributions of the two arcs cancel on their common edges, so the prescribed differences belong to $\{-1,0,1\}$ and have sum $|E(P)|-|E(Q)|=0$.
We therefore obtain consistent integer labels by starting with zero at one vertex and adding these differences around the cycle.
The labels satisfy~\eqref{eq:eqgc-gradient-outline}, and the absolute difference across every edge is at most one, so the resulting map is $1$-Lipschitz.
If both arcs are stationary, all prescribed differences are zero, and we use a constant map instead.

For the converse, we suppose that $q$ satisfies~\eqref{eq:eqgc-gradient-outline} and choose one of the two directions around the cycle.
We write $\vec E(C)$ for the set of edges ordered in this cyclic direction, and we use $\sigma_P,\sigma_Q\in\{-1,1\}$ to indicate whether each arc follows or opposes this direction, using the traversal of $Q$ selected by $\eta$.
For a stationary arc, we may choose either sign, since its length is zero.
When we sum the potential differences around the cycle, they cancel, and we obtain
\begin{align*}
  0=\sum_{(r,z)\in\vec E(C)}\beta^\eta_{P,Q}(r,z)
   =\sigma_P|E(P)|-\sigma_Q|E(Q)|.
\end{align*}

The two signed lengths are therefore equal, and taking their absolute values gives \mbox{$|E(P)|=|E(Q)|$}.
\end{proof}

\subsection{Finite cyclic certificates}

We again use residues modulo an integer to represent the potential by finitely many label classes.
For the lifting argument from \Cref{sec:guarded-model-checking}, we need a labeling of the whole graph, although the prescribed differences concern only the cycle.
This is where we use isometry: a potential that is $1$-Lipschitz on the cycle also satisfies $|q(a)-q(b)|\leq d_G(a,b)$ for vertices $a,b$ of the cycle, so we can extend it to the rest of the graph.
In the finite certificate, we allow edge differences of absolute value at most two, since these include all possible values of $\beta^+$ and $\beta^-$.

\begin{lemma}\label{lem:eqgc-finite-certificate-outline}
\label{lem:eqgc-modular-characterization}
Suppose that $G[C]$ is an isometric cycle, and let $P$ be an $s$-$t$ arc and $Q$ a $u$-$v$ arc of~$C$.
Let $N$ be an integer with $N>4$ and $N>2\cbl(G)$.
Then $|E(P)|=|E(Q)|$ if and only if there exist $\eta\in\{+,-\}$ and a map $c\colon V(G)\to\mathbb Z_N$ satisfying the following conditions.
\begin{enumerate}[label=(\roman*)]
  \item For every oriented edge $(r,z)$ of $G$, the residue $c(z)-c(r)$ has an integer representative in $\{-2,-1,0,1,2\}$.
  \item For every oriented edge $(r,z)$ of $G[C]$, we have $c(z)-c(r)\equiv\beta^\eta_{P,Q}(r,z)\pmod N$.
\end{enumerate}
\end{lemma}

\begin{proof}
For the forward implication, we choose $\eta$ and the $1$-Lipschitz potential $q$ from \Cref{lem:eqgc-potential-outline}.
Since the cycle is isometric, we have $|q(a)-q(b)|\leq d_{G[C]}(a,b)=d_G(a,b)$ for all $a,b\in C$.
We write~$H$ for the connected component containing $C$ and extend $q$ to its vertices by putting
\begin{equation}
  \widehat q(z):=\min_{a\in C}\bigl(q(a)+d_G(a,z)\bigr)
  \qquad\text{for }z\in V(H).
  \label{eq:discrete-mcshane}
\end{equation}
The minimum is a finite integer, since $C$ is nonempty and all the distances in this expression are finite.
For $z\in C$, the choice $a=z$ gives $\widehat q(z)\leq q(z)$, while the inequality $q(z)\leq q(a)+d_G(a,z)$ for every $a\in C$ gives the reverse bound.
Thus the extension agrees with $q$ on $C$.

To check the Lipschitz condition, we consider an edge $rz$ of $H$ and choose a vertex $a\in C$ that gives the minimum for $z$.
We then have
\[
  \widehat q(r)\leq q(a)+d_G(a,r)
  \leq q(a)+d_G(a,z)+1
  =\widehat q(z)+1.
\]
By interchanging $r$ and $z$, we obtain the reverse inequality, so $\widehat q$ is $1$-Lipschitz on $H$.
We give all vertices outside $H$ label zero and reduce the resulting integer labels modulo $N$.
The labeling satisfies condition~(i), and it satisfies condition~(ii) because we have preserved the potential on~$C$.

For the reverse implication, we fix a reference orientation of every edge, as in \Cref{lem:integer-cycle-sums}.
For an edge $e$ whose reference orientation is $r\to z$, we let $\delta(e)$ be the representative of $c(z)-c(r)$ in $\{-2,-1,0,1,2\}$, which is unique because $N>4$.
We choose a binary cycle basis $\mathcal B$ whose cycles have length at most $\cbl(G)$.
Around each $B\in\mathcal B$, the label differences sum to zero modulo $N$, so the signed sum of $\delta$ is divisible by $N$.
Its absolute value is at most $2|E(B)|\leq2\cbl(G)<N$, and hence this sum must be zero.
By \Cref{lem:integer-cycle-sums}, we obtain an integer potential $p$ with the prescribed differences on every edge of $G$.
On an edge of $G[C]$, condition~(ii) forces this difference to equal~$\beta^\eta_{P,Q}$, since both values belong to $\{-2,-1,0,1,2\}$ and are congruent modulo $N$.
The restriction of~$p$ to $C$ therefore satisfies~\eqref{eq:eqgc-gradient-outline}, and \Cref{lem:eqgc-potential-outline} gives equality of the arc lengths.
\end{proof}

\subsection{Expressing the finite certificate in plain MSO}

We fix $M\geq3$ and put $N:=2M+1$, so whenever $M>\cbl(G)$, the assumptions on $N$ in \Cref{lem:eqgc-finite-certificate-outline} hold.
We use the path and partition formulas from \Cref{sec:SP-formula} and describe the remaining conditions below.

We write $\mathsf{Cycle}(C)$ for the formula saying that $C$ is nonempty, $G[C]$ is connected, and every vertex of $C$ has exactly two neighbors in $C$.
We also put
$\mathsf{Arc}(C,P,a,b):=\Path(P,a,b)\wedge P\subseteq C$.
Note that the only arcs that this representation omits are the longer arcs between adjacent cycle vertices, whose vertex sets induce the whole cycle, and we will not need those arcs in the certificate.

We write $\mathsf{Isometric}_M(C)$ for the formula requiring $\mathsf{Cycle}(C)$ and, for every two distinct vertices $a,b\in C$, a set $P$ satisfying $\mathsf{Arc}(C,P,a,b)\wedge\Cert_M(P,a,b)$.
If $M>\cbl(G)$, then \Cref{thm:cert-mso} shows that every witnessing path is shortest in $G$, so we obtain $d_{G[C]}(a,b)=d_G(a,b)$.
Conversely, when $G[C]$ is isometric, a shortest arc between $a$ and $b$ is also shortest in $G$ and provides the required witness.
Thus $\mathsf{Isometric}_M(C)$ expresses exactly that $G[C]$ is an isometric cycle.

We next describe the direction in which a path traverses one of its edges, using the cut that remains when we remove that edge.
For a path set $P$ with endpoints $a,b$, we let $\mathsf{Forw}(P,a,b,r,z)$ require $r,z\in P$, $E(r,z)$, and a set $Z\subseteq P$ containing $a,r$ but neither $b$ nor $z$, such that $rz$ is the only edge of $G[P]$ between $Z$ and $P\setminus Z$.
All these requirements have constant-length MSO formulas.
If the traversal from $a$ to $b$ uses the edge from $r$ to $z$, we take $Z$ to be the vertices of the path prefix ending at $r$.
Conversely, any such set $Z$ forces the traversal to cross from $r$ to $z$, because this is its only edge to the part containing $b$.
For a stationary path, the formula is false, since it would require its only endpoint to belong to $Z$ and not belong to $Z$ at the same time.

For an $s$-$t$ path $P$, we can therefore express $\varepsilon_P(r,z)=1$ by $\mathsf{Forw}(P,s,t,r,z)$ and $\varepsilon_P(r,z)=-1$ by $\mathsf{Forw}(P,s,t,z,r)$, while the value zero is expressed by the negation of both of the conditions.
We can now form, for every $j\in\{-2,-1,0,1,2\}$ and $\eta\in\{+,-\}$, a constant-length formula
\[
  \mathsf{Beta}^{\eta}_j(P,s,t,Q,u,v,r,z)
\]
expressing $\beta^\eta_{P,Q}(r,z)=j$.
For $\eta=+$, we take the disjunction over the pairs $h,k\in\{-1,0,1\}$ with $h-k=j$, requiring $\varepsilon_P(r,z)=h$ and $\varepsilon_Q(r,z)=k$.
For $\eta=-$, we use the pairs with $h+k=j$ instead, so in either case we have a fixed finite disjunction of the formulas just described.

As in the shortest-path certificate, we use sets $\bar X=(X_0,\ldots,X_{N-1})$ to partition all vertices into cyclic label classes.
For $j\in\{-2,-1,0,1,2\}$, we put
\begin{align*}
  \mathsf{Diff}^{\bar X}_j(r,z)
  :=\bigvee_{i=0}^{N-1}\bigl(r\in X_i\wedge z\in X_{[i+j]_N}\bigr),
\end{align*}
where $[i+j]_N$ is the representative of $i+j$ modulo $N$ in $\{0,\ldots,N-1\}$.
Under the partition condition, this formula says that the label of $z$ minus the label of $r$ is congruent to $j$ modulo $N$.

For each fixed choice of $\eta\in\{+,-\}$, we let $\mathsf{Wind}^{\eta}_N(C,P,s,t,Q,u,v)$ existentially quantify the label sets $\bar X$ and require $\Part_N(\bar X)$ together with the following conditions.
For every ordered pair of adjacent vertices $r,z$ of $G$, we require at least one of the five formulas $\mathsf{Diff}^{\bar X}_j(r,z)$ to hold.
If $r,z\in C$, we additionally require at least one of the five conjunctions
\[
  \mathsf{Beta}^{\eta}_j(P,s,t,Q,u,v,r,z)
  \wedge\mathsf{Diff}^{\bar X}_j(r,z)
\]
to hold.
We thus bound the differences on every edge of the graph and prescribe their values on the cycle edges, checking both directions of each edge.
These two directions give compatible requirements, since reversing an edge changes the sign of its prescribed difference.

We define $\mathsf{EqLen}^{C}_N(P,s,t,Q,u,v)$ as the disjunction
\[
  \mathsf{Wind}^{+}_N(C,P,s,t,Q,u,v)
  \vee\mathsf{Wind}^{-}_N(C,P,s,t,Q,u,v).
\]
In this way, we make one choice of $\eta$ for the whole cycle, as required by the finite certificate.
For arcs of an isometric cycle and $M>\cbl(G)$, \Cref{lem:eqgc-finite-certificate-outline} shows that this formula holds exactly when the arcs have equal lengths.
We do not need to choose a cyclic orientation in the formula, since the endpoints of the paths determine their traversals and the two alternatives account for reversing $Q$.

We can now choose shortest arcs for the two distances that we want to compare and require their lengths to agree.
We define
\begin{align*}
  \mathsf{CertEqGC}_M(C,x,y,u,v):={}&
    \mathsf{Isometric}_M(C)
    \wedge x\in C\wedge y\in C\wedge u\in C\wedge v\in C\\
    &{}\wedge\exists P\exists Q\Bigl(
      \mathsf{Arc}(C,P,x,y)\wedge\Cert_M(P,x,y)\\
    &\hspace{23mm}{}\wedge\mathsf{Arc}(C,Q,u,v)\wedge\Cert_M(Q,u,v)\\
    &\hspace{23mm}{}\wedge\mathsf{EqLen}^{C}_N(P,x,y,Q,u,v)
    \Bigr).
\end{align*}

\begin{proof}[Proof of \Cref{thm:eqgc-cert-mso}]
We assume that $M>\cbl(G)$ and first consider an assignment satisfying the certificate.
Then $G[C]$ is an isometric cycle containing the four distinguished vertices, and by \Cref{thm:cert-mso}, the witnessing paths have lengths $d_G(x,y)$ and $d_G(u,v)$.
The final condition says that these lengths agree, so $\EqGC(C,x,y,u,v)$ holds.

Conversely, suppose that $\EqGC(C,x,y,u,v)$ holds.
We choose shortest $x$-$y$ and $u$-$v$ arcs in the cycle, using a stationary arc whenever the two endpoints coincide.
The arcs have equal lengths and, by isometry, are shortest in $G$, so their vertex sets satisfy the path certificates.
By \Cref{lem:eqgc-finite-certificate-outline}, we also have the label classes required by $\mathsf{EqLen}^{C}_N$, and hence all conditions of the final certificate hold.

The formulas for cycles, arcs, and traversal directions have constant length.
The partition condition uses $O(N^2)$ pairwise-disjointness conditions and a covering condition of length $O(N)$, while each modular-difference formula uses $O(N)$ alternatives.
Only a constant number of copies of $\Cert_M$ occur in the expanded certificate, so, since $N=2M+1$, its complete length is polynomial in $M$.
When we expand the abbreviations, we choose fresh bound variables each time to avoid capturing variables from the surrounding formula.
\end{proof}

We now replace both kinds of metric atoms inside larger formulas.
For $\varphi\in\CMSOOCTW$, we obtain $\sigma_M(\varphi)\in\CMSO$ by replacing each $\SP$ atom by a fresh copy of $\Cert_M$ and each $\EqGC$ atom by a fresh copy of $\mathsf{CertEqGC}_M$.
We leave all ordinary atoms, including modulo-counting atoms and unary input predicates, unchanged, and we apply the translation recursively through Boolean connectives and quantifiers, again yielding the standard lemma. 

\begin{lemma}\label{lem:guarded-translation}
Let $G$ be a graph and let $M\geq3$ satisfy $M>\cbl(G)$.
For every $\CMSOOCTW$ formula $\varphi$ and every assignment $\alpha$ of its free variables, we have $G,\alpha\models\varphi$ if and only if $G,\alpha\models\sigma_M(\varphi)$.
Moreover, there is a fixed polynomial $q$ such that $\abs{\sigma_M(\varphi)}\leq q(\abs{\varphi}+M)$.
\end{lemma}

For incidence structures, we again restrict the auxiliary element variables to vertices and the auxiliary set variables to subsets of $V(G)$.
The same statements then hold for the translation from $\CMSOCTW$ to $\CMSOtwo$.

\section{Model checking with shortest paths and guarded equidistance}
\label{sec:first-meta-theorem}

We now combine the two certificates with the model-checking, optimization, and counting theorems from \Cref{sec:mso} to prove our first main theorem.
We restate it here for convenience.

\strongtheorem*

We use the optimization and counting conventions from \Cref{def:mso-optimization,def:mso-counting}, so weights are signed integers given in binary and we measure the running time in the total binary input length.
We do not require a tree decomposition or a cliquewidth expression as part of the input.

\begin{proof}
By Horton's algorithm and \Cref{lem:mcb-bottleneck}, we can compute $\cbl(G)$ in polynomial time~\cite{Horton1987}.
We put $M:=\max\set{3,\cbl(G)+1}$ and form the translated formula $\psi:=\sigma_M(\varphi)$.
By \Cref{lem:guarded-translation}, the formula $\psi$ belongs to $\CMSOtwo$ in part~\textnormal{(i)} and to $\CMSO$ in part~\textnormal{(ii)}, has length polynomial in $\abs{\varphi}+M$, and agrees with $\varphi$ on every assignment.
We apply \Cref{thm:courcelle-treewidth,thm:courcelle-cliquewidth}, respectively, and obtain fixed-parameter algorithms with parameters $\abs{\varphi}+\tw(G)+\cbl(G)$ and $\abs{\varphi}+\cw(G)+\cbl(G)$.
For part~\textnormal{(i)}, the bounds $\tw(G)\leq\ctw(G)$ and $\cbl(G)\leq2\ctw(G)$ from the definition of bag-connected treewidth and~\eqref{eq:cycle-length-bounds} also give fixed-parameter tractability with parameter $\abs{\varphi}+\ctw(G)$.
For part~\textnormal{(ii)}, we simply use the definition $\gcw(G)=\cw(G)+\cbl(G)$, so the second parameter is $\abs{\varphi}+\gcw(G)$.

The translation keeps the same free variables and exactly the same satisfying assignments, so it preserves every additive objective and the number of satisfying or optimum assignments.
The optimization and exact counting statements therefore follow from the same two theorems, including for signed binary weights.
\end{proof}

\section{Dynamic programming on classes of bounded bag-connected treewidth}
\label{sec:dist-meta}

In this section, we allow the distance bounds to be part of the input and prove the packing-covering theorem, \Cref{thm:dist-intro}, together with the results for one-cover and nearest-label formulas, \Cref{thm:one-cover,thm:nearest-label-intro}.
We first recall the three kinds of formulas, and then give one dynamic program that handles all of them.
As explained in the introduction, our main observation is that, on a connected bag, a distance function has only a small number of possible restrictions.
We store these restrictions together with the logical information needed to evaluate the $\CMSOCTW$ side formula.

The covering fragments were suggested by ChatGPT 6.0 Astra and were revised and verified by the authors; see the \hyperref[sec:ai-methodology]{AI methodology statement} at the end of the paper.

We start by recalling the definition of packing-covering templates.
Recall that $\Sep_d(X,Y)$ requires distance at least $d$ between every pair of distinct vertices $x\in X$ and $y\in Y$, while $\Cover_r(U,Y)$ requires every vertex of~$U$ to lie within distance~$r$ of~$Y$.

\begin{definition}[Packing-covering template]
\label{def:packing-covering-template}
Fix vertex-set variables $X_1,\ldots,X_p$ and a formula $\psi(X_1,\ldots,X_p)\in\CMSOCTW$ over a fixed finite expansion of the graph by unary input predicates.
A \emph{packing-covering template} is an expression
\[
  \Phi(\bar X;\bar\delta,\bar\rho)
  :=\psi(\bar X)
  \wedge\bigwedge_{a\in A}\Sep_{\delta_a}(X_{i_a},X_{j_a})
  \wedge\bigwedge_{b\in B}\Cover_{\rho_b}(U_b,X_{\ell_b}).
\]

$A$ and $B$ are finite index sets, and $i_a,j_a,\ell_b\in\{1,\ldots,p\}$.
Each target set $U_b$ is either $V(G)$ or one of the unary input sets $L_1,\ldots,L_q$. 
Finally, $\bar\delta=(\delta_a)_{a\in A}$ and $\bar\rho=(\rho_b)_{b\in B}$ are tuples of numerical parameter symbols whose values are nonnegative integers given with the input.
\end{definition}

The side formula may impose any fixed $\CMSOCTW$ condition, so we can require selected sets to induce paths, forests, connected subgraphs, or shortest paths.
\medskip

We next recall the nearest-label fragment.
\begin{definition}[Single-interface nearest-label formulas]
\label{def:nearest-label-fragment}
We assume that the unary input predicates $L_1,\ldots,L_q$ partition $V(G)$.
For a nonempty set $S$, the atom
\[
  \NearLab_i(v;S)
  \quad\Longleftrightarrow\quad
  \exists u\in S\cap L_i\colon d_G(v,u)=d_G(v,S)
\]
says that at least one closest selected vertex has label $i$, and we keep every label that attains the minimum when there is a tie.
A \emph{single-interface nearest-label formula} $\theta(S,\bar X)$ is a $\CMSOCTW$ formula in which we also allow atoms $\NearLab_i(v;S)$, all referring to the same distinguished free set $S$.
We require $S\neq\varnothing$, and no quantifier may bind this distinguished variable.
The additional free sets $\bar X$ are optional, and when they are absent we write $\theta(S)$.
For optimization and counting, we consider assignments to all displayed free sets.
\end{definition}

Finally, a \emph{one-cover formula} has the form
\[
  \Phi(\bar X;\rho)=\psi(\bar X)\wedge\Cover_\rho(U,S),
\]
where $\psi\in\CMSOCTW$, the target set $U$ is given by the input, and $S$ is one distinguished free set among $\bar X$.
Here the side formula, its free variables, and the radius $\rho$ may all vary with the input.
We count the numerical symbol in $|\Phi|$, but we do not count the binary value assigned to it.
Thus this fragment allows a varying side formula, one covering set, and no separation atoms.

We first give the proof for connected graphs.
Then every vertex has distance at most $n-1$ from every nonempty selected set, which lets us use finite integer potentials.
We then explain the simple changes needed for disconnected graphs after the proof.

Recall that for additive optimization, each free vertex-set variable $X_i$ has a weight function \mbox{$w_i:V(G)\to\mathbb Z$}, with the weights given in signed binary notation.
We minimize or maximize $\sum_i\sum_{v\in X_i}w_i(v)$, and we write $N$ for the total binary input length, including the formula, numerical bounds, and any weights.

\packingcoveringtheorem*

\onecovertheorem*

\nearestlabeltheorem*

\subsection{Distance profiles and local certificates}

We first prove a profile bound and local distance conditions that we will use in the joint proof of \Cref{thm:packing-covering,thm:nearest-label-optimization,thm:one-cover}.
The profile bound limits the number of table rows, and the local conditions let us check covering and nearest labels while we process the decomposition.

For $A\subseteq V(G)$ and an integer $D\geq0$, we put $\pi_A^D(v):=\min\{D,d_G(v,A)\}$, with $\pi_\varnothing^D(v)=D$.
We write $\pi_A^D|_B$ for the restriction of this function to a bag $B$.
Note that its value changes by at most one across each edge.

\begin{lemma}[Distance-profile compression]
\label{lem:dist-profile-compression}
If $G[B]$ is connected and $b:=|B|\geq1$, then
\begin{align*}
  \bigl|\{\pi_A^D|_B:A\subseteq V(G)\}\bigr|
  \leq (D+1)3^{b-1}.
\end{align*}
The same bound holds for all maps $p:B\to\{0,\ldots,D\}$ whose values differ by at most one across every edge of~$G[B]$.
These maps can be enumerated in time $(D+1)3^{b-1}(b+\log(D+2))^{O(1)}$.
\end{lemma}

\begin{proof}
We fix a rooted spanning tree of $G[B]$.
Once we choose the value at the root and the change along every tree edge directed away from it, we know the value at every vertex of the bag.
We have $D+1$ choices for the root value and three choices, $-1$, $0$, and $1$, for each of the~$b-1$ tree-edge increments.
We can enumerate these choices and discard a map if a value leaves the allowed range or if the map violates the Lipschitz condition on another bag edge.
\end{proof}

For covering and nearest labels, we need exact distances.
We can recognize these distances by local conditions: the labels change by at most one across an edge, and every positive label has a neighbor whose label is one smaller.

\begin{lemma}[Exact distance potentials]
\label{lem:exact-distance-potential}
Let $\lambda:V(G)\to\{0,\ldots,n-1\}$ and put $S:=\lambda^{-1}(0)$.
Then $S\neq\varnothing$ and $\lambda(v)=d_G(v,S)$ for every vertex~$v$ if and only if $|\lambda(u)-\lambda(v)|\leq1$ on every edge and every vertex with a positive label has a neighbor whose label is one smaller.
\end{lemma}

\begin{proof}
The distances to a nonempty set satisfy both conditions.
Conversely, from any vertex $v$ we can follow neighbors with decreasing labels until we reach $S$, which takes exactly $\lambda(v)$ steps.
This shows that $S$ is nonempty and that $d_G(v,S)\leq\lambda(v)$.
Since labels change by at most one along each edge, every path from~$v$ to a zero-labeled vertex has length at least~$\lambda(v)$.
\end{proof}

We can also check the nearest labels by a local rule.
For a nonempty set $S$, we put $h(v):=d_G(v,S)$ and write $R(v)\subseteq[q]$ for the labels of all closest selected vertices.
Then
\begin{equation}
  R(v)=
  \begin{cases}
    \{i\},&h(v)=0\text{ and }v\in L_i,\\[1mm]
    \displaystyle\bigcup_{\substack{u\in N_G(v)\\h(u)=h(v)-1}}R(u),&h(v)>0.
  \end{cases}
  \label{eq:nearest-label-recurrence}
\end{equation}
This is because a shortest path from~$v$ to a closest selected vertex begins with a neighbor of distance one less.
Conversely, when we start with such a neighbor, we can extend a shortest path from it to~$S$ by one edge and obtain a shortest path from $v$ to $S$.
The union therefore gives exactly the labels of the closest selected vertices, including every label in a tie.
Once we know $h$, this rule determines~$R$ uniquely, starting at distance zero and continuing in increasing order of distance.

\subsection{The dynamic program}

We now give one dynamic program for all three theorems, using the same decomposition and the same method for evaluating the logical formula.
The information that we store about distances depends on whether we consider separation, covering, or nearest labels.

For packing-covering templates, we store separation profiles, and for each distinct covering set we also store proposed distances to that set.
At each vertex of the current bag, we keep a bit indicating whether we have already introduced an edge to a neighbor whose proposed distance is one smaller.

For nearest labels, we also guess, for each vertex, which labels occur among its closest selected vertices.
Instead of the bit, we keep a set of labels that we collect as we introduce edges.
Whenever we introduce an edge to a neighbor whose proposed distance is one smaller, we add that neighbor's proposed nearest labels to this set.
In this way, we can later check whether the labels collected from these neighbors agree with the labels that we proposed.

We combine this distance information with logical types, which allow us to check arbitrary $\CMSOCTW$ side conditions and count satisfying assignments exactly.
\bigskip

By \Cref{thm:compute-nice-bag-connected-decomposition}, we first compute a nice bag-connected tree decomposition of width $w=\mathcal O(\ctw(G)^2)$ in fixed-parameter time.
We describe the dynamic program in terms of~$w$.

\paragraph{Preparing the decomposition.}
When separation atoms occur in the formula, we first compute all distances $d_G(u,v)$, which we later use to check separation and update the profiles.
The root and leaves of our nice decomposition have empty bags, and every other node introduces a vertex, introduces an edge, forgets a vertex, or joins two equal bags.
Each graph edge is introduced exactly once.
For a node $t$, we write $B_t$ for its bag and $V_t$ for the union of all bags in the subtree rooted at $t$, including $B_t$ itself.
We put $I_t:=V_t\setminus B_t$ and write $E_t$ for the set of edges introduced in this subtree.
The bag $B_t$ separates $I_t$ from $V(G)\setminus V_t$, and every edge incident with a vertex of $I_t$ already belongs to $E_t$.
Each vertex is forgotten exactly once.

\paragraph{Maintaining logical types.}
Since our bag-connected decomposition has width~$w$, we have $\cbl(G)\leq2w$.
Using $M:=\max\{3,2w+1\}$, \Cref{lem:guarded-translation} translates the side formula into an ordinary~$\CMSOtwo$ formula~$\psi^\ast$ of length bounded by a computable function of its original length and~$w$.
For a nearest-label formula $\psi(S,\bar X)$, we choose fresh unary predicates $N_1,\ldots,N_q$, replace every atom $\NearLab_i(v;S)$ by $N_i(v)$, and add the condition $S\neq\varnothing$. We call the resulting formula also $\psi^*$. 
Now the logical part of the dynamic program can work with ordinary unary predicates in all cases.
We use the usual dynamic program on a tree decomposition, keeping a logical type as one coordinate of each table row.
We refer to \cite[Chapter~6]{CourcelleEngelfriet2012} for background on this methodology.

Very briefly, the logical type is defined as follows.
Let $\ell$ be the quantifier rank of~$\psi^\ast$, and let~$\mathcal M$ be the finite set of counting moduli occurring in it.
At node~$t$, consider the incidence structure of the partial graph $(V_t,E_t)$, with its input predicates, its free-set assignments, and a name for each vertex of~$B_t$ that is chosen from a fixed set of $w+1$ boundary symbols.
Two such annotated partial graphs have the same \emph{logical type} if they satisfy the same $\CMSOtwo$ formulas of quantifier rank at most~$\ell$ using moduli from~$\mathcal M$, with the bag names and the fixed free-set assignments.
There are only finitely many types for the fixed signature, moduli, quantifier rank, and bag-size bound.
We store the numerical potentials separately, so their range does not increase the number of logical types.

The standard composition property of these types gives update rules for introducing a vertex, introducing an edge, forgetting a bag name, and joining two partial graphs along their named bag vertices.
Each rule determines the parent type from the child type or types and the local membership data and 
these finite update tables are effectively computable.
At the root of the decomposition, the type determines whether~$\psi^\ast$ holds.

\medskip

For the distance information we make two reductions.
First, we replace each separation threshold~$d$ by $\min\{d,n\}$ and each covering radius $\rho$ by $\min\{\rho,n-1\}$.
Since $G$ is connected and nonempty, every distance between two vertices is at most $n-1$, so these replacements preserve the separation and covering conditions.
An empty covering set covers exactly the empty target set for every finite radius, so the replacement is fine for this case. 
We remove separation atoms with threshold zero because they impose no restriction.

Second, finite distance potentials can describe distances only to nonempty selected sets.
We therefore consider separately every choice of which distinct covering sets are empty and which are nonempty.
If we declare $X_j$ empty, we check that every target set covered by $X_j$ is empty, and we discard the case if this fails.
Otherwise, we forbid membership in $X_j$ when introducing vertices and store no potential for this set.
For each set declared nonempty, we keep one potential and later check the local conditions that make it equal to the distance to that set.
Every complete assignment belongs to exactly one of these cases, so we can add their counts and compare their optimum values.
The number of cases is fixed for a fixed packing-covering template, and a one-cover formula has only two cases.
For nearest-label formulas, we already require the distinguished set $S$ to be nonempty, so we need no such case distinction.

\paragraph{The tables and their coordinates.}
We write $Y_1,\ldots,Y_r$ for all displayed free vertex-set variables of the formula under consideration.

At node $t$, we have finished all local checks at vertices in $I_t$, while vertices in the bag $B_t$ may still have edges introduced later.
Indeed, every neighbor of a vertex in $I_t$ belongs to $V_t$, and every edge incident with such a vertex belongs to $E_t$.
We therefore keep the information needed to complete the local checks at bag vertices as we process the remaining graph.

For packing-covering templates, we use a table
\[
  \mathcal T_t^{\mathrm{pc}}
  [\mu,\tau,\mathbf p,\boldsymbol\lambda,\mathbf a,\mathbf c],
\]
and for nearest-label formulas, we use a table
\[
  \mathcal T_t^{\mathrm{nl}}
  [\mu,\tau,h,R,A,\mathbf c].
\]
For a one-cover formula, we use the first table without separation profiles and with at most one covering potential.
Both tables have the following coordinates.
\begin{description}
  \item[$\mu$.]
  We use a map $\mu\colon B_t\to2^{[r]}$ to specify the free vertex-set memberships on the bag, where $i\in\mu(v)$ means that $v\in Y_i$.
  \item[$\tau$.]
  We store the logical type of $(V_t,E_t)$ with its input predicates, chosen set memberships, and named bag vertices.
  These names identify the vertices shared with the rest of the graph.
  \item[$\mathbf c$.]
  If cardinality constraints are present, we use one counter for each set whose size we need to check.
  Its value is the number of selected vertices of that set in $I_t$, so we count only vertices that have already been forgotten.
  We omit these counters when there are no cardinality constraints.
\end{description}
We keep the bag memberships explicitly as well as in the logical type, since at a join we need to check that both children make the same choices on their shared bag.
Because the counters exclude bag vertices, we can add the child counters without counting any vertex twice.

For each set $X_i$ that occurs in a remaining separation atom, we let $D_i$ be the largest threshold involving that set.
We keep its profile
\[
  p_i(b):=\min\{D_i,d_G(b,X_i\cap I_t)\}
  \qquad(b\in B_t).
\]
Thus $\mathbf p$ gives the distances from bag vertices to selected interior vertices, with every value above~$D_i$ replaced by $D_i$.
We measure these distances in the original graph $G$, so a shortest path used in this definition may leave $V_t$.
The map $\mu$ separately gives the selected vertices in the bag, since bag vertices and interior vertices have different roles at a join.

For every covering set $X_j$ declared nonempty, we keep two maps on $B_t$.
\begin{description}
  \item[$\lambda_j$.]
  We use $\lambda_j(v)\in\{0,\ldots,n-1\}$ as the proposed distance from $v$ to the final selected set $X_j$.
  We require $\lambda_j(v)=0$ exactly when $\mu(v)$ places $v$ in $X_j$.
  \item[$a_j$.]
  We use a bit $a_j(v)$ to indicate whether an introduced edge already joins $v$ to a vertex whose potential is one smaller, so
  \[
    a_j(v)=1
    \quad\Longleftrightarrow\quad
    \exists u\colon uv\in E_t
    \text{ and }\lambda_j(u)=\lambda_j(v)-1.
  \]
\end{description}
This definition refers to the potential values of the partial assignment, including values at forgotten vertices, but for a bag vertex we store only the resulting bit.
The bit tells us whether at least one suitable neighbor exists, so several such neighbors do not give several choices.

For nearest labels, we use the following data instead.
\begin{description}
  \item[$h$.]
  We use $h\colon B_t\to\{0,\ldots,n-1\}$ as the proposed distance to $S$, with $h(v)=0$ exactly when $v\in S$.
  \item[$R$.]
  At each $v\in B_t$, we propose a nonempty set $R(v)\subseteq[q]$ of labels that occur among its closest selected vertices.
  If $h(v)=0$ and $v\in L_i$, we require $R(v)=\{i\}$.
  In the structure used for the logical type, we let $N_i(v)$ hold exactly when $i\in R(v)$.
  \item[$A$.]
  We use $A(v)\subseteq R(v)$ for the labels already collected from neighbors whose potential is one smaller, considering only introduced edges:
  \[
    A(v):=\bigcup_{\substack{uv\in E_t\\h(u)=h(v)-1}}R(u).
  \]
\end{description}
Thus $R(v)$ is the label set that we expect in the complete solution, while $A(v)$ is the part already obtained through processed edges.
When $h(v)>0$, we check $A(v)=R(v)$ before forgetting~$v$.
The potentials $\lambda_j$ and $h$ describe proposed distances to the final selected sets, whereas the separation profiles describe distances only to selected vertices in $I_t$.
We choose the potential values when processing the bag vertices and then check the local conditions that make these values correct.

We require every stored profile and potential to be $1$-Lipschitz on $G[B_t]$.
We check this condition on all edges induced by the bag, including edges that have not yet been introduced.
Together with the connectedness of the bag, this condition gives the bound on the number of possible rows.

\paragraph{What an entry represents.}
We now define the partial objects represented by a row, so that we can state precisely what the table stores.
A partial object contains the chosen free vertex sets on $V_t$, any chosen free edge sets on $E_t$, and proposed metric data on all of $V_t$.
For packing-covering, these metric data are maps $\lambda_j\colon V_t\to\{0,\ldots,n-1\}$ for the covering sets declared nonempty.
For nearest labels, they are a map $h\colon V_t\to\{0,\ldots,n-1\}$ and nonempty sets $R(v)\subseteq[q]$ for all $v\in V_t$.
We use data on the whole processed vertex set only to define the meaning of a row, the algorithm itself stores only the bag coordinates listed above.

For a row index $\xi$, we write $\mathcal F_t(\xi)$ for the set of partial objects that satisfy the following conditions.
\begin{enumerate}[label=(\roman*)]
  \item The memberships, potentials, and nearest-label sets on $B_t$ agree with the corresponding coordinates of $\xi$.
  The coordinate $\tau$ is the logical type of this partial graph with its chosen memberships and labels, and $\mathbf c$ gives the required cardinalities in $I_t$.
  A covering set declared empty contains no vertex of $V_t$.
  \item In the packing-covering case, every separation atom holds for all relevant pairs of distinct vertices in $V_t$, with distances measured in $G$.
  Each profile is exactly the truncated distance to the corresponding selected interior set, as defined above.
  \item Every potential is $1$-Lipschitz on $G[V_t]$, and its zero set is exactly the corresponding selected set restricted to $V_t$.
  For a covering potential $\lambda_j$, every vertex $v\in I_t$ with $\lambda_j(v)>0$ has a neighbor $u$ with $\lambda_j(u)=\lambda_j(v)-1$.
  Every vertex $v\in I_t$ in a target set covered by $X_j$ satisfies the corresponding bound $\lambda_j(v)\leq\rho_b$.
  The bits $a_j$ have exactly the meanings stated above.
  \item In the nearest-label case, the proposed label set at a zero-potential vertex is the singleton containing its input label.
  For every $v\in V_t$, the union of $R(u)$ over edges $uv\in E_t$ with $h(u)=h(v)-1$ is contained in $R(v)$.
  For every $v\in I_t$ with $h(v)>0$, this union equals~$R(v)$.
  At a bag vertex, the union is exactly the stored set $A(v)$.
\end{enumerate}

We omit a condition when the table has no corresponding coordinate.
All local checks are therefore complete on $I_t$, while the bits and collected label sets allow us to finish the checks at bag vertices later.
We can also check the Lipschitz condition on all of $G[V_t]$ as we proceed, since every edge with an interior endpoint has already been introduced and every other edge lies within the bag.

We describe minimization first, and obtain maximization by reversing the cost comparisons.
For a partial object $\mathfrak a$, we define its cost by
\[
  W_t(\mathfrak a)
  :=\sum_{i=1}^r\sum_{v\in Y_i\cap I_t}w_i(v).
\]
We therefore charge vertex weights only after the vertices have been forgotten.
If free edge sets also have weights, we add their contributions on $E_t$, charging each edge when it is introduced.
Each row stores the triple
\begin{align*}
  \mathcal T_t[\xi]
  &:=\bigl(C_t[\xi],O_t[\xi],K_t[\xi]\bigr),\\
  C_t[\xi]
  &:=|\mathcal F_t(\xi)|,\\
  O_t[\xi]
  &:=\min_{\mathfrak a\in\mathcal F_t(\xi)}W_t(\mathfrak a),\\
  K_t[\xi]
  &:=\bigl|\{\mathfrak a\in\mathcal F_t(\xi):
       W_t(\mathfrak a)=O_t[\xi]\}\bigr|.
\end{align*}
Here $C_t$ counts all represented partial objects, $O_t$ is their smallest cost, and $K_t$ counts the objects with that cost.
We give an empty row the value $(0,+\infty,0)$.
When we only need feasibility or counting, we set all weights to zero, and a row is feasible exactly when $C_t[\xi]>0$.

The same partial choice of free sets may have several proposed metric certificates, and we treat them as different objects in $\mathcal F_t(\xi)$.
This is necessary because the final distances may depend on selected vertices outside $V_t$.
At the root, however, the local checks will force the unique correct certificate for each satisfying assignment (as we prove below).

\paragraph{Combining contributions to a row.}
We initialize each row with $(0,+\infty,0)$.
Suppose that a transition contributes $(C,O,K)$ to a row that already contains $(C',O',K')$.
We add the total counts, keep the smaller cost, and add the optimum counts when the costs are equal.
Thus the new entry is
\begin{align*}
  \widetilde C&=C+C',
  \qquad \widetilde O=\min\{O,O'\},\\
  \widetilde K&=
  \begin{cases}
    K,&O<O',\\
    K',&O'<O,\\
    K+K',&O=O'.
  \end{cases}
\end{align*}

We use this rule for every valid contribution and ignore child rows whose count is zero.
At a node with one child $s$, a valid transition with local cost $\Delta$ contributes
\[
  \bigl(C_s[\xi_s],O_s[\xi_s]+\Delta,K_s[\xi_s]\bigr)
\]
to the resulting parent row.
At a join, a valid pair of child rows contributes
\[
  \bigl(
    C_{t_1}[\xi_1]C_{t_2}[\xi_2],
    O_{t_1}[\xi_1]+O_{t_2}[\xi_2],
    K_{t_1}[\xi_1]K_{t_2}[\xi_2]
  \bigr).
\]

Counts are multiplied because we combine independent choices from the two subtrees, and we add their costs because neither child has charged a shared bag vertex.
We next describe the validity tests and the resulting parent coordinates for each kind of node.

\paragraph{Empty leaves.}
At an empty leaf, we have $V_t=E_t=\varnothing$, so there is one empty assignment and one empty certificate.
All bag maps are empty, all counters are zero, and the logical type is the type $\tau_0$ of the empty structure.
We give this row the value $(1,0,1)$ and every other row the value $(0,+\infty,0)$.

\paragraph{Introducing a vertex.}
We let $s$ be the child of $t$, where $B_t=B_s\cup\{v\}$ and $v\notin V_s$.
The new vertex has no neighbor in $I_s$, and no edge incident with $v$ belongs to $E_s$.
We can therefore choose its bag data now, but an introduced edge cannot yet show that it has a neighbor of smaller potential.
For each child row with positive count, we try all possible memberships $\mu(v)$ and potential values at~$v$.
In the nearest-label case, we also try all nonempty sets $R(v)$.
We reject a choice if it puts $v$ in a set declared empty, violates the condition on a zero set, or violates the Lipschitz condition on an edge between $v$ and $B_s$.
If $h(v)=0$ and $v\in L_i$, we also require $R(v)=\{i\}$.

For each covering potential, we set $a_j(v):=0$, and for nearest labels, we set $A(v):=\varnothing$.
We keep the existing potentials, bits, and label sets on $B_s$.
We compute the new logical type from the child type, the input predicates at $v$, and the memberships chosen for $v$.
In the nearest-label case, these memberships include $N_i(v)$ exactly for the indices $i\in R(v)$.
The counters do not change and the local cost is $\Delta=0$, since we have not yet forgotten $v$.

For the separation profiles, we have $I_t=I_s$, so we keep the old values and compute only the value at $v$.
Since $B_s$ separates $v$ from $I_s$, we set
\[
  p_i(v)
  :=\min\left\{D_i,
       \min_{b\in B_s}\bigl(d_G(v,b)+p_i(b)\bigr)
     \right\}.
\]

As before, a minimum over an empty set is $\infty$.
If the selected interior set is nonempty, every path from $v$ to that set meets $B_s$, which gives one inequality in the corresponding equality for exact distances.
For the other inequality, we concatenate shortest paths through a bag vertex.
Replacing an old distance above $D_i$ by $D_i$ does not change the new value below $D_i$, so the same formula works for the truncated profiles.
If the selected interior set is empty, the formula gives~$D_i$, as required.

Only pairs involving the new vertex can now violate separation.
For an atom $\Sep_d(X_i,X_j)$ with $v\in X_i$, we require $p_j(v)\geq d$ and $d_G(v,b)\geq d$ for every $b\in X_j\cap B_s$.
If $v\in X_j$, we make the same checks with $X_i$ and $X_j$ exchanged.
The profile test checks all relevant interior vertices at once, while the other tests check the selected bag vertices one by one.
We never compare $v$ with itself, so we allow it to belong to both sets.
If all checks pass, the child triple contributes unchanged to the resulting parent row.

\paragraph{Introducing an edge.}
We let $s$ be the child of $t$ and suppose that $t$ introduces the edge $uv$.
The vertex set and bag do not change, and we have $E_t=E_s\cup\{uv\}$.
We use the new edge to update the information about neighbors of smaller potential at its endpoints.
For each covering potential, we set
\begin{align*}
  a_j^t(v)&:=a_j^s(v)\vee
       [\lambda_j(u)=\lambda_j(v)-1],\\
  a_j^t(u)&:=a_j^s(u)\vee
       [\lambda_j(v)=\lambda_j(u)-1],
\end{align*}
where a bracketed condition has value one when it holds and zero otherwise.
We leave the other bits unchanged.
For nearest labels, if $h(u)=h(v)-1$, we replace $A(v)$ by $A(v)\cup R(u)$, and if $h(v)=h(u)-1$, we make the corresponding update at $u$.
If the potentials are equal, neither collected label set changes.
We reject a row if an updated set is not contained in the corresponding proposed set $R$.
These rules include all relevant information from the edge, so we do not choose one particular neighbor as a witness.

All memberships, separation profiles, potentials, proposed label sets, and counters otherwise stay the same.
The Lipschitz condition at the endpoints has already been checked because both vertices lie in the bag.
We now try all free edge-set memberships, if any, and compute the new logical type by the edge-introduction rule.
Each choice contributes the child triple with local cost zero, unless free edge sets have weights.
In that case, we add the weights of the chosen edge-set memberships as the local cost.

\paragraph{Forgetting a vertex.}
We let $s$ be the child of $t$, where $B_s=B_t\cup\{v\}$.
We now have $V_t=V_s$ and $I_t=I_s\cup\{v\}$.
Since every edge incident with $v$ has been introduced, we can finish all local checks at this vertex.
For each covering potential with $\lambda_j(v)>0$, we require $a_j(v)=1$.
For each target set $U_b$ covered by $X_j$, we also require $\lambda_j(v)\leq\rho_b$ whenever $v\in U_b$.
For nearest labels, we require $A(v)=R(v)$ when $h(v)>0$.
This equality checks the complete nearest-label recurrence at $v$, and it also guarantees a neighbor of smaller potential because $R(v)$ is nonempty.
At potential zero, we have already checked the selected-set membership and the singleton input label, and we require no neighbor of smaller potential.

The vertex set $V_t$ has not changed, so we need no new separation tests.
We do, however, update the profiles because $v$ now belongs to the interior.
For every remaining bag vertex $b$, we set
\[
  p_i^t(b):=
  \begin{cases}
    \min\{p_i^s(b),d_G(b,v)\},&v\in X_i,\\
    p_i^s(b),&v\notin X_i.
  \end{cases}
\]

This gives exactly the distance to the enlarged selected interior set, truncated at $D_i$.
We restrict all other bag maps to $B_t$ and obtain the new logical type by removing the boundary name of $v$, while keeping the vertex and its memberships in the partial graph.

For each set $Y_i$ whose size we track, we increase its counter by one if $i\in\mu(v)$.
We also charge the local selection cost
\[
  \Delta(v):=\sum_{i\in\mu(v)}w_i(v),
\]
so the parent row receives $(C_s[\xi_s],O_s[\xi_s]+\Delta(v),K_s[\xi_s])$.
Because we forget $v$ exactly once, we include each of its weight and cardinality contributions exactly once.

\paragraph{Joining two subtrees.}
We let $t_1,t_2$ be the children of $t$ and write $B:=B_t=B_{t_1}=B_{t_2}$ for their common bag.
The two processed vertex sets intersect exactly in $B$, their interiors are disjoint, and their introduced-edge sets are disjoint.
We combine two rows only when their memberships $\mu$ agree on $B$ and their proposed global metric data also agree there.
For packing-covering, we require agreement of every $\lambda_j$, while for nearest labels, we require agreement of $h$ and $R$.
The profiles, predecessor bits, and collected label sets may differ because they describe information from different subtrees.

For covering, we set $a_j^t(v):=a_j^{t_1}(v)\vee a_j^{t_2}(v)$ for every $v\in B$.
For nearest labels, we set $A_t(v):=A_{t_1}(v)\cup A_{t_2}(v)$.
These updates combine the information from the two introduced-edge sets.
The union is contained in the common proposed set $R(v)$ because both child sets are contained in it.
We keep the common memberships, potentials, and proposed label sets, add the counters, and compute the parent logical type by joining vertices with the same boundary name.

It remains to check separation between selected vertices in different interiors.
The bag $B$ separates these interiors in the whole graph, including paths that leave $V_t$.
For $Z_1\subseteq I_{t_1}$ and $Z_2\subseteq I_{t_2}$, we therefore have
\[
  d_G(Z_1,Z_2)
  =\min_{b\in B}\bigl(d_G(Z_1,b)+d_G(b,Z_2)\bigr).
\]

Every path between the sets meets $B$, which gives one inequality, and concatenating shortest paths to and from a bag vertex gives the other.
If either set is empty, both sides are $\infty$.
If $B$ is empty, the connectedness of $G$ ensures that at least one child interior is empty, so both sides are again $\infty$.

For each atom $\Sep_d(X_i,X_j)$, we hence require
\begin{align*}
  \min_{b\in B}
    \bigl(p_i^{t_1}(b)+p_j^{t_2}(b)\bigr)&\geq d,\\
  \min_{b\in B}
    \bigl(p_j^{t_1}(b)+p_i^{t_2}(b)\bigr)&\geq d.
\end{align*}

Truncation does not change these tests, since $D_i,D_j\geq d$ and any sum below $d$ uses two distances that were not truncated.
The interiors are disjoint, so every tested pair consists of distinct vertices.
Pairs within one child and pairs involving a bag vertex have already been checked.
We obtain the parent profile by
\[
  p_i^t(b):=\min\{p_i^{t_1}(b),p_i^{t_2}(b)\}
  \qquad(b\in B),
\]
since the selected parent interior is the union of the selected child interiors.

Two compatible partial objects have a unique union because they agree on all set memberships and proposed metric data on $B$.
Conversely, restricting a parent object to the two subtrees gives a unique pair of compatible child objects.
The separation tests above check all new pairs between the interiors.
The local potential conditions at an interior vertex do not change, since it has no neighbor in the other interior or outside its own processed vertex set.
We therefore obtain every valid parent object from exactly one compatible pair of child objects.
The costs add because no bag vertex has been charged and the introduced-edge sets are disjoint, so the contribution is
\[
  \bigl(
    C_{t_1}[\xi_1]C_{t_2}[\xi_2],
    O_{t_1}[\xi_1]+O_{t_2}[\xi_2],
    K_{t_1}[\xi_1]K_{t_2}[\xi_2]
  \bigr).
\]

\paragraph{Correctness.}
At a leaf, the table has the required meaning because there is exactly one empty object.
When we introduce a vertex or an edge, each parent object determines its child object and the unique choice of newly introduced data.
Our tests check exactly the conditions that become relevant at that step.
When we forget a vertex, the partial object itself stays the same, but we complete its local checks and charge its cost.
The object also determines the updated profiles and counters.
At a join, the correspondence proved above shows that every compatible pair gives a valid parent object and every parent object occurs once.
By induction over the decomposition, $C_t$, $O_t$, and $K_t$ therefore have the meanings stated at every node.
All later compatibility tests depend only on the row index, so partial objects in the same row have the same possible extensions.
We may thus keep only their best cost and the number of objects with that cost when computing optimum values and counts.

\paragraph{Reading the answer at the root.}
We let $t_0$ be the empty root, where $V_{t_0}=I_{t_0}=V(G)$ and $E_{t_0}=E(G)$.
We accept exactly the rows whose logical type satisfies $\psi^\ast$ and whose counters satisfy the required cardinality bounds.
The boundary maps are empty, but all their local conditions have now been checked on the whole graph.

For each covering set declared nonempty, the candidate potential is $1$-Lipschitz, its zero set is~$X_j$, and every positive-potential vertex has a neighbor whose potential is one smaller.
Following such neighbors must reach a zero, so $X_j$ is indeed nonempty.
By \Cref{lem:exact-distance-potential}, we then have $\lambda_j(v)=d_G(v,X_j)$ for every vertex.
The checked covering inequalities are therefore exactly the required covering conditions, and all separation conditions hold by the table invariant.

For nearest labels, the nonemptiness of $R(v)$ and the equality $A(v)=R(v)$ at positive-potential vertices guarantee a neighbor of smaller potential.
The same potential lemma gives $h(v)=d_G(v,S)$ for the nonempty set $S$.
The label conditions are exactly \eqref{eq:nearest-label-recurrence}, including the singleton label at distance zero.
Induction on $h(v)$ therefore shows that $R(v)$ contains exactly the labels of all closest selected vertices.
Thus the predicates $N_i$ used in the logical type have the intended meanings, including every label involved in a tie.

Conversely, a satisfying assignment gives an accepted root object when we use its true distances and nearest-label sets.
The assignment determines these data uniquely, and it also determines the profiles, predecessor bits, collected label sets, counters, and logical types at all nodes.
Each satisfying assignment therefore contributes exactly one root object, even though its restrictions may have had several proposed certificates earlier in the computation.
Introducing the same boundary vertex in different branches creates no extra choices at the root because the data must agree when those branches are joined.
We also introduce no counting choices for quantified logical witnesses or for particular neighbors of smaller potential.

We obtain the number of satisfying assignments by summing $C_{t_0}[\xi]$ over all accepted rows and all disjoint choices of empty covering sets.
The instance is feasible exactly when this sum is positive.
For a feasible instance, the smallest value of $O_{t_0}[\xi]$ is the optimum, and we count optimum assignments by summing $K_{t_0}[\xi]$ over the rows with that value.
If we store a pointer to one transition giving each optimum, we can follow these pointers backwards to recover an optimum assignment.
For maximization, we use the largest cost instead of the smallest and give an empty row cost $-\infty$.

\paragraph{The number of rows and the running time.}
Every nonempty bag has size $b\leq w+1$.
We generate the possible profile and potential maps using \Cref{lem:dist-profile-compression}, discard choices that fail the required conditions, and keep only rows reached by the transitions.
It is enough to generate all $1$-Lipschitz maps in the relevant range, since the transitions ensure the more specific meanings of the profiles and potentials.
We never enumerate subsets of $V_t$ or the full partial certificates used to define $\mathcal F_t$.

For a fixed packing-covering template, we let $s$ be the number of separation profiles and~$H$ the number of covering sets declared nonempty in the current case.
The profiles give at most $(n+1)^s3^{s(b-1)}$ choices, and the covering potentials together with their bits give at most $n^H3^{H(b-1)}2^{Hb}$ choices.
The membership map has $2^{rb}$ possibilities, while the number of logical types depends only on $\Phi$ and $w$.
If we track $k$ cardinalities, the counters give at most $(n+1)^k$ choices.
We therefore have at most
\[
  f_\Phi(w)(n+1)^{s+H+k}
\]
rows, where we include the fixed choice of counters in the dependence of $f_\Phi$.
The numbers $s$, $H$, and~$k$ are fixed here, so the polynomial exponent depends only on the template and that fixed choice of counters.

For nearest labels, there are at most $n3^{b-1}$ potential maps.
For each vertex-label pair, the label is either absent from $R(v)$, present in $R(v)$ but absent from $A(v)$, or present in both.
We therefore have at most $3^{qb}$ choices for $R$ and $A$, and at most $n3^{b-1}3^{qb}$ choices for all metric data.
We can include the dependence on $q$, the free-set memberships, and the logical formula in a function of $|\psi|+q+w$.
If we impose cardinality constraints on $k$ selected sets, we keep one counter for each of these sets, counting its vertices that have already been forgotten.
Each counter has at most $n+1$ possible values, so we obtain at most $f(|\psi|+q+w)(n+1)^{k+1}$ rows.
Without cardinality constraints, we use no counters and take $k=0$.

For a one-cover formula, we use no separation profile and only one covering potential, or no potential when the covering set is empty.
In the nonempty case, there are at most $n3^{b-1}2^b$ choices for the metric data.
With $k$ fixed counters, we therefore have at most $f(|\Phi|+w)(n+1)^{k+1}$ rows.
An empty bag has just one choice of empty metric maps in every case.

At a join, we may check every pair of child rows, while at a node with one child, we try all permitted new memberships and metric data and compute the resulting parent row.
These operations take a polynomial number of steps in the table sizes, with an additional factor depending on the formula and $w$ for the logical-type updates.
For nearest-label and one-cover formulas, the polynomial exponent is independent of the formula and graph parameters when the number of extra cardinality counters is fixed.
For a fixed packing-covering template, the exponent may also depend on its numbers of profiles and covering sets.

We perform all counting with exact integers.
A partial object is described by its free-set memberships and, at each vertex, a bounded number of potential values and finite label sets, where the number of coordinates depends on the formula.
Each finite potential value needs $\mathcal O(\log(n+1))$ bits, so the logarithm of the total number of partial objects is polynomial in the binary input length~$N$.
The counts therefore have polynomial bit length, and the same holds for sums of the binary selection weights.

Before processing a cardinality condition $\ell\leq|Y_i|\leq u$, we intersect its allowed interval with~$[0,n]$ and reject the instance if the intersection is empty.
An equality is the case $\ell=u$, and we handle a one-sided bound in the same way.
Reading and reducing the binary numerical bounds takes $O(N)$ time, so a running-time bound in terms of $n$ still includes this initial term when all selection weights are zero or one.

Finally, we computed the nice decomposition in fixed-parameter time with width $w=O(\ctw(G)^2)$, and its number of nodes is bounded by a function of $\ctw(G)$ times a fixed polynomial in $n$.
Substituting these bounds into the table analysis gives the running times in our three theorems \Cref{thm:packing-covering,thm:nearest-label-optimization,thm:one-cover}.
$\hfill\square$

\medskip

\paragraph{Disconnected graphs.}
It is not difficult to handle also disconnected graphs. We simply 
also allow a potential to have value $\infty$.
Across an edge, we require either two finite values differing by at most one or two infinite values.
At every vertex, including those with infinite potential, we require potential zero if and only if the vertex belongs to the corresponding selected set.
We require a neighbor of potential one less only at vertices with positive finite potential.
These conditions force the potential to equal the distance to the selected set in every component meeting that set, and to equal $\infty$ in every other component.

Since each nonempty bag is connected, its potential map is either one of the finite maps already considered or the single all-infinite map.
Thus the number of states increases by at most a constant factor for each potential.
For every covering set declared nonempty, we also add its nonemptiness to the side formula and reject an infinite potential at any vertex of a corresponding target set.
The bounds can still be reduced as above, since every finite graph distance is at most~$n-1$ and every input threshold or radius is finite.
The separation profiles and their updates continue to use $\min\{D_i,\infty\}=D_i$.
At a join with an empty bag, there is no path between the two interiors, so the separation equality still has value $\infty$ on both sides even if both interiors are nonempty.

For nearest labels, a vertex in a component disjoint from $S$ has distance $\infty$ to every selected vertex, so its nearest-label set contains all labels occurring in $S$.
We consider each nonempty set $Q\subseteq[q]$ and add conditions to the side formula requiring
\[
  S\cap L_i\neq\varnothing\quad(i\in Q),
  \qquad
  S\cap L_i=\varnothing\quad(i\notin Q).
\]
We also keep $S\neq\varnothing$ in the side formula.
At every vertex with $h(v)=\infty$, we require $R(v)=Q$ and $A(v)=\varnothing$.
For an infinite covering potential, we similarly fix the predecessor bit at zero.
We perform the decreasing-neighbor updates and recurrence tests only at finite-potential vertices, so we never use an expression such as $\infty-1$.
There are at most $2^q$ choices of $Q$, and each nonempty selected set determines exactly one of them, so adding their counts preserves the number of satisfying assignments.

\subsection{Nearest-set objectives}
\label{sec:nearest-set-objectives}

In nearest-label evaluation, we ask which labels occur among the selected vertices closest to a given vertex, allowing several labels when there is a tie.
We now consider objectives that depend on the distance itself, with each vertex contributing a cost determined by its distance to the selected set.
Such objectives occur in classical network location problems~\cite{Hakimi1964,KarivHakimi1979}.
For example, in the $p$-median problem, we choose $p$ facilities and minimize the total weighted distance from the clients to their nearest facilities
and in the $p$-means problem, we choose $p$ facilities and minimize the total weighted squared distance from the clients to their nearest facilities.
In uncapacitated facility location, we also pay an opening cost for every facility that we choose, and we minimize the sum of the opening costs and the service costs without prescribing the number of facilities.

Both problems, i.e., $p$-median and $p$-means are known to be W[1]-hard parameterized by treewidth using 
a folklore reduction from the $k$-Center problem, which is shown in~\cite{KatsikarelisLampisPaschos2019}. Moreover,
both are known to be in XP parameterized by treewidth (see~\cite[Section~5.2]{ChakrabartyConroySarkar2025} for $p$-median and \cite{ChakrabartyConroySarkar2025} for $p$-means).

Here we work with unweighted graphs and allow an arbitrary cost function of the nearest distance at each vertex, including functions that are not monotone and values that forbid particular distances.
We obtain exact fixed-parameter algorithms in bag-connected treewidth, while the number of selected vertices and the cost functions remain part of the input.
We may also impose logical conditions on the selected set and count feasible or optimum sets.
Since the cost of a vertex depends on its distance to the selected set, we cannot express this objective using only fixed weights on the selected vertices.
We therefore keep the distance costs outside the logical language and handle them directly in the dynamic program from \Cref{sec:dist-meta}.
The tables already store the distance to the selected set, so we only need to add the corresponding cost when we forget a vertex. This is just a small modification of our earlier dynamic program.

In particular, we obtain fixed-parameter algorithms for
$p$-median and $p$-means on unweighted graphs, parameterized
by bag-connected treewidth, while the number $p$ of facilities remains
part of the input.
To the best of our knowledge, these parameterized tractability
results were not previously known.

An instance of \textsc{Nearest-Set Optimization} consists of a connected nonempty unweighted graph $G$, a candidate set $F\subseteq V(G)$, a rational opening cost $o_s$ for each $s\in F$, and a cost table $h_v\colon\{0,\ldots,n-1\}\to\mathbb Q\cup\{+\infty\}$ for each vertex $v$, where $n:=|V(G)|$.
All finite costs are rational numbers written in binary, and the total input length $N$ includes the opening costs and the cost tables, or their chosen representation as described below.
We seek a nonempty set $S\subseteq F$ that minimizes
\[
  \sum_{s\in S}o_s+
  \sum_{v\in V(G)}h_v\bigl(d_G(v,S)\bigr).
\]
We treat $+\infty$ as a forbidden distance in both the minimization and maximization variants, and we may also require $|S|=p$ or $|S|\leq p$ for a nonnegative integer $p$ given with the input.
Instead of explicit tables, we may use any fixed representation from which all table entries can be computed in time polynomial in the input length.

\begin{theorem}\label{thm:nearest-set-optimization}
We can solve \textsc{Nearest-Set Optimization}, its maximization variant, and exact counting of feasible or optimum sets in fixed-parameter time with parameter $\ctw(G)$.
We may also impose a $\CMSOCTW$ condition $\psi(S)$, with running time $f(|\psi|+\ctw(G))N^{O(1)}$ for a computable function $f$.
\end{theorem}

\begin{proof}
We proceed exactly as in \Cref{sec:dist-meta} and use the exact-potential states from the covering part of \Cref{thm:packing-covering}.
We take the zero set of the potential as $S$ and allow a zero value only at a candidate vertex.
When we forget a vertex $v$, we reject a forbidden distance and add the cost $h_v(\lambda(v))$, together with $o_v$ if $\lambda(v)=0$.
If there is a cardinality bound, we also count the forgotten vertices with potential zero.
At a join, we add the child counters and costs because the two interiors are disjoint, and at the root, we check the final cardinality bound.

The local potential conditions force a nonempty zero set, and by \Cref{lem:exact-distance-potential}, every accepted labeling is exactly $\lambda(v)=d_G(v,S)$.
Conversely, each feasible set determines this labeling and all predecessor bits uniquely.
It therefore gives exactly one accepted root object, so the extra distance data do not increase its count.
On a nonempty connected bag of size $b$, we have at most $n3^{b-1}2^b$ potential-and-bit states, while an empty bag has one empty state.
The cardinality counter has at most $n+1$ values, so the dynamic program takes $2^{O(w)}N^{O(1)}$ time even if we check all pairs of child states at each join.
Adding and comparing the binary rational costs takes polynomial time as well.

To count feasible sets, we add the counts of alternatives and multiply the counts of compatible child rows at joins in the straight forward way.
\end{proof}

\section{Logical boundaries and computational hardness}
\label{sec:guarded-logic}
\label{sec:exact-dependence}
\label{sec:encoding-hardness}

We now show that the structural conditions in our guarded meta-theorems hold exactly for the graph classes on which the corresponding logics are monadically dependent.
To prove this, we represent each vertex of an arbitrary graph $H$ by a vertex at an even position along a sufficiently long path.
Between any two distinct representatives, we therefore have a path vertex exactly halfway between them, at equal distance from both.
We call this vertex their midpoint, and we choose the positions so that different pairs of representatives have different midpoints.
We then give a midpoint a special label exactly when its two representatives correspond to adjacent vertices of $H$.
This label is a unary input predicate, which allows us to recover the edges of $H$ by testing whether a labeled vertex is at equal distance from two representatives.

We use this construction both for unguarded distance comparisons on a shortest path and for guarded equidistance on an arc shorter than half of a geodesic cycle.
We first prove the results on monadic dependence, and then use the same constructions to prove hardness in the cases where we can build the encodings in polynomial time.
For the result in which a distance threshold is part of the input, we use a separate construction based on a Turing-machine computation, which we present at the end of this section.

The hardness proofs were generated by ChatGPT 6.0 Astra and were revised and verified by the authors; see the \hyperref[sec:ai-methodology]{AI methodology statement} at the end of the paper.

\subsection{Monadic interpretations}

We use transductions in which each output vertex is one input element, so we do not use copies or tuples of input elements.
In a transduction we first add a fixed finite number of unary predicates to the input structure, where we may choose their sets of elements freely.
We call this a \emph{monadic lift}, and we can think of the new predicates as additional colors.
We then define the output graph by two fixed formulas: a domain formula $\nu(x)$, which selects its vertices, and an adjacency formula~$\eta(x,y)$, which defines its edges.
We require the adjacency formula to be symmetric and irreflexive, so the output is a simple undirected graph.
For vertex-set logics, we use the adjacency structure of~$G$ as the input, while for incidence logics we use~$\mathcal I(G)$ and may color both vertex elements and edge elements.
We always evaluate metric atoms in the original graph~$G$, even when the domain formula selects only some of its vertices.
For further background on interpretations, we refer to Courcelle and Engelfriet and to the recent survey by Pilipczuk~\cite{CourcelleEngelfriet2012,Pilipczuk2025}.

\begin{definition}[Monadic dependence]
\label{def:monadic-interpretation}
\label{def:monadic-dependence}
A graph class~$\mathscr C$ is \emph{monadically $\mathcal L$-dependent} if there is no fixed $\mathcal L$-interpretation that produces every finite graph, up to isomorphism, from monadic lifts of members of~$\mathscr C$ in the appropriate encoding.
Otherwise, we call~$\mathscr C$ \emph{monadically $\mathcal L$-independent}.
\end{definition}

We use the following classical characterizations for the underlying second-order logics.
These characterizations also hold with our copyless transductions; see~\cite[Theorems~32 and~33]{Pilipczuk2025}.

\begin{theorem}
\label{thm:incidence-transduction}
\label{thm:classical-dependence}
For every graph class~$\mathscr C$, the following statements hold.
\begin{enumerate}[label=(\roman*)]
  \item The class~$\mathscr C$ has bounded treewidth if and only if its incidence structures are monadically $\MSOtwo$-dependent, equivalently monadically $\CMSOtwo$-dependent.
  \item The class~$\mathscr C$ has bounded cliquewidth if and only if its adjacency structures are monadically $\CMSO$-dependent.
\end{enumerate}
\end{theorem}

The first characterization follows from the Excluded Grid Theorem and Seese's interpretation argument~\cite{RobertsonSeymour1986,Seese1991}.
The second characterization is due to Courcelle and Oum~\cite{CourcelleOum2007}.
For the directions from bounded width to dependence, we use the fact that fixed counting monadic interpretations preserve bounded cliquewidth.
This also applies when the input structures are incidence structures of graphs of bounded treewidth~\cite{CourcelleEngelfriet2012}.

\subsection{Encodings in geodesics and geodesic cycles}

We first consider the~$\Close$ and~$\EqDist$ atoms.
We represent each vertex of the graph that we want to encode by a vertex at an even position along a path.
For any two distinct representatives, there is then a path vertex exactly halfway between them, which we call their midpoint.
We choose the positions so that different unordered pairs have different midpoints, and we give a midpoint a special color exactly when its two representatives should be adjacent in the output graph.
We also make sure that no midpoint between two distinct representatives is itself a representative, since we will need this when we replace the colors by pendant leaves.
An increasing sequence of integers is called a \emph{Sidon sequence} if different unordered pairs of its entries have different sums, where we also allow an entry to be paired with itself.
For background on Sidon sequences and their constructions, we refer to O'Bryant's survey~\cite{OBryant2004}.
For our encoding, the following elementary construction is sufficient, and we include its proof for completeness.

\begin{lemma}
\label{lem:sidon}
For every $m\geq1$, we can compute in polynomial time integers $0<a_1<\cdots<a_m\leq3m^3$ such that all sums $a_i+a_j$, for $1\leq i\leq j\leq m$, are distinct.
\end{lemma}

\begin{proof}
For $m=1$, we choose $a_1:=1$.
For $m\geq2$, we put $B:=2m+1$ and $a_i:=Bi^2+i$.
Suppose that $a_i+a_j=a_k+a_\ell$, where $i\leq j$ and $k\leq\ell$.
By reducing this equality modulo~$B$, we obtain $i+j=k+\ell$, since both sums lie between~$2$ and~$2m$ and are therefore smaller than~$B$.
The original equality now also gives $i^2+j^2=k^2+\ell^2$, and together these two equalities imply $ij=k\ell$.
Two integers are determined, up to their order, by their sum and product, so we have $i=k$ and $j=\ell$.
This argument also covers repeated indices, such as $i=j$.
The sequence is strictly increasing, and we have $a_m=(2m+1)m^2+m\leq3m^3$ for $m\geq2$.
We can compute all its entries in polynomial time.
\end{proof}

We can use the unguarded predicates to test whether two finite distances from the same vertex are equal by the formulas 
\begin{align*}
  \mathsf{Mid}_{\EqDist}(z;x,y)
    &:=\EqDist(z,x,z,y),\\
  \mathsf{Mid}_{\Close}(z;x,y)
    &:=\neg\Close(z,x,y)\wedge\neg\Close(z,y,x).
\end{align*}

When $x,y,z$ lie on one geodesic and $x\neq y$, either formula holds exactly when~$z$ is the midpoint of the $x$-$y$ subpath.
For~$\Close$, we need to keep in mind that the two negated comparisons also hold when both distances are infinite, whereas~$\EqDist$ requires finite distances.
In the construction below, all colored vertices lie in one connected component, so all distances that we compare are finite.

\begin{lemma}
\label{lem:metric-midpoint-encoding}
\label{lem:md-geodesic-coding}
\label{lem:geodesic-cycle-interpretation}
Let~$H$ be a graph on $m\geq1$ vertices.
\begin{enumerate}[label=(\roman*)]
  \item If~$G$ contains a geodesic of length at least~$6m^3$, then we can add two unary predicates to~$G$ so that a fixed existential $\FO[\EqDist]$-interpretation and a fixed existential $\FO[\Close]$-interpretation each produce~$H$.
  \item If~$G$ contains a geodesic cycle of length greater than~$12m^3$, then we can add three unary predicates to~$G$ so that a fixed $\MSO[\EqGC]$-interpretation produces~$H$.
\end{enumerate}
The interpretations do not depend on~$H$ or~$m$.
If we are given the geodesic or the cyclic order of the geodesic cycle, we can compute the unary predicates in polynomial time.
The same constructions work in all extensions of these logics, including their incidence and counting versions.
\end{lemma}

\begin{proof}
We write $V(H)=\set{1,\ldots,m}$ and choose the sequence from~\Cref{lem:sidon}.
For~\textnormal{(i)}, we choose a geodesic $p_0p_1\cdots p_{2a_m}$ and put
\begin{align*}
  U&:=\set{p_{2a_i}:1\leq i\leq m},&
  Q&:=\set{p_{a_i+a_j}:ij\in E(H)}.
\end{align*}

We use the unary input set~$U$ to identify the representatives of the vertices of~$H$.
For each edge~$ij$ of~$H$, we put the midpoint $p_{a_i+a_j}$ between its representatives $p_{2a_i}$ and $p_{2a_j}$ into the unary input set~$Q$.
Every subpath of a geodesic is also a geodesic, so we have $d_G(p_s,p_t)=|s-t|$.
It follows that, for either $R\in\set{\EqDist,\Close}$ and distinct~$i,j$,
\begin{align*}
  \mathsf{Mid}_R(p_c;p_{2a_i},p_{2a_j})
  \quad\Longleftrightarrow\quad c=a_i+a_j.
\end{align*}

We use~$U$ as the domain of the interpretation and make two distinct representatives~$x,y$ adjacent exactly when there is a vertex $z\in Q$ satisfying $\mathsf{Mid}_R(z;x,y)$.
Since different pairs have different midpoints, we obtain exactly the edges of~$H$.
For either choice of~$R$, the domain and adjacency formulas are existential first-order formulas.

For~\textnormal{(ii)}, we choose an arc $p_0p_1\cdots p_{2a_m}$ of the given geodesic cycle~$C$.
As \mbox{$2a_m\leq6m^3<|C|/2$}, this arc is shorter than half of the cycle, so the distance in~$G$ between any two of its vertices is the difference between their positions.
We use the same sets~$U,Q$ as before and add a third unary input set~$Z$ consisting of all vertices of~$C$.
We again use~$U$ as the domain of the interpretation.
For two distinct representatives~$x,y$, we define adjacency by quantifying a vertex set~$C$ that contains exactly the vertices of~$Z$ and asking whether some $z\in Q$ satisfies $\EqGC(C,z,x,z,y)$.
Since~$C$ is the vertex set of the chosen geodesic cycle, the guard holds, and the distance equality is exactly the midpoint test that we used in~\textnormal{(i)}.
We therefore obtain~$H$ again.
This definition only requires vertex-set quantification, and we can use it on incidence structures by restricting the element variables to vertices and the set variables to vertex sets.
Finally, we can compute the unary predicates directly from the listed positions, which gives the claimed polynomial running time.
\end{proof}

We obtain the graph with no vertices by taking~$U=\varnothing$.
It follows that a class with unbounded component diameter is monadically independent for each of the two unguarded first-order extensions, while a class with unbounded geodesic-cycle length is monadically independent for~$\MSO[\EqGC]$.
In particular, the first statement applies to the class of paths, and the second applies to the class of cycles.

\subsection{Dependence for unguarded comparisons}
\label{sec:md-metric-dependence}

Write $D(G)$ for the maximum diameter of a connected component of~$G$.

The above encoding shows that bounded component diameter is necessary for dependence.
It is also the only additional requirement imposed by either unguarded predicate, as bounded distances can simply be expressed by first-order formulas. 

\begin{lemma}
\label{lem:md-bounded-diameter-elimination}
For every fixed~$D$, both $\EqDist$ and~$\Close$ are first-order definable on graphs with~$D(G)\leq D$.
The definitions also apply to incidence structures.
\end{lemma}

\begin{proposition}
\label{prop:md-general-reduction}
Let~$\mathcal L$ contain first-order logic, and suppose that substituting first-order formulas for relation symbols in an $\mathcal L$-formula always gives an $\mathcal L$-formula.
For either $R\in\set{\EqDist,\Close}$, a graph class is monadically $\mathcal L[R]$-dependent if and only if it is monadically $\mathcal L$-dependent and has bounded component diameter.
\end{proposition}

\begin{proof}
If a class is $\mathcal L[R]$-dependent, then it is also $\mathcal L$-dependent. 
The class must also have bounded component diameter by~\Cref{lem:metric-midpoint-encoding}\textnormal{(i)}.
Conversely, suppose that the class is $\mathcal L$-dependent and has a fixed bound on component diameter.
In any $\mathcal L[R]$-interpretation, we replace each metric atom by its first-order definition from~\Cref{lem:md-bounded-diameter-elimination}.
We then obtain an $\mathcal L$-interpretation with the same outputs on all monadic lifts of the class, so these outputs cannot include all finite graphs.
\end{proof}

We now combine this proposition with~\Cref{thm:classical-dependence} to obtain the statements of~\Cref{thm:unguarded-dependence-intro} for arbitrary graph classes, including the counting extensions.

\unguardeddependencetheorem*

For hereditary classes, we can state the incidence characterization in terms of treedepth.
The following lemma follows from the well known fact that a class has bounded treedepth if and only if it does not contain long paths; see~\cite[Theorem~13.3]{NesetrilOssonaDeMendez2012}.

\begin{lemma}
\label{lem:md-hereditary-structure}
A hereditary graph class has bounded treedepth if and only if it has bounded treewidth and bounded component diameter.
\end{lemma}

\begin{corollary}
\label{cor:md-hereditary-treedepth}
On hereditary graph classes, bounded treedepth is equivalent to monadic dependence for each of $\MSOtwo[\EqDist]$, $\MSOtwo[\Close]$, and $\MSOtwo[\EqDist,\Close]$, and for their counting extensions.
\end{corollary}

Note that here we need the assumption that the class is hereditary.
For example, if we add a universal vertex to each path, we obtain a class with treewidth at most two and component diameter at most two, while its treedepth is unbounded.
We also obtain an algorithmic consequence from the first-order definitions of the distance comparisons.
Model checking, additive optimization, and exact counting for $\CMSOtwo[\EqDist,\Close]$ and $\CMSO[\EqDist,\Close]$ are fixed-parameter tractable with parameters $|\varphi|+\tw(G)+D(G)$ and $|\varphi|+\cw(G)+D(G)$, respectively.
Since the substitutions preserve every assignment of the free sets, we can apply the ordinary Courcelle theorems directly.

\subsection{Dependence for guarded comparisons}

For guarded equidistance, we use a long geodesic cycle to carry out the above midpoint construction.
We combine this construction with the classical characterizations in terms of treewidth and cliquewidth, and obtain exactly the structural conditions from our guarded meta-theorems.

\guardeddependencetheorem*

\begin{proof}
We first suppose that~$\mathscr C$ has bounded bag-connected treewidth, so its treewidth and cycle-basis length are also bounded.
We choose $M\geq3$ larger than the cycle-basis length of every graph in~$\mathscr C$.
By~\Cref{lem:guarded-translation}, we can translate every fixed $\CMSOCTW$-interpretation into an ordinary $\CMSOtwo$-interpretation that gives the same output on every monadic lift of a graph in~$\mathscr C$.
Since~$\mathscr C$ has bounded treewidth, the translated interpretation cannot produce all finite graphs, by~\Cref{thm:classical-dependence}.
We therefore obtain dependence for~$\CMSOCTW$ and for each of the smaller logics in the statement.

For the converse, we suppose that~$\mathscr C$ has unbounded bag-connected treewidth.
Its treewidth or its geodesic-cycle length must then be unbounded~\cite{DiestelMuller2018}.
If treewidth is unbounded, then \Cref{thm:classical-dependence}\textnormal{(i)} gives an $\MSOtwo$-interpretation that produces all finite graphs from monadic lifts of the incidence structures of graphs in~$\mathscr C$.
If geodesic-cycle length is unbounded, we obtain an $\MSO[\EqGC]$-interpretation with this property from~\Cref{lem:metric-midpoint-encoding}\textnormal{(ii)}.
In either case, the interpretation belongs to~$\MSOGC$, so we already have independence for the smaller guarded logic.
This proves the incidence characterizations, both with and without counting.

For the vertex-set logics, we recall that $\gcw(G)=\cw(G)+\cbl(G)$.
If~$\mathscr C$ has bounded geodesic cliquewidth, we can again choose one modulus $M$ for the whole class and translate every fixed $\CMSOOCTW$-interpretation into an ordinary $\CMSO$-interpretation.
The outputs remain the same, and bounded cliquewidth prevents the translated interpretation from producing all finite graphs.
We therefore have dependence for both~$\CMSOGC$ and~$\CMSOOCTW$.
If geodesic cliquewidth is unbounded, then cliquewidth or cycle-basis length is unbounded.
In the first case, \Cref{thm:classical-dependence}\textnormal{(ii)} shows that~$\mathscr C$ is monadically $\CMSO$-independent.
In the second case, geodesic-cycle length is also unbounded because $\cbl(G)\leq\gcl(G)$, so we use the cycle construction from~\Cref{lem:metric-midpoint-encoding}\textnormal{(ii)}.
Thus dependence for either of these vertex-set logics is equivalent to bounded geodesic cliquewidth.
\end{proof}

We were originally hoping that the shortest path predicate alone would give us this characterization, however, this is not the case. 
We show this with even cycles, and we also describe the simpler case of complete graphs.

\begin{proposition}
\label{prop:mso1sp-cliques}
\label{prop:mso2sp-even-cycles}
The class of complete graphs is monadically $\CMSOSP$-dependent, and the class of even cycles is monadically $\CMSOtwoSP$-dependent.
Consequently, $\CMSOtwoSP$-dependence does not characterize bounded bag-connected treewidth, and $\CMSOSP$-dependence does not characterize bounded geodesic cliquewidth.
\end{proposition}

\begin{proof}
On a complete graph, $\SP(P,x,y)$ says that $P=\set{x}$ when $x=y$, and that $P=\set{x,y}$ otherwise.
We can express both conditions in first-order logic, treating membership in~$P$ as a unary predicate.
We can therefore replace every shortest-path atom by an ordinary formula, and dependence follows from bounded cliquewidth.

For an even cycle, we add a binary relation~$A$ that pairs each vertex with the vertex halfway around the cycle, which we call its antipode.
We use this additional relation only in the proof of dependence.
For distinct vertices~$x,y$, an $x$-$y$ path set is shortest exactly when it is contained in an $x$-$a$ path set, where~$a$ is the antipode of~$x$.
Indeed, the two $x$-$a$ paths are the two halves of the cycle, so a path starting at~$x$ is contained in one of them exactly when its length is at most half the cycle length.
We can express this condition in ordinary monadic second-order logic using~$A$ and the path-set formula.
For $x=y$, we use the condition $P=\set{x}$ as before.

We next show that these expanded incidence structures have bounded treewidth.
We write the cyclic order as $a_0,\ldots,a_{n-1},b_0,\ldots,b_{n-1}$, where $n\geq2$ and $a_i$ and~$b_i$ are antipodal.
For $0\leq i<n-1$, we take the bags
$\set{a_0,b_0,a_i,b_i,a_{i+1},b_{i+1}}$
in this order, which gives a path decomposition of the cycle together with its antipodal matching.
Each bag has at most six vertices.
For every cycle edge, we attach a leaf bag containing its endpoints and its incidence element to a displayed bag that contains both endpoints.
We obtain a tree decomposition of width at most five for the expanded incidence structure, and adding unary predicates does not change this bound.
After replacing the shortest-path atoms, we therefore have fixed counting monadic interpretations of structures of bounded treewidth.
Their output graphs have bounded cliquewidth, so they cannot include all finite graphs.

An even cycle of length~$r$ has cycle-basis length~$r$, so its geodesic cliquewidth is at least~$r$.
Its bag-connected treewidth is also at least~$r/2$, by~\eqref{eq:cycle-length-bounds}.
Even cycles therefore have unbounded values of both parameters.
Since dependence for~$\CMSOtwoSP$ also implies dependence for the vertex-set fragment~$\CMSOSP$, the even cycles give both claimed counterexamples.
\end{proof}
\subsection{From efficient encodings to hardness}
\label{sec:effective-hardness}

The preceding results show that the required encodings exist, but they do not tell us how to construct them efficiently.
In particular, the graph in which we encode a given graph may be very large.
To prove computational hardness, we therefore also require that we can construct the encoding in polynomial time.

\begin{definition}[Polynomially effective universality]
\label{def:effective-universality}
We say that a graph class~$\mathscr C$ is \emph{polynomially $\mathcal L$-universal} if, given a finite graph~$H$, we can construct in polynomial time a monadic lift of a graph $G_H\in\mathscr C$ such that one fixed $\mathcal L$-interpretation produces~$H$.
The family of unary predicates and the interpretation do not depend on~$H$, and the size of~$G_H$ is polynomial in~$|V(H)|+|E(H)|$.
We use the adjacency or incidence encoding required by~$\mathcal L$.
\end{definition}

We translate a sentence about the output graph back to the input structure by replacing adjacency with the formula that defines it and restricting vertex quantifiers to the interpreted domain.
For a quantified vertex set, we also require that all its elements belong to this domain.
We call this procedure \emph{backward translation}.
The sentences that we translate below use only ordinary graph adjacency, so we do not need to evaluate distances in the output graph.
We use the following standard consequences of backward translation through interpretations.
For the translation theorem, we refer to~\cite[Theorem~7.10]{CourcelleEngelfriet2012}, and for the parameterized hardness results, we refer to~\cite{FlumGrohe2005}.
The existential first-order statement follows by translating the usual sentence expressing the existence of a clique of size~$k$, while the fixed-sentence NP-hardness statement follows by translating the existential monadic second-order sentence for 3-colorability.

\begin{lemma}
\label{thm:effective-hardness}
\label{lem:encoding-hardness-transfer}
Suppose that a fixed interpretation gives polynomially effective encodings of all finite graphs.
Then the following statements hold.
\begin{enumerate}[label=(\roman*)]
  \item If the interpretation is first-order in the input vocabulary, then first-order model checking on the corresponding input structures is $\mathrm{AW}[*]$-hard parameterized by sentence length.
  If its domain and adjacency formulas are existential first-order, then existential first-order model checking is $\mathrm{W}[1]$-hard.
  \item If a monadic logic~$\mathcal L$ contains $\MSO$ and allows backward translation through the interpretation, then model checking~$\mathcal L$ on the corresponding input structures is NP-hard for one fixed sentence.
  If the interpretation is first-order, we can choose this sentence to have only existential set quantifiers followed by a first-order formula.
\end{enumerate}
\end{lemma}

\subsection{Unguarded hardness on paths and caterpillars}

On paths, we can construct the encoding in polynomial time by taking a path of length~$2a_m$ and assigning the two unary predicates to the positions defined above.
To obtain graphs without input colors, we replace these predicates by small groups of pendant leaves.
This construction works for both unguarded predicates, since adding the leaves does not change distances between vertices of the original path.
We call a tree a \emph{caterpillar} if we can obtain it from a path by attaching any number of leaves to the vertices of the path.
The following lemma follows by encoding colors in the obvious way by vertex degrees. 

\begin{lemma}
\label{lem:uncolored-midpoint-encoding}
There are fixed existential $\FO[\EqDist]$- and $\FO[\Close]$-interpretations such that we can encode every finite graph in an uncolored caterpillar of polynomial size, and we can construct this caterpillar in polynomial time.
\end{lemma}

\begin{proof}
We first assume that $V(H)=\set{1,\ldots,m}$ with $m\geq1$, and we choose the sequence from~\Cref{lem:sidon}.
We take a path with positions $0,\ldots,2a_m+1$, which forms the backbone of the caterpillar.
We attach three leaves at each representative position~$2a_i$, and for each edge~$ij$ of~$H$ we attach two leaves at position~$a_i+a_j$.
By the Sidon property, the positions for different edges are distinct, and none of them is a representative position.
All these positions are internal vertices of the backbone, so the representatives have degree five, the vertices representing edges have degree four, and every other vertex has degree at most two.

For each fixed~$d$, we can express degree at least~$d$ by an existential first-order formula asking for~$d$ distinct neighbors.
We therefore use degree at least five as the domain formula.
For either $R\in\set{\EqDist,\Close}$, we declare two distinct representatives~$x,y$ adjacent if there is a vertex~$z$ of degree at least four satisfying $\mathsf{Mid}_R(z;x,y)$.
Every such vertex~$z$ lies on the backbone.
It cannot be another representative, since an equality $2a_h=a_i+a_j$ with $i\neq j$ would contradict the Sidon property, which also includes sums with repeated indices.
Thus a vertex can satisfy the adjacency test exactly when we added the corresponding edge marker, as in~\Cref{lem:metric-midpoint-encoding}.
Both defining formulas are existential first-order, and the caterpillar has $O(m^3)$ vertices.
If~$H$ has no vertices, we use a single-vertex tree, on which the domain formula selects no vertex.
\end{proof}

\begin{theorem}
\label{thm:unguarded-model-checking-hardness}
\label{prop:eqdist-tree-hardness}
\label{thm:close-tree-hardness}
For either $R\in\set{\EqDist,\Close}$, the following statements hold both on paths expanded by two fixed unary predicates and on uncolored caterpillars.
\begin{enumerate}[label=(\roman*)]
  \item Model checking $\FO[R]$ is $\mathrm{AW}[*]$-complete parameterized by sentence length.
  \item Model checking the existential fragment of~$\FO[R]$ is $\mathrm{W}[1]$-complete parameterized by sentence length.
  \item There is a fixed $\mathrm{EMSO}_1[R]$ sentence whose model-checking problem is NP-complete.
\end{enumerate}
The hardness statements therefore also hold on trees, even when we add bag-connected treewidth to the parameter.
They remain true when both unguarded predicates are available.
\end{theorem}

\begin{proof}
We apply~\Cref{lem:encoding-hardness-transfer} to the polynomial path encoding and to~\Cref{lem:uncolored-midpoint-encoding}.

For the upper bounds, we first compute all graph distances and use them to construct the relations $\EqDist$ and $\Close$ explicitly in polynomial time.
We can then apply the standard membership results for first-order and existential first-order model checking over a fixed relational vocabulary~\cite{FlumGrohe2005}.
For a fixed existential-monadic sentence, we guess its sets and check the remaining first-order formula in polynomial time, which proves membership in NP.
The same NP argument applies to $\EqGC$, because, once its set argument is fixed, we can check in polynomial time that it induces a cycle, compare its internal distances with the graph distances, and test the required equality.

All the graphs in these constructions are trees, so their bag-connected treewidth is at most one.
\end{proof}

The same hardness results hold on fixed unary expansions of any class in which we can efficiently construct sufficiently long geodesics in graphs of polynomial size.
More precisely, we require an algorithm that, given~$m$, constructs in time polynomial in~$m$ a graph in the class together with a listed geodesic of length at least~$6m^3$.
We can then apply the encoding to define the two unary predicates and use the same fixed existential interpretations.
Unbounded component diameter alone does not guarantee such an algorithm.

\medskip

A related idea appears in the hardness proof of Grohe and Schweikardt~\cite[Section~4]{GroheSchweikardt2018}.
They encode arbitrary graphs in trees of height three and recover adjacency by comparing the numbers of certain neighbors.
Here we use positions along a path and compare distances to marked vertices, so we also use numerical information inside a simple graph to describe an arbitrary graph.
Their counting logic allows comparisons between cardinalities, whereas the counting atoms of $\CMSO$ used here only test cardinalities modulo fixed integers.
Their construction also gives a conclusion about monadic independence, since the formulas that select the encoded vertices and define their adjacency are fixed and do not depend on the graph being encoded.
We therefore obtain a fixed $\mathrm{FOC}(\{P_=\})$-interpretation that produces every finite graph from an uncolored tree of height at most three, where $P_=$ tests whether two counting terms have the same value.
In our terminology, this means that the class of trees of height at most three is monadically $\mathrm{FOC}(\{P_=\})$-independent, so bounded treedepth does not guarantee monadic dependence for this logic.

\subsection{Effective guarded hardness}

For guarded equidistance, we can use the same approach when we can construct a sufficiently long geodesic cycle.
For the treewidth obstruction, we use the classical grid encoding.
We state the two assumptions separately, since the cycle construction uses only vertex sets, while the grid construction also uses edge sets.

\begin{definition}[Constructing obstructions in polynomial time]
\label{def:polynomial-geodesic-richness}
\label{def:polynomial-grid-richness}
\label{def:efficient-ctw-encoding}
We say that a graph class~$\mathscr C$ is \emph{polynomially geodesically rich} if, given~$m\geq1$, we can construct in polynomial time a graph of polynomial size in~$\mathscr C$ together with a geodesic cycle of length greater than~$12m^3$, whose vertices are listed in cyclic order.
We say that~$\mathscr C$ is \emph{polynomially grid-minor rich} if, given~$r\geq1$, we can construct in polynomial time a graph of polynomial size in~$\mathscr C$ together with a minor model of the $r\times r$ grid, including lists of its branch sets and one edge realizing each grid adjacency.
Here the time and size bounds are polynomial in the numerical values of~$m$ and~$r$, respectively.
If at least one of these two conditions holds, we say that~$\mathscr C$ \emph{admits efficient encodings of its bag-connected-treewidth obstructions}.
\end{definition}

\begin{proposition}
\label{prop:geodesic-rich-universal}
\label{cor:geodesic-rich-hardness}
If~$\mathscr C$ is polynomially geodesically rich, then it is polynomially $\MSO[\EqGC]$-universal using three unary predicates.
On these unary expansions, there is a fixed $\mathrm{EMSO}_1[\EqGC]$ sentence whose model-checking problem is NP-complete.
In particular, this holds on cycles with three unary predicates.
\end{proposition}

\begin{proof}
We use the cycle constructed by the assumed algorithm and apply~\Cref{lem:metric-midpoint-encoding}\textnormal{(ii)} to obtain the encoding in polynomial time.
To obtain a sentence with only existential set quantifiers, we quantify a set~$C$ equal to the input predicate~$Z$ and three sets that partition the representatives into color classes.
For any two distinct representatives~$x,y$ in the same color class, we require that there is no vertex~$z\in Q$ satisfying $\EqGC(C,z,x,z,y)$.
After the initial existential set quantifiers, all these conditions are first-order conditions with guarded metric atoms, and together they express that the encoded graph is 3-colorable.
\end{proof}

We next combine the standard interpretation of finite graphs in grids with the interpretation of minors in incidence structures~\cite[Lemmas~5.1 and~6.7]{BlumensathCourcelle2010}.
Given a grid minor model, we can construct the encoding in polynomial time, using a fixed finite number of unary predicates and no metric predicates.
We therefore obtain the following statement.

\begin{proposition}
\label{prop:grid-rich-universal}
Every polynomially grid-minor rich class is polynomially $\MSOtwo$-universal.
\end{proposition}

We can now combine this proposition with the cycle encoding, since the grid encoding uses only $\MSOtwo$ and the cycle encoding uses only $\MSO[\EqGC]$, both of which are available in $\MSOGC$.
By applying \Cref{lem:encoding-hardness-transfer}\textnormal{(ii)}, we obtain the following consequence.

\begin{corollary}
\label{cor:effective-ctw-hardness}
If a graph class admits efficient encodings of its bag-connected-treewidth obstructions, then model checking~$\MSOGC$ on a fixed finite unary expansion of its incidence structures is NP-hard for one fixed sentence.
It is therefore para-NP-hard parameterized by formula length, also for all counting extensions and for the logic that additionally contains~$\SP$.
\end{corollary}

This proves the colored statement of~\Cref{thm:hardness-intro}.
For the uncolored statement, we use the assumed polynomial encoding that removes colors and translate the fixed hardness sentence through this encoding.
Such an encoding must preserve the meaning of the metric atoms as well as the interpreted graph relations.
The construction with pendant leaves works for the unguarded predicates, but for a general class there is no guarantee that a construction removing colors stays inside that class.

For vertex-set logic, the cycle construction already gives hardness for~$\CMSOGC$.
The same NP-hardness result holds whenever the class is polynomially $\CMSO$-universal, by the general transfer lemma.
The cliquewidth characterization alone does not give the required polynomial-time encoding, and the grid construction uses edge sets, which are not available in vertex-set logic.
Thus the logical characterizations hold without any efficiency assumption, while the general hardness results require the additional algorithmic assumptions stated above.

\subsection{Hardness for distance thresholds}
\label{sec:dist-hardness}

The remaining hardness result concerns~$\Dist$ when its threshold is part of the input.
Here we compare a distance with one input number, rather than comparing two distances as we do with~$\EqDist$ and~$\Close$.
We use a separate construction for this case.

We encode accepting computations of a fixed nondeterministic Turing machine deciding \mbox{\textsc{3-Satisfiability}}.
We first view a computation as a grid whose rows describe successive configurations of the machine, and then write these rows one after another along a path.
Unary input predicates describe the initial configuration and the row boundaries, while existentially quantified vertex sets describe the symbols in the cells.
We use the metric predicate to identify corresponding cells in consecutive rows, which allows us to check that each row follows from the preceding one by a legal transition.
The numerical threshold of~$\Dist$ gives the width of a row.
\pagebreak

\thmdisthardnessintro*

\begin{proof}
Statement~\textnormal{(ii)} follows from~\Cref{thm:unguarded-model-checking-hardness} for~$\Close$ on colored paths, so we only need to prove~\textnormal{(i)}.
We fix a nondeterministic single-tape Turing machine~$M$ with a one-way tape that decides \textsc{3-Satisfiability} in polynomial time.
We choose a polynomial~$p$ that bounds its running time and tape usage on inputs of length~$m$, with $p(m)\geq m+1$.
We use $C:=p(m)+2$ cells for each configuration, including two permanent boundary cells that the head never scans.
We allow a halted configuration to be followed by the same configuration, so that we can describe every computation using exactly $H:=p(m)+1$ rows.

We let $\Gamma$ be the fixed configuration alphabet, which includes the boundary symbols and symbols describing both the state and the scanned tape symbol at the head position.
We choose a fixed finite set $\mathcal W\subseteq\Gamma^6$ of allowed $2$-by-$3$ windows.
If the upper middle cell of a window contains the head, we require its three lower symbols to agree with one legal transition of~$M$.
If the upper triple contains no head, we require its middle symbol to be copied to the lower row.
If the head appears only at an end of the upper triple, this window imposes no further condition, and we forbid upper triples containing more than one head.
For a valid upper configuration, the window centered at its head therefore determines all cells that may change.
The other windows, together with the boundary conditions, require every remaining cell to stay unchanged.
Thus the lower row satisfies all window conditions exactly when it is a legal successor of the upper row.
Since we fix the initial row to be the initial configuration, we can apply this argument successively to all rows.

For an input word~$a$ of length~$m$, we take the path $P_a=v_0v_1\cdots v_{HC-1}$ and use $v_{rC+c}$ for the cell in row~$r$ and column~$c$.
We let the unary predicate $\mathsf{Cell}$ hold at every vertex of this path.
For $i\in\{0,1,2\}$, we let the unary predicate~$Q_i$ hold at exactly those vertices~$v_j$ with $j\equiv i\pmod 3$.
We can then define the directed successor relation $\mathsf{Succ}(x,y)$ by requiring that $x$ and~$y$ are adjacent and belong to consecutive residue classes, in the order $Q_0,Q_1,Q_2,Q_0$.
We use a further predicate $R_{\mathrm{end}}$ for the last cell of each row and define the horizontal successor relation by $\mathsf{Hor}(x,y):=\mathsf{Succ}(x,y)\wedge\neg R_{\mathrm{end}}(x)$.
The remaining unary predicates describe the permanent boundary symbols, the required symbols of the initial row, and the last row.
The number of these predicates depends only on~$M$.

We set $\eta(\delta):=C$ and define the vertical successor relation by
\begin{equation}
  \mathsf{Ver}_\delta(x,y)
  :=
  \Dist(x,y,\delta)
  \wedge
  \exists z\,
  \bigl(\mathsf{Succ}(z,y)\wedge\neg\Dist(x,z,\delta)\bigr).
  \label{eq:directed-distance-jump}
\end{equation}
For $x=v_i$ and $y=v_j$, the witness~$z$ must be $v_{j-1}$, and the distance conditions become $|j-i|\geq C$ and $|j-1-i|<C$.
Since the indices are integers, these two inequalities hold exactly when $j=i+C$.
Consequently, we have
\begin{equation}
  (P_a,\eta)\models\mathsf{Ver}_\delta(v_i,v_j)
  \quad\Longleftrightarrow\quad
  j=i+C,
  \label{eq:jump-correctness}
\end{equation}
so the formula identifies exactly the corresponding cells in consecutive rows.

For each $\gamma\in\Gamma$, we existentially quantify a vertex set~$X_\gamma$.
We require these sets to partition~$\mathsf{Cell}$, agree with the required initial and boundary symbols, satisfy the window conditions between every two consecutive rows, and place an accepting state in the last row.
We can express all these requirements by fixed first-order conditions using the unary predicates, $\mathsf{Hor}$, and~$\mathsf{Ver}_\delta$, with every cell variable restricted to~$\mathsf{Cell}$.
We write $\theta_\delta(\bar X)$ for their conjunction and obtain the fixed existential-monadic template
\begin{equation}
  \Psi(\delta):=\exists (X_\gamma)_{\gamma\in\Gamma}\,\theta_\delta(\bar X).
  \label{eq:fixed-dist-template}
\end{equation}
By the window conditions and~\eqref{eq:jump-correctness}, we have $(P_a,\eta)\models\Psi(\delta)$ exactly when~$M$ has an accepting computation on~$a$.
We can construct the colored path and the binary value $\eta(\delta)=C$ in polynomial time, which proves NP-hardness in~\textnormal{(i)}.
For membership in NP, we guess the sets~$X_\gamma$ and evaluate the fixed first-order formula in polynomial time.
Finally, we take~$\sigma$ to be the union of the fixed unary signatures used in this construction and in the colored-path construction for~\textnormal{(ii)}.
\end{proof}

Thus the logic with an unrestricted numerical threshold is already hard at bag-connected treewidth one, even for one fixed template and one threshold given as part of the input.
The packing-covering and one-cover theorems avoid this obstruction by restricting where numerical distance conditions may occur.

\subsection{W[1]-hardness of packing-covering}
\label{sec:hardness-packing-covering}

For every fixed packing-covering template, our algorithm is fixed-parameter tractable in bag-connected treewidth, although the exponent of the polynomial running time may depend on the template.
We now ask whether we can choose this exponent independently of the template and obtain a running time of the form $F(\abs{\Phi}+\ctw(G))n^{O(1)}$, while still keeping the numerical distance values outside the parameter.
We show that such an algorithm would imply $\mathrm{FPT}=\mathrm{W}[1]$, even when the input graph is a tree.

Recall that $\Sep_d(X,Y)$ requires distance at least~$d$ between every pair of distinct vertices $x\in X$ and $y\in Y$, while $\Cover_r(U,S)$ requires every vertex of the input target set~$U$ to lie within distance~$r$ of some vertex of~$S$.
In each covering atom, we use one unary input predicate for~$U$ and one selected-set variable for~$S$.
The selected sets satisfy a side formula, and the numerical distance conditions appear positively in a conjunction outside this formula.
We prove hardness separately for templates with only separation conditions and for templates with only covering conditions, using an ordinary $\MSO$ side formula in both cases.

\begin{theorem}
\label{thm:packing-covering-hardness}
Deciding whether a packing-covering template has a satisfying assignment is $\mathrm{W}[1]$-hard parameterized by $\abs{\Phi}+\ctw(G)$, even on unweighted trees over a fixed finite unary signature, with an ordinary $\MSO$ side formula and no additional numerical cardinality constraints.
This holds under each of the following restrictions:
\begin{enumerate}[label=(\roman*)]
  \item We use only separation conditions, all with the same input threshold, and require every selected set to be a singleton.
  \item We use only covering conditions, all with the same input radius and the same unary input target set, and each covering atom has one selected-set variable as its second argument.
\end{enumerate}
\end{theorem}

\begin{proof}
We give two parameterized reductions from \textsc{Multicolored Clique}, whose hardness follows from the $\mathrm{W}[1]$-completeness of \textsc{Clique} proved by Downey and Fellows~\cite{DowneyFellows1995W1}.
We are given a graph~$H$ with color classes $V_1,\ldots,V_k$, and we ask whether we can choose a vertex $v_i\in V_i$ for every $i\in[k]$ so that the chosen vertices are pairwise adjacent.
We may assume that $k\geq2$, that every color class is nonempty, and that every edge joins different color classes, since these restrictions do not affect the hardness of the problem.
We put $n_H:=\abs{V(H)}$, assign distinct codes $\alpha(v)\in[n_H]$ to the vertices of~$H$, and set
\[
  M:=n_H+1,\qquad A:=2M+2,\qquad B:=2,
  \qquad D:=A+B+M=3M+4.
\]
We use~$D$ as the only numerical distance value in each reduction, and give it as part of the input without including its value in the template length.

\paragraph{The common tree construction.}
We start with a root~$r$ and attach one choice gadget for every vertex of~$H$ and one choice gadget for every edge of~$H$.
Each gadget has its own center, which we join to~$r$, and we attach all its remaining vertices to this center along paths or pendant edges.
Thus each gadget forms one component of the tree after we delete~$r$.

For a vertex $v\in V_i$ with code $a=\alpha(v)$, we attach two paths to its center, with no other vertices in common, and call their other endpoints $x_v^+$ and~$x_v^-$.
We choose their lengths so that
\[
  d_T(r,x_v^+)=A+a,
  \qquad
  d_T(r,x_v^-)=A+M-a.
\]
The lengths of the paths from the center are therefore $A+a-1$ and $A+M-a-1$, respectively.
For an edge $e=uv$ with $u\in V_i$, $v\in V_j$, and $i<j$, we similarly attach four paths that meet only at the center, with other endpoints $y_{e,i}^+,y_{e,i}^-,y_{e,j}^+,y_{e,j}^-$.
For $h\in\{i,j\}$, we choose their lengths so that
\[
  d_T(r,y_{e,h}^+)=B+b_h,
  \qquad
  d_T(r,y_{e,h}^-)=B+M-b_h,
\]
where $b_i=\alpha(u)$ and $b_j=\alpha(v)$.
All paths from a center to a terminal have positive length, since every code lies between~$1$ and~$M-1$.

We use fixed unary predicates to distinguish the root, the centers, and the six terminal types.
These are the two terminal types in vertex gadgets and the four terminal types in edge gadgets, where we distinguish the smaller and larger color indices of the endpoints.
We also use a unary predicate~$U$ that contains both terminals of every vertex gadget and no other vertices.

To distinguish the $k+\binom{k}{2}$ roles of the gadgets without changing the unary signature, we give these roles distinct numbers from~$1$ to~$k+\binom{k}{2}$ and attach exactly~$h$ pendant marker vertices to every center whose role has number~$h$.
We mark these vertices with one further unary predicate.
For each role, we can then write an ordinary first-order formula saying that the center has exactly the required number of marked neighbors, and the length of this formula depends only on~$k$.
The markers do not change any distance between terminals, and the resulting graph~$T$ is still a tree.

\paragraph{The side formula.}
For every color~$i$, we introduce two selected vertex-set variables $X_i^+,X_i^-$, and for every pair $i<j$, we introduce four variables $Y_{ij,i}^+,Y_{ij,i}^-,Y_{ij,j}^+,Y_{ij,j}^-$.
We write $\bar X$ and $\bar Y$ for these two tuples of variables.
We use an ordinary $\MSO$ formula~$\psi_k(\bar X,\bar Y)$ to require that each set is a singleton containing a terminal of the appropriate type, that the two terminals selected for color~$i$ belong to the same vertex gadget of that role, and that the four terminals selected for the pair~$i,j$ belong to the same edge gadget of that role.
To express the last two requirements, we existentially choose a center of the required role and require the corresponding terminals to lie in the same component of $T-r$ as that center.
We can express this in ordinary $\MSO$, since two vertices lie in the same component of $T-r$ exactly when some connected vertex set contains both vertices and avoids~$r$.
We can therefore construct~$\psi_k$ with length bounded by a function of~$k$, without referring to~$n_H$, the numerical codes, or the distance value~$D$.

An assignment satisfying~$\psi_k$ chooses a vertex~$v_i$ from each color class and an edge~$e_{ij}$ between each pair of color classes, but the endpoints of~$e_{ij}$ need not yet be the selected vertices.
We write $a_i:=\alpha(v_i)$ and let $b_{ij,h}$ be the code of the endpoint of~$e_{ij}$ in~$V_h$, for $h\in\{i,j\}$.
Every path between a terminal of a vertex gadget and a terminal of an edge gadget passes through~$r$, so the relevant distances are
\begin{align}
  d_T(x_{v_h}^+,y_{e_{ij},h}^-)
    &=D+a_h-b_{ij,h},\label{eq:pc-hardness-plus}\\
  d_T(x_{v_h}^-,y_{e_{ij},h}^+)
    &=D-a_h+b_{ij,h}.\label{eq:pc-hardness-minus}
\end{align}
With these two equalities, we can use either separation or covering conditions to check that the selected vertices agree with the endpoints of the selected edges.

\paragraph{Separation conditions alone.}
We first consider the template
\[
  \Phi_k^{\mathrm{sep}}(\bar X,\bar Y;D)
  :=\psi_k(\bar X,\bar Y)\wedge
  \bigwedge_{1\leq i<j\leq k}\ \bigwedge_{h\in\{i,j\}}
  \Bigl(
    \Sep_D(X_h^+,Y_{ij,h}^-)
    \wedge
    \Sep_D(X_h^-,Y_{ij,h}^+)
  \Bigr).
\]
All selected sets are singletons, and the two terminals in each comparison belong to different gadgets, so the convention that $\Sep$ ignores equal vertices has no effect here.
For each pair $i<j$ and each $h\in\{i,j\}$, equations~\eqref{eq:pc-hardness-plus} and~\eqref{eq:pc-hardness-minus} show that the two separation conditions are equivalent to $a_h\geq b_{ij,h}$ and $a_h\leq b_{ij,h}$, respectively.
They therefore hold together exactly when $a_h=b_{ij,h}$.
Since the codes are distinct, these equalities require every selected edge~$e_{ij}$ to have endpoints~$v_i$ and~$v_j$.
Thus every satisfying assignment gives a multicolored clique, and conversely, the vertices and edges of a multicolored clique give a satisfying assignment.
This proves the required equivalence using only positive separation conditions with the common threshold~$D$.

\paragraph{Covering conditions alone.}
For the second reduction, we use the same tree and the same choices, but replace the separation conditions by covering conditions.
We first observe that any two distinct vertices of the input target set~$U$ have distance greater than~$D$.
If the terminals belong to different vertex gadgets, their distance is at least
\[
  2A+2=4M+6>D.
\]
The two terminals in one vertex gadget have distance
\[
  (A+a-1)+(A+M-a-1)=2A+M-2=5M+2>D,
\]
where the last inequality follows from $M\geq3$.

For every pair $i<j$ and every $h\in\{i,j\}$, we introduce two further selected-set variables~$W_{ij,h}^+$ and~$W_{ij,h}^-$, and write $\bar W$ for the tuple of these variables.
We extend~$\psi_k$ by the first-order conditions
\begin{align*}
  \forall z\,\Bigl(z\in W_{ij,h}^+
    &\leftrightarrow
    \bigl((U(z)\wedge z\notin X_h^+)\vee z\in Y_{ij,h}^-\bigr)\Bigr),\\
  \forall z\,\Bigl(z\in W_{ij,h}^-
    &\leftrightarrow
    \bigl((U(z)\wedge z\notin X_h^-)\vee z\in Y_{ij,h}^+\bigr)\Bigr).
\end{align*}
Thus $W_{ij,h}^+$ contains every target except~$x_{v_h}^+$, together with the selected terminal~$y_{e_{ij},h}^-$.
Similarly, $W_{ij,h}^-$ contains every target except~$x_{v_h}^-$, together with the selected terminal~$y_{e_{ij},h}^+$.
These conditions determine~$\bar W$ uniquely once we have chosen~$\bar X$ and~$\bar Y$.
We call the resulting side formula~$\psi'_k(\bar X,\bar Y,\bar W)$ and use the template
\begin{align*}
  \Phi_k^{\mathrm{cov}}(\bar X,\bar Y,\bar W;D)
  := {}&\psi'_k(\bar X,\bar Y,\bar W)\\
  &{}\wedge\bigwedge_{1\leq i<j\leq k}\ \bigwedge_{h\in\{i,j\}}
  \Bigl(
    \Cover_D(U,W_{ij,h}^+)
    \wedge
    \Cover_D(U,W_{ij,h}^-)
  \Bigr).
\end{align*}

To check the first covering condition, note that every target other than~$x_{v_h}^+$ belongs to~$W_{ij,h}^+$ and is therefore covered at distance zero.
The remaining target~$x_{v_h}^+$ has distance greater than~$D$ from every other target, so the only vertex of~$W_{ij,h}^+$ that can cover it is~$y_{e_{ij},h}^-$.
By~\eqref{eq:pc-hardness-plus}, we therefore have
\[
  \Cover_D(U,W_{ij,h}^+)
  \quad\Longleftrightarrow\quad a_h\leq b_{ij,h}.
\]
The same argument, using~\eqref{eq:pc-hardness-minus}, gives
\[
  \Cover_D(U,W_{ij,h}^-)
  \quad\Longleftrightarrow\quad a_h\geq b_{ij,h}.
\]
Both equivalences hold for assignments satisfying~$\psi'_k$.
Thus the two covering conditions again require $a_h=b_{ij,h}$, and the same argument as in the separation reduction shows that the template is satisfiable exactly when~$H$ contains a multicolored clique.
Every covering atom uses the one input predicate~$U$ as its target and one selected-set variable as its covering set, and all atoms use the same radius~$D$.
The conditions defining the covering sets belong entirely to the ordinary side formula, so we do not use unions as arguments of covering atoms.

\paragraph{The parameter and the size of the construction.}
In both reductions, the template depends only on~$k$, while the numerical value of~$D$ is part of the input and is not included in the template length.
We use $O(n_H+\abs{E(H)})$ gadgets, each with $O(n_H)$ vertices on its paths and~$O(k^2)$ marker vertices.
The tree therefore has
\[
  O\bigl((n_H+k^2)(n_H+\abs{E(H)})\bigr)
\]
vertices, and we can construct it in polynomial time, since $k\leq n_H$.
We can also construct each template in time bounded by a computable function of~$k$.

Moreover, $T$ has bag-connected treewidth~$1$.
To see this, we use the root bag~$\{r\}$ and, for each nonroot vertex~$v$, a bag~$\{\operatorname{parent}(v),v\}$, which we attach to the bag of its parent, or to~$\{r\}$ if its parent is~$r$.
Every bag induces a connected subgraph, every edge is contained in a bag, and the bags containing any fixed vertex form a connected subtree.
We therefore have a bag-connected tree decomposition of width~$1$, and $\abs{\Phi}+\ctw(T)$ is bounded by a function of~$k$ alone.
An algorithm with running time $F(\abs{\Phi}+\ctw(T))\abs{V(T)}^c$, for a computable function~$F$ and an absolute constant~$c$, would consequently give an FPT algorithm for \textsc{Multicolored Clique}.
This proves both hardness statements.
\end{proof}

The theorem concerns the uniform problem, in which the template varies, so it does not affect our FPT algorithms for each fixed template.
It also explains why we can obtain a stronger running-time guarantee for one-cover formulas: although all covering conditions in our reduction have the same radius and the same target set, the number of covering-set variables grows with~$k$.
Thus, unless $\mathrm{FPT}=\mathrm{W}[1]$, we cannot extend the
uniform guarantee to arbitrary packing-covering templates while
keeping the numerical distance values outside the parameter.

\section{More algorithmic applications}\label{sec:applications}

In this section, we highlight some more applications of Theorem~\ref{thm:strong-intro}. We consider the algorithmic task \textsc{Convex Hull} defined below, which is central to \emph{graph convexity}~\cite{pelayo2013geodesic,araujo2025introduction}. 
One of the most well studied convexity notions for graphs is the \emph{geodesic convexity}, where a set $S$ of vertices of a graph $G$ is
considered convex if no vertex outside of $S$ lies on a shortest path between
two vertices inside of $S$. A \emph{convex hull} of a set $S$ of vertices is
the smallest convex set containing $S$. A set $S$ is a hull set if its convex hull is the whole vertex set. 
Given a graph $G$ and an integer $k$, the objective of \textsc{Convex Hull} is to decide if $G$ has a hull set of cardinality at most $k$. Kant\'{e}, Marcilon, and Sampaio~\cite{KanteMarcilonSampaio2019} proved that \textsc{Hull Number} is $\mathrm{W}[1]$-hard parameterized by the solution size together with treewidth.
The problem remains NP-hard on chordal graphs~\cite{BessyDouradoPensoRautenbach2018}.
See \cite{Chandran2022ShadowGraphs,Coelho2022ComplementaryPrisms,BessyDouradoPensoRautenbach2018,AraujoMorelSampaioSoaresWeber2016,Dourado2016PkFree,Coelho2015Inapproximability,Dourado2009HullComputation} for more algorithmic results on \textsc{Convex Hull}. Write $\operatorname{Conv}(X)$ for the formula requiring every shortest path between vertices of $X$ to be contained in $X$. Observe that \textsc{Convex Hull} is the additive optimization version of the following
\begin{align*}
	\operatorname{Hull}(S)
	:=\forall X\,
	\bigl(S\subseteq X\wedge\operatorname{Conv}(X)\rightarrow X=V(G)\bigr).
\end{align*}

By Theorem~\ref{thm:strong-intro}, we have the following corollary.

\begin{corollary}
	\textsc{Convex Hull} is FPT w.r.t.\ $\gcw$.
\end{corollary}

We next consider halfspace separation in geodesic convexity. A set $H$
of vertices of a graph $G$ is a \emph{halfspace} if both $H$ and
$V(G)\setminus H$ are convex. Given a connected graph $G$ and two
disjoint nonempty sets $A,B\subseteq V(G)$, the objective of
\textsc{Halfspace Separation}~\cite{Chepoi2024Separation} is to decide if $G$ has a halfspace $H$
such that $A\subseteq H$ and $B\subseteq V(G)\setminus H$.
The problem is NP-complete~\cite{SeiffarthHorvathWrobel2023Halfspace}.
Nair~\cite{Nair2026Halfspace} gave polynomial-time algorithms on weakly
bridged graphs, pseudo-modular graphs, and basis graphs of matroids.
Using the formula $\operatorname{Conv}(X)$ defined above, observe that
\textsc{Halfspace Separation} is expressed by the following formula:
\begin{align*}
\operatorname{Half~Space(A,B)}:= \exists H\,\bigl(A\subseteq H\wedge B\cap H=\emptyset
 \wedge\operatorname{Conv}(H)
 \wedge\operatorname{Conv}(V(G)\setminus H)\bigr).
\end{align*}
Here $A$ and $B$ are represented by unary input predicates.
By Theorem~\ref{thm:strong-intro}, we have the following corollary.
\begin{corollary}\label{cor:halfspace-separation}
 On connected graphs, \textsc{Halfspace Separation} is FPT w.r.t.\ $\gcw$.
 The same parameterization permits computing a minimum-weight separating
 halfspace and exactly counting separating halfspaces.
\end{corollary}

A vertex $s$ strongly resolves $u$ and $v$ if $u$ lies on a shortest $v-s$ path or $v$ lies on a shortest $u-s$ path. Define on connected graphs, metric betweenness by the vertex-set formula $\operatorname{Betw}(x,z,y)
:=\exists P\,(\SP(P,x,y)\wedge z\in P)$.  Thus a strong resolving set is defined by
\begin{align*}
	\forall u\forall v\,
	\bigl(u\neq v\rightarrow\exists s\in S\,
	(\operatorname{Betw}(v,u,s)\vee\operatorname{Betw}(u,v,s))\bigr).
\end{align*}
Given a graph $G$ and an integer $k$, the objective of \textsc{Strong Metric Dimension} is to decide if $G$ has a strong resolving set of cardinality at most $k$. \textsc{Strong Metric Dimension} remains NP-hard even on graphs with constant pathwidth. See \cite{TaleStrongMetric2026,DasGupta2017StrongMetricApproximation,Moravcik2017MetricStrongDimensions,Kratica2012StrongMetricILP,OellermannPetersFransen2007} for more results. By Theorem~\ref{thm:strong-intro}, we have the following corollary.

\begin{corollary}
	\textsc{Strong Metric dimension} is FPT w.r.t.\ $\gcw$.
\end{corollary}

We next consider edge-monitoring problems in graphs. A set $S$ of vertices
of a connected graph~$G$ is a \emph{distance-edge-monitoring set} if,
for every edge $e$, there exist vertices $s\in S$ and $t\in V(G)$ such
that $e$ lies on every shortest $s$--$t$ path. A set $S$ is a
\emph{monitoring edge-geodetic set} if both $s$ and $t$ can be chosen
in $S$. In either case, deleting $e$ increases the distance between
$s$ and $t$. Given a connected graph~$G$ and an integer $k$, the
objective of \textsc{Distance-Edge-Monitoring Set} is to decide if~$G$
has a distance-edge-monitoring set of cardinality at most $k$.
The objective of \textsc{Monitoring Edge-Geodetic Set} is to decide
if $G$ has a monitoring edge-geodetic set of cardinality at most $k$.

Foucaud et al.~\cite{FoucaudKaoKlasingMillerRyan2022} proved that
\textsc{Distance-Edge-Monitoring Set} is NP-complete even on apex graphs
and gave a polynomial-time logarithmic-factor approximation algorithm
for its minimization version. For \textsc{Monitoring Edge-Geodetic Set},
Foucaud et al.~\cite{FoucaudMarcilleSandeepSenTaruni2026} proved
NP-hardness even on $2$-apex $3$-degenerate graphs, and gave a
polynomial-time algorithm on interval graphs and FPT algorithms
parameterized by cliquewidth plus diameter, or by treewidth on
chordal graphs.
See \cite{DevDeyFoucaudKrishnaSulochana2025,FoucaudKaoKlasingMillerRyan2022,FoucaudMarcilleSandeepSenTaruni2026,BiloEtAl2025MonitoringApproximation}
for more results.

For an edge $ab$ and a shortest-path set $P$, the condition
$a,b\in P$ says exactly that $ab$ is an edge of the path, since
shortest paths are induced. Observe that
\textsc{Monitoring Edge-Geodetic Set} is the additive optimization
version of the following formula:
\begin{align*}
 \forall a\forall b\,
 \Bigl(E(a,b)\rightarrow
 \exists s\in S\,\exists t\in S\,\forall P\,
 (\SP(P,s,t)\rightarrow a,b\in P)\Bigr).
\end{align*}
For \textsc{Distance-Edge-Monitoring Set}, replace $\exists t\in S$
by $\exists t$. By Theorem~\ref{thm:strong-intro}, we have the
following corollary.
\begin{corollary}\label{cor:distance-edge-monitoring}
 On connected graphs, \textsc{Distance-Edge-Monitoring Set} and
 \textsc{Monitoring Edge-Geodetic Set} are FPT w.r.t.\ $\gcw$.
\end{corollary}

We next consider visibility problems in graphs. For a set $S$ of vertices
of a graph $G$, two vertices $u,v$ are \emph{$S$-visible} if some shortest
$u$--$v$ path has no internal vertex in $S$. A set $S$ is a
\emph{mutual-visibility set} if every two distinct vertices of $S$ are
$S$-visible. Given a graph $G$ and an integer $k$, the objective of
\textsc{Mutual Visibility} is to decide if $G$ has a mutual-visibility
set of cardinality at least $k$. Di Stefano~\cite{DiStefano2022}
proved that the problem is NP-complete.

Three related notions are obtained by requiring $S$-visibility for
every pair of distinct vertices of $G$, every pair with at least one
endpoint in $S$, or every pair whose endpoints both belong to $S$ or
both belong to $V(G)\setminus S$. These define \emph{total},
\emph{outer}, and \emph{dual mutual-visibility sets}, respectively.
Given a graph $G$ and an integer $k$, the objective of each
corresponding decision problem is to decide if $G$ has such a set
of cardinality at least $k$. All three problems are
NP-complete~\cite{CiceroneDiStefanoDrozekHedzetKlavzarYero2023}.
Cicerone and Di Stefano~\cite{CiceroneDiStefano2025} gave linear-time
algorithms for all four problems on distance-hereditary graphs. See \cite{BiloEtAl2026VisibilityApproximation,CiceroneEtAl2025HypercubeVisibility} for more results. Write $\operatorname{Vis}_S(u,v)$ for the formula
\begin{align*}
	\operatorname{Vis}_S(u,v)
	:=\exists P\,
	\bigl(\SP(P,u,v)\wedge P\cap S\subseteq\{u,v\}\bigr).
\end{align*}
Observe that \textsc{Mutual Visibility} is the additive optimization
version of the following formula:
\begin{align*}
	\forall u\forall v\,
	\bigl(u\in S\wedge v\in S\wedge u\neq v
	\rightarrow\operatorname{Vis}_S(u,v)\bigr).
\end{align*}
The other three variants are expressed by changing the pairs for
which $\operatorname{Vis}_S(u,v)$ is required.
By Theorem~\ref{thm:strong-intro}, we have the following corollary.
\begin{corollary}
	\textsc{Mutual Visibility}, \textsc{Total Mutual Visibility},
	\textsc{Outer Mutual Visibility}, and \textsc{Dual Mutual Visibility}
	are FPT w.r.t.\ $\gcw$.
\end{corollary}

An \emph{edge-geodetic set} $S$ requires every edge to lie on some
shortest path between two vertices of $S$. Its vertex-set formulation
is obtained from the monitoring edge-geodetic formula above by replacing
$\forall P\,(\SP(P,s,t)\rightarrow a,b\in P)$ with
$\exists P\,(\SP(P,s,t)\wedge a,b\in P)$.
By Theorem~\ref{thm:strong-intro}, we also obtain the following corollary.
\begin{corollary}\label{cor:edge-geodetic-monitoring}
 On connected graphs, \textsc{Minimum Edge Geodetic Set}, together with
 its vertex-weighted and exact-counting variants, is FPT w.r.t.\ $\gcw$.
\end{corollary}

We next consider shortest paths that are required to stay far apart.
Given a connected graph~$G$, two terminal pairs $(s_1,t_1)$ and
$(s_2,t_2)$, and an integer $d\geq0$, the objective of
\textsc{Two Distance-Separated Shortest Paths} is to decide if $G$
has a shortest $s_1$--$t_1$ path $P_1$ and a shortest $s_2$--$t_2$
path $P_2$ such that
\begin{align*}
 d_G(P_1,P_2):=\min\{d_G(u,v):u\in V(P_1),\ v\in V(P_2)\}\geq d.
\end{align*}
All distances are measured in $G$, and each path must be shortest
between its own terminals. For $d=1$, this is precisely
\textsc{Two Disjoint Shortest Paths}; for $d=2$, the paths must also
have no edge between them.
Eilam-Tzoreff~\cite{EilamTzoreff1998} gave a polynomial-time algorithm
for \textsc{Two Disjoint Shortest Paths}.
Bentert et al.~\cite{BentertNichterleinRenkenZschoche2021} gave an
$O(nm)$-time algorithm on unweighted graphs, and Akmal, Vassilevska
Williams, and Wein~\cite{AkmalVassilevskaWilliamsWein2024} gave a
randomized linear-time detection algorithm, also allowing positive
edge weights.

For an arbitrary fixed number $k$ of terminal pairs,
Lochet~\cite{Lochet2021DisjointShortestPaths} proved polynomial-time
solvability of \textsc{Disjoint Shortest Paths} on unweighted graphs.
Bentert et al.~\cite{BentertNichterleinRenkenZschoche2021} improved
the running time and proved $\mathrm{W}[1]$-hardness parameterized by
$k$. On planar graphs, Pilipczuk, Stamoulis, and
W{\l}odarczyk~\cite{PilipczukStamoulisWlodarczyk2026} gave an
FPT algorithm parameterized by $k$, even with positive edge weights.
Menashe and Zehavi~\cite{MenasheZehavi2026} gave polynomial-time
algorithms on split graphs and well-partitioned chordal graphs.
See \cite{ChitnisThomasWirth2026,BentertFominGolovach2025,
GottschauKaiserWaldmann2019,KobayashiSako2019,BercziKobayashi2017}
for more results on hardness, approximation, kernelization, and the
weighted and directed variants of \textsc{Disjoint Shortest Paths}.

For ordinary treewidth, a direct consequence of $\MSOtwo$
optimization~\cite{ArnborgLagergrenSeese1991,CourcelleMosbah1993}
is that \textsc{Two Distance-Separated Shortest Paths} is FPT w.r.t
$\tw(G)+d$. Indeed, separation by at least $d$ is first-order
definable by a formula whose size depends on $d$. We can therefore
minimize the total number of edges in two paths joining the prescribed
terminal pairs and satisfying this separation condition. A pair of
shortest paths satisfying the separation condition exists if and only
if the minimum equals
$d_G(s_1,t_1)+d_G(s_2,t_2)$. In particular, every fixed-$d$ case is
FPT w.r.t.\ $\tw(G)$. The cited hardness parameterized by the number
of paths does not establish hardness parameterized by treewidth for
the two-path problem with an input separation threshold.

Observe that, for $d>0$, \textsc{Two Distance-Separated Shortest Paths}
is expressed by the following packing--covering template, where the
set variables $P_1,P_2$ represent the vertex sets of the paths:
\begin{align*}
 \SP(P_1,s_1,t_1)\wedge\SP(P_2,s_2,t_2)
 \wedge P_1\cap P_2=\varnothing\wedge\Sep_d(P_1,P_2).
\end{align*}
The disjointness conjunct is necessary because $\Sep$ omits pairs of
equal vertices. For $d=0$, both the disjointness and separation
conjuncts are omitted. The terminals can be represented by singleton
unary input predicates. By Theorem~\ref{thm:packing-covering}, we have
the following corollary.

\begin{corollary}\label{cor:distance-separated-shortest-paths}
\textsc{Two Distance-Separated Shortest Paths}
 is FPT w.r.t.\ $\ctw(G)$.
\end{corollary}

\bibliographystyle{alpha-doi}
\bibliography{references}

\section*{AI methodology}
\label{sec:ai-methodology}
We used ChatGPT 6.0 Astra to improve the language of the paper and to identify errors in the writing.
The covering fragments in \Cref{sec:dist-meta} were suggested by this tool.
The hardness proofs in \Cref{sec:encoding-hardness}, including
\Cref{sec:hardness-packing-covering}, were fully given by this tool.
We then revised and verified the resulting text and suggestions ourselves.
We take full responsibility for the content of the paper.

\end{document}